\documentclass[10pt, a4paper,american]{article}
\usepackage[pdftex]{graphicx}

\usepackage{amsmath, amssymb, amsthm, amsfonts}
\usepackage{float, color}
\usepackage{ascmac, fancybox}
\usepackage{bm}

\usepackage{bbm}

\usepackage[margin=1in]{geometry}

\makeatletter

\theoremstyle{plain}
\numberwithin{equation}{section}
\theoremstyle{plain}
\newtheorem{theorem}{Theorem}[section]
\theoremstyle{plain}

\theoremstyle{definition}
\newtheorem{definition}[theorem]{Definition}
\theoremstyle{definition}
\newtheorem{example}[theorem]{Example}
\theoremstyle{plain}
\newtheorem{lemma}[theorem]{Lemma}
\theoremstyle{plain}
\newtheorem{corollary}[theorem]{Corollary}
\theoremstyle{definition}
\newtheorem{remark}[theorem]{Remark}
\newtheorem*{conjecture}{Conjecture}
\makeatother

\begin{document}
	
\title{Efficiency of estimators for locally asymptotically normal quantum statistical models}
\author{Akio Fujiwara%
	\thanks{fujiwara@math.sci.osaka-u.ac.jp}\\
	{Department of Mathematics, Osaka University}\\ 
	{Toyonaka, Osaka 560-0043, Japan}\\ \\
	and \\ \\
	Koichi Yamagata%
	\thanks{yamagata@se.kanazawa-u.ac.jp (Current affiliation: Kanazawa University)}\\
	{Principles of Informatics Research Division} \\
		{National Institute of Informatics} \\
	{Hitotsubashi, Chiyoda-ku, Tokyo 101-8430, Japan}
}%
%\date{\today}
\date{}

\maketitle

%%%%% defs %%%%%%

%%%%%%%%%%%%%%%%%%

%%%%%%%%%%

\global\long\def\E{\mathcal{E}}%
\global\long\def\S{\mathcal{S}}%
\global\long\def\R{\mathbb{R}}%
\global\long\def\C{\mathbb{C}}%
\global\long\def\Q{\mathbb{Q}}%
\global\long\def\Z{\mathbb{Z}}%
\global\long\def\D{\mathcal{D}}%
\global\long\def\M{\mathcal{M}}%
\global\long\def\X{\mathcal{X}}%
\global\long\def\F{\mathcal{F}}%
\global\long\def\B{\mathcal{B}}%
\global\long\def\T{\mathcal{T}}%
\global\long\def\para{\xi}%
\global\long\def\P{\mathcal{P}}%
\global\long\def\Para{\Xi}%
\global\long\def\H{\mathcal{H}}%
\global\long\def\Tr{{\rm Tr}\,}%
\global\long\def\L{\mathcal{L}}%
\global\long\def\conv#1{\stackrel{#1}{\rightsquigarrow}}%
\global\long\def\convp#1{\underset{#1}{\to}}%
\global\long\def\cconv#1{\stackrel[\D]{#1}{\rightsquigarrow}}%
\global\long\def\ket#1{\left|#1\right\rangle }%
\global\long\def\bra#1{\left\langle #1\right|}%
\global\long\def\i{\sqrt{-1}}%
%\foreignlanguage{japanese}{}
%\global\long\def\hlc#1#2{\highlightc{#1}{#2}}%
%\foreignlanguage{japanese}{}
\global\long\def\ind{\mathbbm{1}}%
%\foreignlanguage{japanese}{}
\global\long\def\indicate{\ind}%
%\foreignlanguage{japanese}{}
\global\long\def\superL{{\mathcal L}}%
%\foreignlanguage{japanese}{}
\global\long\def\superR{{\mathcal R}}%
%\foreignlanguage{japanese}{}
\global\long\def\braket#1#2{\left\langle #1\mid#2\right\rangle }%
\global\long\def\Ratio{\mathcal{R}}%
\global\long\def\N{\mathbb{N}}%
\global\long\def\A{\mathcal{A}}%

\global\long\def\trans{\,^{t}}

\global\long\def\convd#1{\overset{#1}{\rightsquigarrow}}%
\global\long\def\convp#1{\overset{#1}{\rightarrow}}%
\global\long\def\convbe{\stackrel{\rightsquigarrow}{\rightsquigarrow}}%

\global\long\def\tcirc#1{\textcircled{\footnotesize#1}}%
\global\long\def\tbox#1{\ovalbox{\footnotesize#1}}%
\global\long\def\tdbox#1{\doublebox{\footnotesize#1}}%

\global\long\def\convd#1{\overset{#1}{\rightsquigarrow}}%
%\def\convp#1{\overset{#1}{\rightarrow}}
%\def\convbe{\stackrel[q]{\rightsquigarrow}{\rightsquigarrow}}
%\def\convbe{\stackrel{\rightsquigarrow}{\rightsquigarrow}}
%\global\long\def\span#1{{\rm span}_{\footnotesize#1}}%
%\def\Tr{{\rm Tr}}
%\global\long\def\span#1{{\rm Span}_{#1}}%

%%%%%%%%%%

\begin{abstract}
We herein establish an asymptotic representation theorem for locally asymptotically normal quantum statistical models. 
This theorem enables us to study the asymptotic efficiency of quantum estimators such as quantum regular estimators and quantum minimax estimators, leading to a universal tight lower bound beyond the i.i.d.~assumption. 
This formulation complements the theory of quantum contiguity developed in the previous paper [Fujiwara and Yamagata, {\it Bernoulli} {\bf 26} (2020) 2105-2141],
providing a solid foundation of the theory of weak quantum local asymptotic normality.
\end{abstract}

%----------------------------------------------------------------------------------------------------------------------------------------------------------------
\section{Introduction}\label{sec:Introduction}
%----------------------------------------------------------------------------------------------------------------------------------------------------------------

In classical statistics, a sequence $\{ P_{\theta}^{(n)}:\theta\in\Theta\subset\R^{d}\}$ of statistical models on measurable spaces $(\Omega^{(n)}, \F^{(n)})$ is called {\em locally asymptotically normal} (LAN) at $\theta_{0}\in\Theta$ 
(in the `weak' sense) 
if the log-likelihood ratio $\log \,({dP_\theta^{(n)}}/{dP_{\theta_{0}}^{(n)}})$ 
is expanded in the local parameter $h:=\sqrt{n}(\theta-\theta_0)$ 
as 
\begin{equation} \label{eq:lan}
 \log\frac{dP_{\theta_{0}+h/\sqrt{n}}^{(n)}}{dP_{\theta_{0}}^{(n)}}
 =h^{i}\Delta_{i}^{(n)}-\frac{1}{2}h^{i}h^{j}J_{ij}+o_{P_{\theta_0}}(1).
\end{equation}
Here, $\Delta^{(n)}=(\Delta_1^{(n)},\,\dots,\,\Delta_d^{(n)})$ is a list of $d$-dimensional random vectors on each $(\Omega^{(n)}, \F^{(n)})$ that exhibits 
\[ \Delta^{(n)}\convd 0 N(0,J) \]
with $J$ being a $d\times d$ real symmetric strictly positive matrix, 
the arrow $\convd h$ stands for the convergence in distribution under $P_{\theta_{0}+h/\sqrt{n}}^{(n)}$,  the remainder term $o_{P_{\theta_0}}(1)$ converges in probability to zero under $P_{\theta_0}^{(n)}$, and Einstein's summation convention is used. 

There is an obvious similarity between \eqref{eq:lan} and the log-likelihood ratio of the Gaussian shift model:
\begin{equation*} %\label{eq:normal_ratio}
 \log \frac{dN(Jh,J)}{dN(0,J)} (X_1,\dots,\,X_d)=h^i X_i -\frac{1}{2}h^ih^j J_{ij}.
\end{equation*}
In fact, this similarity is a manifestation of a profound connection between the local parameter model 
$\{ P_{\theta_{0}+h/\sqrt{n}}^{(n)} : h\in\R^{d}\}$ 
and the Gaussian shift model $\{ N(Jh,J) : h\in\R^{d}\}$, 
playing an important role in asymptotic statistics \cite{vaart}. 

In general, a statistical theory comprises two parts: 
one is to prove the existence of a statistic that possesses a certain desired property (direct part), and 
the other is to prove the non-existence of a statistic that exceeds that property (converse part). 
In the problem of asymptotic efficiency, 
the converse part, the impossibility to do asymptotically better than the best which can be done in the limit situation, 
is ensured by the so-called asymptotic representation theorem \cite[Theorem 7.10]{vaart}. 

\begin{theorem}[Asymptotic representation theorem] \label{thm:crep}
Assume that $\{ P_{\theta}^{(n)} : \theta\in\Theta\subset\R^{d}\}$
is LAN at $\theta_{0}\in\Theta$.
Let $T^{(n)}$ be statistics on the local models $P_{\theta_{0}+h/\sqrt{n}}^{(n)}$
that are weakly convergent under every $h\in\R^d$. 
Then, there exists a randomized statistic $T$
on the Gaussian shift model $N(Jh,J)$ such that 
$T^{(n)}\conv hT\label{eq:crep}$
for every $h$.
\end{theorem}

For an accessible proof, see Appendix \ref{suppl:clan}.
Theorem \ref{thm:crep} allows us to deduce in several precise mathematical senses that no estimator can asymptotically do better than what can be achieved in the limiting Gaussian shift model. 
For example, this theorem leads to the convolution theorem, which tells us that \emph{regular} estimators (estimators whose asymptotic behavior in a small neighborhood of $\theta_0$ is more or less stable as the parameter varies) have a limiting distribution which in a very strong sense is more disperse than the optimal limiting distribution which we expect from the limiting statistical problem.
Another option is to use the representation theorem to derive the asymptotic minimax theorem, 
telling us that the worst behavior of an estimator as $\theta$ varies in a shrinking neighborhood of $\theta_0$ cannot improve on what we expect from the limiting problem.
This theorem applies to \emph{all} possible estimators, but only discusses their \emph{worst} behavior in a neighborhood of $\theta_0$. 

Extending the notion of local asymptotic normality to the quantum domain was pioneered by Gu\c{t}\u{a} and Kahn \cite{guta_qubit, guta_qudit}. 
They proved that, given a quantum parametric model 
$\S(\C^{D})=\{ \rho_{\theta}>0 : \theta\in\Theta\subset\R^{D^{2}-1}\} $
comprising the totality of faithful density operators on a $D$-dimensional
Hilbert space and a point $\theta_{0}$ on the parameter space $\Theta$
at which $\rho_{\theta_{0}}$ is nondegenerate 
(i.e., every eigenvalue of $\rho_{\theta_{0}}$ is simple), 
there exist quantum channels $\Gamma^{(n)}$ and $\Lambda^{(n)}$, 
as well as compact sets  $K^{(n)} \subset \R^{D^{2}-1}$ satisfying $K^{(n)}\uparrow \R^{D^2-1}$, such that 
\[
\lim_{n\to\infty}\sup_{h\in K^{(n)}} \left\Vert \sigma_{h}-\Gamma^{(n)}(\rho_{\theta_{0}+h/\sqrt{n}}^{\otimes n}) \right\Vert _{1}=0
\quad\text{and}\quad
\lim_{n\to\infty}\sup_{h\in K^{(n)}} \left \Vert \Lambda^{(n)}(\sigma_{h})-\rho_{\theta_{0}+h/\sqrt{n}}^{\otimes n} \right\Vert _{1}=0,
\]
where $\{ \sigma_{h} :  h\in\R^{D^{2}-1}\}$ is a family of classical/quantum-mixed Gaussian shift model.
Later, Lahiry and Nussbaum \cite{low_rank} extended their formulation to models that comprise non-faithful density operators but have the same rank.
Note that these formulations are not a direct analogue of the weak LAN
defined by (\ref{eq:lan}); in particular, the convergence to a quantum
Gaussian shift model is evaluated not by the convergence in distribution
but by the convergence in trace norm. In this sense, their formulation
could be referred to as a `strong' q-LAN (cf., \cite{GillGuta}). 
Meanwhile, Gu\c{t}\u{a} and Jen\v{c}ov\'a \cite{GutaJencova} also tried to formulate a `weak' q-LAN based on the Connes cocycle derivative, which was sometimes regarded as a proper quantum analogue of the likelihood ratio. However, they did not establish an asymptotic expansion formula which would be directly analogous to \eqref{eq:lan}.

A different approach to a `weak' q-LAN was put forward by the present authors \cite{qlan_first, qcontiguity}.
Given two quantum states $\rho,\sigma\in\S(\H)$ on a finite dimensional Hilbert space $\H$, 
define the {\em square-root likelihood ratio} $\Ratio\left(\sigma\mid\rho\right)$ of $\sigma$ relative to $\rho$ as the positive operator $R$ satisfying the quantum Lebesgue decomposition 
$\sigma=R\rho R+\sigma^{\perp}$,
where the singular part $\sigma^{\perp}$ is the positive operator that satisfies $\Tr\rho\sigma^{\perp}=0$.
The notion of (weak) q-LAN is defined as follows. (See \cite{qcontiguity} for details.)

\begin{definition}[q-LAN]\label{def:qLAN}
A sequence $\S^{(n)}=\{ \rho_{\theta}^{(n)}\mid\theta\in\Theta\subset\R^{d}\}$ of quantum statistical models on Hilbert spaces $\H^{(n)}$ is called {\em quantum locally asymptotically normal} (q-LAN) at $\theta_0\in\Theta$ 
if the square-root likelihood ratio 
$R_{h}^{(n)}=\Ratio(\rho_{\theta_{0}+h/\sqrt{n}}^{(n)}\mid\rho_{\theta_{0}}^{(n)})$
is expanded in $h\in\R^{d}$ as
\begin{equation} \label{eq:R_expand}
 \log\left( R_{h}^{(n)} + o_{L^{2}}(\rho_{\theta_{0}}^{(n)}) \right)^2
 =h^{i}\Delta_{i}^{(n)}-\frac{1}{2}(h^{i}h^{j} J_{ij} )I^{(n)}+o_{D}(h^{i}\Delta_{i}^{(n)},\rho_{\theta_{0}}^{(n)}).
\end{equation}
Here, $\Delta^{(n)}=(\Delta_{1}^{(n)},\dots,\Delta_{d}^{(n)})$ is
a list of observables on each $\H^{(n)}$ that exhibits
\begin{equation*} %\label{eq:qLAN_convd}
\Delta^{(n)}\conv{\rho_{\theta_{0}}^{(n)}}N(0,J)
\end{equation*}
with $J$ being a $d\times d$ complex nonnegative matrix%
\footnote{
For a complex covariance matrix $J$, the state $N(0,J)$ is regarded as a hybrid classical/quantum Gaussian state. 
Specifically, $N(0, J)$ is classical if and only if ${\rm Im}\,J=0$, and is purely quantum if and only if ${\rm Im}\,J$ is invertible. 
For more information, see 
%Section B of Supplementary Material. 
Appendix \ref{app:degenerate_CCR}.
}
satisfying ${\rm Re}\, J>0$, 
the arrow $\conv{\rho_{\theta_{0}}^{(n)}}$ stands for the quantum convergence in distribution under $\rho_{\theta_0}^{(n)}$ defined by the convergence of the quasi-characteristic function,
and $o_{L^{2}}(\rho_{\theta_{0}}^{(n)})$ and $o_{D}(h^{i}\Delta_{i}^{(n)},\rho_{\theta_{0}}^{(n)})$ are infinitesimal remainder terms in $L^2$ and in distribution, respectively.
\end{definition}

One may recognize a clear parallelism between the classical definition \eqref{eq:lan} and the quantum one \eqref{eq:R_expand}. 
In fact, the theory of weak q-LAN based on \eqref{eq:R_expand} has been successfully applied to quantum statistical models satisfying mild regularity conditions,
culminating in the derivation of (an abstract version of) the quantum Le Cam third lemma \cite{qcontiguity}.  
However, this theory is not yet fully satisfactory because it lacks tools to cope with the converse problems, that is, to prove the impossibility of doing asymptotically better than the best which can be done on the limiting model specified by the quantum Le Cam third lemma.
For example, we do not know conditions to get rid of asymptotically superefficient estimators that break the Holevo bound in an i.i.d.~model. 

In the context of these circumstances, we aim to establish a noncommutative counterpart of Theorem \ref{thm:crep} that enables us to study the converse part in quantum asymptotic statistics.
The paper is organized as follows. 
In Section \ref{sec:main_results}, we summarize the main results, including the asymptotic quantum representation theorem for q-LAN models, and a universal tight bound for efficiency that generalizes the Holevo bound to generic (not necessarily i.i.d.) models.  
This section will also serve as an overview of the paper.
In Section \ref{sec:preliminary}, we provide some mathematical tools and a number of lemmas that are used in the proof of the representation theorem, and the proof of the theorem itself is carried over to the succeeding Section \ref{sec:proof_qrep}. 
In Section \ref{sec:applications}, we apply the representation theorem 
to the analysis of efficiency for sequences of quantum estimators such as the quantum Hodges estimator, 
quantum regular estimators, quantum minimax estimators, and the quantum James-Stein estimator.
Section \ref{sec:conclusions} is devoted to concluding remarks.

Some additional materials are provided in Appendix, including 
a proof of Theorem \ref{thm:crep} (Appendix \ref{suppl:clan}),
a comprehensible account of degenerate canonical commutation relation (CCR) and hybrid classical/quantum Gaussian states (Appendix \ref{app:degenerate_CCR}),
a detailed account of the notion of $D$-extendibility 
(Appendix \ref{app:D-extension}),
and proofs of lemmas and theorems presented in Sections \ref{sec:preliminary} and \ref{sec:applications} (Appendix \ref{app:ProofLemmas} and \ref{app:ProofApplications}, respectively). 

%----------------------------------------------------------------------------------------------------------------------------------------------------------------
\section{Main results}\label{sec:main_results}
%----------------------------------------------------------------------------------------------------------------------------------------------------------------

Assume that a sequence $\S^{(n)}=\{ \rho_{\theta}^{(n)} : \theta\in\Theta\subset\R^{d}\}$ of quantum statistical models is q-LAN at $\theta_0\in\Theta$ as in Definition \ref{def:qLAN}.
In view of the classical representation theorem (Theorem \ref{thm:crep}), one may envisage the following

\begin{conjecture}
Let $M^{(n)}=\{M^{(n)} (B)\}_{B\in\B(\R^s)}$ be a sequence of POVMs over the Borel $\sigma$-algebra $\B(\R^s)$
of $\R^s$ such that the corresponding sequence of classical probability measures 
\[ \L_{h}^{(n)}:=\Tr \rho_{\theta_{0}+h/\sqrt{n}}^{(n)} M^{(n)} \]
is weakly convergent to some probability measure $\L_{h}$ for every $h$. 
Then there would exist a POVM $M^{(\infty)}=\{M(B)\}_{B\in\B(\R^s)}$ on ${\rm CCR}({\rm Im}\,J)$ such that 
\[ \phi_h(M^{(\infty)}(B))=\L_{h}(B) \]
for every $h$, where $\phi_h \sim N(({\rm Re}\,J)\,h,J)$.
\end{conjecture}

However, such a naive guess fails, as the following example shows. 

\begin{example}\label{eg:noLimitPOVM}
Let us consider the following one-dimensional pure state model:
\[
\rho_{\theta}=\frac{2}{e^{\theta}+e^{-\theta}}\,e^{\frac{\theta}{2}\sigma_x}\, \rho_0 \, e^{\frac{\theta}{2}\sigma_x},
\qquad (\theta\in\R)
\]
where
\[
 \sigma_x=\begin{pmatrix}0 & 1\\ 1 & 0\end{pmatrix} 
 \quad \mbox{and} \quad 
 \rho_0=\begin{pmatrix}1 & 0\\ 0 & 0\end{pmatrix}.
\]
This model has an SLD $\sigma_x$ at $\theta=0$. 
Let
\[
\Delta^{(n)}:=\frac{1}{\sqrt{n}}\sum_{k=1}^{n}I^{\otimes(k-1)}\otimes\sigma_x\otimes I^{\otimes(n-k)}. 
\]
Then it is shown (cf., \cite[Section 3.2]{qlan_first}, \cite[Section 7.3]{qcontiguity}) that $\rho_\theta^{\otimes n}$ is q-LAN at $\theta=0$, and 
\[
\Delta^{(n)}\conv{\rho_{h/\sqrt{n}}^{\otimes n}}N(h,1).
\]
However, there is a sequence of POVMs that does not have a limiting POVM on the (classical) Gaussian shift model $N(h,1)$. 

Let $M^{(n)}$ be a binary-valued POVM on $\rho_{h/\sqrt{n}}^{\otimes n}$ defined by
\[
M^{(n)}(0)=\rho_0^{\otimes n},\quad
M^{(n)}(1)=I^{(n)}-\rho_0^{\otimes n}.
\]
Then
\[
\lim_{n\to\infty}\Tr\rho_{h/\sqrt{n}}^{\otimes n}M^{(n)}(0)
=\lim_{n\to\infty}\left(\Tr\rho_{h/\sqrt{n}}\;\rho_0\right)^{n}=e^{-\frac{1}{4}h^{2}}, 
\]
and thus the sequence of POVMs has a limiting distribution
\[
\L_{h}(0)=e^{-\frac{1}{4}h^{2}},\quad
\L_{h}(1)=1-e^{-\frac{1}{4}h^{2}}
\]
for each $h\in\R$. 

Now, suppose that this distribution is realized by a binary-valued POVM $M^{(\infty)}$ that is independent of $h$.
Since the limiting Gaussian shift model $N(h,1)$ is classical, $M^{(\infty)}$ is represented by a measurable function $m(x)$ on $\R$ such that 
\[
M^{(\infty)}(0)=m(x),\quad
M^{(\infty)}(1)=1-m(x).
\]
Specifically, $0 \le m(x) \le 1$ for all $x\in\R$, and
\begin{equation}\label{eq:counter_ex}
 e^{-\frac{1}{4}h^{2}}=\int_{-\infty}^\infty m(x) p_h(x) dx
\end{equation}
for all $h\in\R$, where $p_h(x)=\frac{1}{\sqrt{2\pi} }e^{-(x-h)^2/2}$ is the density function of $N(h,1)$. 
However, \eqref{eq:counter_ex} has the solution
\[
 m(x)=\sqrt{2}\,e^{-\frac{1}{2}x^{2}},\quad(\mbox{a.e.})
\]
which does not fulfill the requirement that $0 \le m(x) \le 1$. This is a contradiction.
\end{example}

Example \ref{eg:noLimitPOVM} demonstrates that we need some additional condition to establish an asymptotic representation theorem in the quantum domain.
In fact, the following condition will prove to be sufficient.

\begin{definition}[$D$-extendibility]\label{def:propertyD}
Given a sequence
$\S^{(n)}=\{ \rho_{\theta}^{(n)} : \theta\in\Theta\subset\R^{d}\}$
of quantum statistical models on $\H^{(n)}$, a sequence $X^{(n)}=(X_{1}^{(n)},\dots,X_{r}^{(n)})$ of observables on $\H^{(n)}$ is called {\em asymptotically $D$-invariant} at $\theta_0\in\Theta$ if it fulfills the following requirements: 
\begin{equation}\label{eq:X_cond2}
X^{(n)}\conv{\rho_{\theta_{0}}^{(n)}}N(0,\Sigma)
\end{equation}
for some $r\times r$ nonnegative matrix $\Sigma$ with ${\rm Re}\,\Sigma>0$, 
and
\begin{equation}\label{eq:X_cond3}
\lim_{n\to\infty}
\Tr\sqrt{\rho_{\theta_{0}}^{(n)}}e^{\sqrt{-1}\xi^{i}X_{i}^{(n)}}\sqrt{\rho_{\theta_{0}}^{(n)}}e^{\sqrt{-1}\eta^{i}X_{i}^{(n)}} 
=e^{-\frac{1}{2}
 \begin{pmatrix} \xi \\ \eta \end{pmatrix}^\top
 \begin{pmatrix} \Sigma & \Sigma\#\Sigma^{\top} \\ \Sigma\#\Sigma^{\top} & \Sigma^{\top} \end{pmatrix}
 \begin{pmatrix} \xi \\ \eta \end{pmatrix}}
\end{equation}
for all $\xi,\eta\in\R^{r}$, where $\#$ stands for the operator geometric mean \cite{{Ando},{kubo}}. 

A sequence $\S^{(n)}=\{ \rho_{\theta}^{(n)} : \theta\in\Theta\subset\R^{d}\}$ of quantum statistical models that is q-LAN at $\theta_0\in\Theta$ is called {\em $D$-extendible} at $\theta_0$
if there exists a sequence $X^{(n)}=(X_{i}^{(n)})_{1\le i\le r}$ of observables as well as an $r\times d$ real matrix $F$ such that 
\begin{equation}\label{eq:X_cond1}
\Delta_{k}^{(n)}=F_{k}^{i}X_{i}^{(n)}\qquad (1\le k \le d,\; n\in\N)
\end{equation}
and $X^{(n)}$ is asymptotically $D$-invariant at $\theta_0\in\Theta$. 
Such a sequence $X^{(n)}$ is called a {\em $D$-extension} of $\Delta^{(n)}$. 
\end{definition}

\begin{remark}\label{rem:DinvExt}
One may have the impression that the condition \eqref{eq:X_cond3} is strange and intractable; but in reality it is not too restrictive in applications. 
For example, 
let $\S=\{ \rho_{\theta} : \theta\in\Theta\subset\R^{d}\}$ be a quantum statistical model on a finite dimensional Hilbert space $\H$. 
Then, under some mild regularity conditions, 
the sequence $\S^{(n)}:=\{ \rho_{\theta}^{\otimes n} : \theta\in\Theta\subset\R^{d}\}$ of i.i.d.~models on $\H^{\otimes n}$ is not only q-LAN at a given $\theta_0\in\Theta$ \cite[Theorem 7.6]{qcontiguity}, but also $D$-extendible at $\theta_0$. 
For a proof, see Appendix \ref{app:D-extension}, where the idea behind the term `asymptotic $D$-invariance' is also clarified 
and a proper perspective on the model in Example \ref{eg:noLimitPOVM} is demonstrated.
There are of course models $\S^{(n)}$ that are non-i.i.d.~but are, nevertheless, q-LAN and $D$-extendible;
a simple example is provided in Appendix \ref{app:D-extension}.
\end{remark}

With this additional requirement of $D$-extendibility, we can prove the following

\begin{theorem}[Asymptotic quantum representation theorem]\label{thm:qRep}
Assume that a sequence 
$\S^{(n)}=\{ \rho_{\theta}^{(n)} : \theta\in\Theta\subset\R^{d}\}$
of quantum statistical models is q-LAN and $D$-extendible at $\theta_0\in\Theta$.  
Let $M^{(n)}=\{M^{(n)} (B)\}_{B\in\B(\R^s)}$ be a sequence of POVMs over $\R^s$ such that the corresponding sequence of classical probability measures 
\[ \L_{h}^{(n)}:=\Tr \rho_{\theta_{0}+h/\sqrt{n}}^{(n)} M^{(n)} \]
is weakly convergent to some probability measure $\L_{h}$ for every $h$. 
Then there exists a POVM $M^{(\infty)}=\{M(B)\}_{B\in\B(\R^s)}$ on ${\rm CCR}({\rm Im}\,\Sigma)$ such that 
\[ \phi_h(M^{(\infty)}(B))=\L_{h}(B) \]
for every $h$, where $\phi_h \sim N(({\rm Re}\,\tau)\,h,\Sigma)$ with $\tau=\Sigma F$.
\end{theorem}

Theorem \ref{thm:qRep} allows us to convert a statistical problem for the local parameter model $\{\rho^{(n)}_{\theta_0+h/\sqrt{n}} : h \in \R^d \}$ into another one for the limiting quantum Gaussian shift model 
$\{ N( ({\rm Re}\, \tau) h, \Sigma) : h \in \R^d \}$. 
It is thus natural to expect that the Holevo bound%
\footnote{
The Holevo bound $c_G^{(H)}$ for a generic quantum statistical model $\{\rho_\theta: \theta\in\Theta\subset\R^d\}$ on a Hilbert space $\H$ is given by the minimum of $\Tr GZ(B)+\Tr\left| \sqrt{G}\,{\rm Im}\,Z(B) \sqrt{G}\right|$ over all Hermitian operators $B=(B_1,\dots, B_d)$ on $\H$ satisfying the local unbiasedness condition ${\rm Re}\Tr \rho_{\theta} L_i B_j=\delta_{ij}$, where $L_i$ is the $i$th SLD and $Z(B)$ is the $d\times d$ matrix whose $(i,j)$th entry is $Z_{ij}(B):=\Tr\rho_{\theta} B_j B_i$. 
The reduced expression \eqref{eq:rep_bound} for the quantum Gaussian shift model $\{ N( ({\rm Re}\, \tau) h, \Sigma) : h \in \R^d \}$ is derived in \cite[Appendix B]{qlan_first}.
}
for the limiting model $N( ({\rm Re}\, \tau) h, \Sigma)$, given a weight matrix $G>0$, i.e., 
\begin{align}\label{eq:rep_bound}
c_{G}^{(rep)} & :=\min_{K}\left\{ \Tr GZ+\Tr\left|\sqrt{G}\,{\rm Im}\,Z\sqrt{G}\right|:Z=K^{\top}\Sigma K,\right.\\
 & \qquad\qquad \left.K\text{ is an \ensuremath{r\times d} real matrix satisfying }K^{\top}\left({\rm Re}\,\tau\right)=I\right\}, \nonumber
\end{align}
will be of fundamental importance in quantum asymptotics. 
Note that the $D$-extension in Definition \ref{def:propertyD} is not unique; however, it can be shown that the bound $c_{G}^{(rep)}$ is independent of the choice of a $D$-extension (Corollary \ref{cor:rep_bound}).
In what follows, we shall call this universal bound the asymptotic representation bound.

Indeed, the bound $c_{G}^{(rep)}$ plays a crucial role in asymptotic quantum statistics. 
For example, it gives the ultimate limit of estimation for regular estimators (Theorems \ref{thm:regular} and \ref{thm:achieve}) and minimax estimators (Theorem \ref{thm:minimax_local}). 
Moreover, the bound $c_{G}^{(rep)}$ for an i.i.d.~model $\S^{(n)}=\{\rho_\theta^{\otimes n}\}$ is identical to the standard Holevo bound $c_{G}^{(H)}$ for the base model $\rho_\theta$ (Theorem \ref{thm:achieve} and Appendix \ref{app:D-extension}).
Thus, the asymptotic representation bound $c_{G}^{(rep)}$ can be regarded as a 
fully generalized version of the Holevo bound that is also applicable to non-i.i.d.~models.

Incidentally, as one can see from the proof, Theorem \ref{thm:qRep} is valid even if the scaling factors $\sqrt{n}$ in Definition \ref{def:qLAN} and Theorem \ref{thm:qRep} are both replaced with an arbitrary monotone increasing positive sequence $r_n \uparrow\infty$. 
Also, one can replace the domain $\R^d$ of the local parameter $h$ to an arbitrary subset of $\R^d$. 
Classical analogues of these generalizations are found, for example, in \cite[Definition 7.14, Theorem 9.4]{vaart}.  

%----------------------------------------------------------------------------------------------------------------------------------------------------------------
\section{Preliminaries} \label{sec:preliminary}
%----------------------------------------------------------------------------------------------------------------------------------------------------------------

In this section, we devise some mathematical tools and prepare a number of lemmas toward the proof of Theorem \ref{thm:qRep}.
First we give a condition for a quantum Gaussian state to be pure.
We then introduce a new way of representing bounded operators on a ${\rm CCR}(S)$ which is analogous to the Husimi representation \cite{Husimi}.
We further extend quantum L\'evy-Cram\'er continuity theorem \cite{qLevyCramer} and quantum Le Cam third lemma \cite{qcontiguity} so that they are directly applicable to the proof of Theorem \ref{thm:qRep}. 
All the proofs of the lemmas and corollaries presented in this section are deferred to 
Appendix \ref{app:ProofLemmas}.
For the definition of von Neumann algebra ${\rm CCR}(S)$ with possibly degenerate $S$ and quantum Gaussian states on it, see 
Appendix \ref{app:degenerate_CCR}.

\subsection{Condition for a quantum Gaussian state to be pure}

\begin{lemma}[Minimum uncertainty]\label{lem:pure_gaussian}
Let $J=V+\i S$ be a $d\times d$ nonnegative matrix in which $S={\rm Im}\,J$ is invertible. 
Then the quantum Gaussian state $N(0,J)$ on the von Neumann algebra ${\rm CCR}(S)$ is pure if and only if $\det V=\det S$.
\end{lemma}

\begin{proof} 
%See Section D of Supplementary Material \cite{qrep_supp}. 
See Appendix \ref{app:ProofLemmas}.
\end{proof}

\begin{corollary}\label{lem:pure_double}
Let $J=V+\i S$ be a $d\times d$ nonnegative matrix in which both $V={\rm Re}\,J$ and $S={\rm Im}\,J$ are invertible. Then the quantum Gaussian state 
\[
N\left(\begin{pmatrix}0\\
0
\end{pmatrix},\begin{pmatrix}J & J\#J^{\top}\\
J\#J^{\top} & J^{\top}
\end{pmatrix}\right)
\]
is pure.
\end{corollary}

\begin{proof} 
%See Section D of Supplementary Material \cite{qrep_supp}. 
See Appendix \ref{app:ProofLemmas}.
\end{proof}

\subsection{Sandwiched coherent state representation of operators on a CCR algebra}

Let $\H$ be a Hilbert space that represents the von Neumann algebra ${\rm CCR}(S)$, where $S$ is a skew-symmetric real $d\times d$ matrix that is not necessarily invertible, and let $\{X_i\}_{1\le i\le d}$ be the canonical observables of ${\rm CCR}(S)$. 
Fix a cyclic%
\footnote{
A vector $\psi\in\H$ is called {\em cyclic} for a linear subspace $\mathcal{A}$ of $B(\H)$ if the linear space $\mathcal{A}\psi:=\{A\psi: A\in\mathcal{A}\}$ is norm-dense in $\H$.
}
unit vector $\psi\in\H$ for ${\rm CCR}(S)$, and let 
\[ \psi(\xi):=e^{\i\xi^{i}X_{i}}\psi,\quad (\xi\in\R^{d}). \]
Associated with a bounded operator $A\in B(\H)$ is a continuous function $\varphi_A: \R^d\times\R^d\to \C$ defined by
\[ \varphi_{A}(\xi;\eta): =\langle \psi(\xi), A \psi(\eta) \rangle, \quad (\xi, \eta\in\R^{d}). \]
We shall call $\varphi_A$ the {\em sandwiched coherent state representation} of a bounded operator $A$. 

We are interested in the converse problem: when does a function $\varphi: \R^d\times\R^d\to \C$ uniquely determine an operator $A\in B(\H)$ satisfying $\varphi(\xi;\eta)=\langle \psi(\xi), A \psi(\eta) \rangle$? 
Let $D$ be a dense subset of $\R^d$. 
A function $\varphi: D \times D\to \C$ is called {\em positive semidefinite} if, for all $r\in\N$ and $\{ \xi^{(i)}\} _{1\le i \le r} \subset D$, the $r\times r$ matrix whose $(i,j)$th entry is $\varphi(\xi^{(i)};\xi^{(j)})$ is positive semidefinite, i.e., 
\[ \left[\varphi(\xi^{(i)};\xi^{(j)})\right]_{1\leq i,j \leq r} \geq 0. \]
In this case we denote $\varphi \succ 0$. 
Further, for two functions $\varphi_1$ and $\varphi_2$, we denote $\varphi_1 \succ \varphi_2$ if $\varphi_{1}-\varphi_{2}\succ 0$.

\begin{lemma}\label{thm:bochoner_dshift}
Suppose that $\varphi:D\times D\to\C$ satisfies $0 \prec \varphi \prec  \varphi_{I}$. 
Then there exists a unique operator $A$ satisfying $0\le A\le I$ and $\varphi=\varphi_{A}$. 
Consequently, $\varphi$ is continuously extended to the totality of $\R^{d}\times\R^{d}$.
\end{lemma}

\begin{proof} 
%See Section D of Supplementary Material \cite{qrep_supp}. 
See Appendix \ref{app:ProofLemmas}.
\end{proof}

Lemma \ref{thm:bochoner_dshift} establishes a one-to-one correspondence between bounded operators $A$ satisfying $0\le A\le I$ and functions $\varphi$ satisfying $0\prec \varphi \prec \varphi_I$. 
In what follows, the operator $A$ that is recovered from the function $\varphi$ is denoted by $V(\varphi)$. 

Now, let $S=O_{c}\oplus S_{q}\oplus S_{a}$, where $O_{c}$ is a $d_{c}\times d_{c}$ zero matrix, $S_{q}$ a $d_{q}\times d_{q}$ skew-symmetric real invertible matrix, and $S_{a}$ a $d_{a}\times d_{a}$ skew-symmetric real invertible matrix%
\footnote{
The subscripts $c$, $q$, and $a$ stand for the classical, quantum, and ancillary systems, respectively.
}. 
Then, ${\rm CCR}(S)={\rm CCR}(O_{c})\otimes {\rm CCR}(S_{q})\otimes {\rm CCR}(S_{a})$, and the canonical observables are 
\[
 \left\{ \hat X_{c,i}:=X_{c,i}\otimes I_{q}\otimes I_{a}\right\} _{i} \cup 
 \left\{ \hat X_{q,j}:=I_{c}\otimes X_{q,j}\otimes I_{a}\right\} _{j} \cup
 \left\{ \hat X_{a,k}:=I_{c}\otimes I_{q}\otimes X_{a,k}\right\} _{k},
\]
where $\{X_{c,i}\}_{i}$, $\{X_{q,j}\}_{j}$, and $\{X_{a,k}\}_{k}$ are the canonical observables of ${\rm CCR}(O_{c})$, ${\rm CCR}(S_{q})$, and ${\rm CCR}(S_{a})$, respectively.
In the Schr\"odinger representation, the algebra ${\rm CCR}(S)$ is represented on the Hilbert space $\H:=\H_c\otimes \H_q \otimes \H_a$, where $\H_c:=L^2(\R^{d_c})$, $\H_q:=L^2(\R^{d_{q}/2})$, and $\H_a:=L^2(\R^{d_{a}/2})$, and
\begin{align*}
{\rm CCR}(O_{c}) & =\overline{{\rm Span}}^{\rm SOT}\left\{ e^{\i \xi_{c}^{i} X_{c,i} }\right\} _{\xi_{c}\in\R^{d_{c}}}=L^{\infty}(\R^{d_{c}}), \\
{\rm CCR}(S_{q}) & =\overline{{\rm Span}}^{\rm SOT}\left\{ e^{\i \xi_{q}^{j} X_{q,j} }\right\} _{\xi_{q}\in\R^{d_{q}}}=B(\H_q),\\
{\rm CCR}(S_{a}) & =\overline{{\rm Span}}^{\rm SOT}\left\{ e^{\i \xi_{a}^{k} X_{a,k} }\right\} _{\xi_{a}\in\R^{d_{a}}}=B(\H_a),
\end{align*}
where $\overline{{\rm Span}}^{\rm SOT}$ denotes the closure of the linear span with respect to the strong operator topology (SOT). 
Since $L^{\infty}(\R^{d_{c}})$ is a maximal abelian subalgebra%
\footnote{A von Neumann subalgebra $M$ of $B(\H)$ that satisfies $M'=M$ is called a {\em maximal abelian subalgebra} (MASA). 
The name comes from the fact that if $N$ is an abelian von Neumann algebra such that $M\subset N \subset B(\H)$, then $M=N$. In fact, since $M\subset N$, we have $(M\subset N\subset)\, N'\subset M'=M$, so that $M=N$. 
}
of $B(\H_c)$,
the celebrated commutant theorem \cite{{KadisonRingrose_v2},{Hiai:2021}} yields  
\begin{equation}\label{eq:commutant}
\left({\rm CCR}(O_{c})\otimes {\rm CCR}(S_{q})\otimes I_{a}\right)'={\rm CCR}(O_{c})\otimes I_{q}\otimes {\rm CCR}(S_{a}).
\end{equation}
In this identity, $I_q$ and $I_a$ symbolically represent the centers of ${\rm CCR}(S_q)$ and ${\rm CCR}(S_a)$, respectively. 

Let $\psi\in\H$ be a cyclic unit vector for ${\rm CCR}(S)$. 
Then, the sandwiched coherent state representation of $A\in {\rm CCR}(S)$ is given by
\begin{align*}
\varphi_{A}(\xi_{c},\xi_{q},\xi_{a};\eta_{c},\eta_{q},\eta_{a})  
 & =\langle\psi(\xi_{c},\xi_{q},\xi_{a}), A \psi(\eta_{c},\eta_{q},\eta_{a}) \rangle, 
\end{align*}
where $\xi_{c},\eta_{c}\in\R^{d_{c}}$, $\xi_{q},\eta_{q}\in\R^{d_{q}}$, $\xi_{a},\eta_{a}\in\R^{d_{a}}$, and
\[
 \psi(\xi_{c},\xi_{q},\xi_{a})=e^{\i (\xi_{c}^{i} \hat X_{c,i}+\xi_{q}^{j} \hat X_{q,j}+\xi_{a}^{k} \hat X_{a,k}) } \psi.
\] 
Conversely, due to Lemma \ref{thm:bochoner_dshift}, a bounded continuous function $\varphi(\xi_{c},\xi_{q},\xi_{a};\eta_{c},\eta_{q},\eta_{a})$ satisfying $0\prec \varphi \prec \varphi_I$ uniquely determines an operator $A=V(\varphi)$ satisfying $0\le A \le I$. 
Moreover, the following Lemma gives a criterion for $V(\varphi)$ to be an element of ${\rm CCR}(O_{c})\otimes {\rm CCR}(S_{q})\otimes I_{a}$, 
which means that $V(\varphi)$ can be regarded as an operator acting on ${\rm CCR}(O_{c})\otimes {\rm CCR}(S_{q})$. 

\begin{lemma}\label{thm:commutant_qc}
Suppose that a bounded continuous function $\varphi(\xi_{c},\xi_{q},\xi_{a};\eta_{c},\eta_{q},\eta_{a})$ that fulfills the condition $0\prec \varphi \prec \varphi_I$ satisfies the identity
\[
\varphi(\xi_{c},\xi_{q},\xi_{a};\eta_{c},\eta_{q},\eta_{a})=e^{-\i\xi_{a}^{\top}S_{a}\eta_{a}}\varphi(\xi_{c}-\eta_{c},\xi_{q},\xi_{a}-\eta_{a};0,\eta_{q},0)
\]
for all $\xi_{c},\eta_{c}\in\R^{d_{c}}$, $\xi_{q},\eta_{q}\in\R^{d_{q}}$, $\xi_{a},\eta_{a}\in\R^{d_{a}}$. 
Then 
\[
V(\varphi)\in {\rm CCR}(O_{c})\otimes {\rm CCR}(S_{q})\otimes I_{a}. 
\]
\end{lemma}

\begin{proof} 
%See Section D of Supplementary Material \cite{qrep_supp}. 
See Appendix \ref{app:ProofLemmas}.
\end{proof}

\subsection{Sandwiched quantum L\'evy-Cram\'er continuity theorem}

In this subsection, we generalize the quantum L\'evy-Cram\'er continuity theorem \cite{qLevyCramer} and quantum Le Cam third lemma \cite{qcontiguity} in forms suitable for our discussion. 
Throughout this subsection, we use the following notations.
For each $n\in\N$, let $\rho^{(n)}$ be a quantum state and $X^{(n)}=(X_{1}^{(n)},\dots,X_{d}^{(n)})$ be a list of observables on a finite dimensional Hilbert space $\H^{(n)}$. 
Further, let
$X^{(\infty)}=(X_{1}^{(\infty)},\dots,X_{d}^{(\infty)})$ be the canonical observables for a quantum Gaussian state 
$\rho^{(\infty)}\sim N(h,J)$ 
with $J_{ij}=\Tr \rho^{(\infty)} X_j^{(\infty)} X_i^{(\infty)}$.

%--------------------------------------------------------------------------------------------------

The following Lemma is a variant of the noncommutative L\'evy-Cram\'er continuity theorem \cite{qLevyCramer,qLevyCramer2}.

\begin{lemma}[Sandwiched L\'evy-Cram\'er continuity theorem]\label{thm:sand_bounded}
Assume that
\begin{equation} \label{eq:sand_weak_origin}
 ( X^{(n)},\rho^{(n)} )\conv{}N(h,J), 
\end{equation}
and that a uniformly bounded sequence $\{A^{(n)}\}_{n\in\N\cup\{\infty\}}$ of observables 
satisfies 
\begin{equation} \label{eq:sand_weak}
 \lim_{n\to\infty}\Tr\rho^{(n)}e^{\i \xi^{i}X_{i}^{(n)}}A^{(n)}e^{\i\eta^{i}X_{i}^{(n)}}
	=\Tr\rho^{(\infty)}e^{\i \xi^{i}X_{i}^{(\infty)}}A^{(\infty)}e^{\i\eta^{i}X_{i}^{(\infty)}}
\end{equation}
for all $\xi, \eta\in \Q^d$. 
Then for any $\left\{ \xi_{s}\right\} _{s=1}^{r_{1}}, \left\{ \eta_{t}\right\} _{t=1}^{r_{2}}\subset\R^{d}$ and any real-valued bounded Borel functions $\left\{ f_{s}\right\} _{s=1}^{r_{1}}$, $\left\{ g_{t}\right\} _{t=1}^{r_{2}}$ whose discontinuity points form Lebesgue null sets, the following identity holds: 
\begin{align} \label{eq:sand_func}
 & \lim_{n\to\infty}\Tr\rho^{(n)}\left\{ \prod_{s=1}^{r_{1}}f_{s}(\xi_{s}^{i}X_{i}^{(n)})\right\} A^{(n)}\left\{ \prod_{t=1}^{r_{2}}g_{t}(\eta_{t}^{i}X_{i}^{(n)})\right\} ^{*} \\
 & \qquad=\Tr\rho^{(\infty)}\left\{ \prod_{s=1}^{r_{1}}f_{s}(\xi_{s}^{i}X_{i}^{(\infty)})\right\} A^{(\infty)}\left\{ \prod_{t=1}^{r_{2}}g_{t}(\eta_{t}^{i}X_{i}^{(\infty)})\right\}^{*}.  \nonumber
\end{align}
\end{lemma}

\begin{proof} 
%See Section D of Supplementary Material \cite{qrep_supp}. 
See Appendix \ref{app:ProofLemmas}.
\end{proof}

When $A^{(n)}=I^{(n)}$ for all $n\in\N\cup\{\infty\}$, Lemma \ref{thm:sand_bounded} is subsumed by \cite[Lemma 5.3]{qcontiguity}.
In this sense, Lemma \ref{thm:sand_bounded} is a slight generalization of \cite[Lemma 5.3]{qcontiguity}.
However, the assumption of boundedness for functions $f_s$ and $g_t$ in Lemma \ref{thm:sand_bounded} sometimes causes inconvenience in applications. 
We therefore further aim for generalizing Lemma \ref{thm:sand_bounded} to unbounded functions. 
The key to the generalization is the notion of uniform integrability \cite{qcontiguity}.

For quantum states $\{ \rho^{(n)}\} _{n\in\N}$ and observables $\{ B^{(n)}\} _{n\in\N}$ on Hilbert spaces $\{ \H^{(n)}\} _{n\in\N}$, 
we say that $B^{(n)}$ is {\em uniformly integrable} under $\rho^{(n)}$ 
if for all $\varepsilon>0$, there exists $L>0$ that satisfies
\[
\Tr\rho^{(n)} \left|B^{(n)}-h_L(B^{(n)})\right| <\varepsilon
\] 
for all $n$, where the function $h_L$ is defined by
\[ 
h_L(x)=\begin{cases}
x & (|x|\leq L)\\
0 & (|x|>L)
\end{cases}.
\] 
Using the notion of uniform integrability, Lemma \ref{thm:sand_bounded} is generalized as follows.

\begin{lemma}[Sandwiched L\'evy-Cram\'er continuity theorem: an extended version]\label{thm:sand_uni_int}
Under the same setting as in Lemma \ref{thm:sand_bounded} except that the functions $f_{1}$ and $g_{1}$ can be unbounded, assume further that both $\{ f_{1}(\xi_{1}^{i}X_{i}^{(n)}+o_{1}^{(n)})^{2}\} _{n\in\N\cup\{\infty\}}$
and $\{ g_{1}(\eta_{1}^{i}X_{i}^{(n)}+o_{2}^{(n)})^{2}\} _{n\in\N\cup\{\infty\}}$
are uniformly integrable under $\{ \rho^{(n)}\} _{n\in\N\cup\{\infty\}}$, 
where $o_{1}^{(n)}=o_{D}(\xi_{1}^{i}X_{i}^{(n)},\rho^{(n)})$ and $o_{2}^{(n)}=o_{D}(\eta_{1}^{i}X_{i}^{(n)},\rho^{(n)})$ 
for $n\in\N$ and $o_1^{(\infty)}=o_2^{(\infty)}=0$.
Then the following identity holds: 
\begin{align*}
 & \lim_{n\to\infty}\Tr\rho^{(n)}f_{1}(\xi_{1}^{i}X_{i}^{(n)}+o_{1}^{(n)})\left\{ \prod_{s=2}^{r_{1}}f_{s}(\xi_{s}^{i}X_{i}^{(n)})\right\} A^{(n)}\left\{ \prod_{t=2}^{r_{2}}g_{t}(\eta_{t}^{i}X_{i}^{(n)})\right\} ^{*}g_{1}(\eta_{1}^{i}X_{i}^{(n)}+o_{2}^{(n)})\\
 & \qquad =\Tr\rho^{(\infty)}\left\{ \prod_{s=1}^{r_{1}}f_{s}(\xi_{s}^{i}X_{i}^{(\infty)})\right\} A^{(\infty)}\left\{ \prod_{t=1}^{r_{2}}g_{t}(\eta_{t}^{i}X_{i}^{(\infty)})\right\} ^{*}. \nonumber
\end{align*} 
\end{lemma}

\begin{proof} 
%See Section D of Supplementary Material \cite{qrep_supp}. 
See Appendix \ref{app:ProofLemmas}.
\end{proof}

By using Lemma \ref{thm:sand_uni_int}, we can further generalize quantum Le Cam third lemma under q-LAN \cite[Corollary 7.5]{qcontiguity} as follows. 

\begin{corollary}[Sandwiched Le Cam third lemma under $D$-extendibility]\label{cor:lecam3_sand}
Assume that a sequence 
$\S^{(n)}=\{ \rho_{\theta}^{(n)} : \theta\in\Theta\subset\R^{d}\}$
of quantum statistical models is q-LAN 
and $D$-extendible at $\theta_0\in\Theta$ 
as in Definition \ref{def:propertyD}.  
Then
\begin{equation}\label{eq:lecam3_sand4}
 \left(X^{(n)},\rho_{\theta_{0}+h/\sqrt{n}}^{(n)}\right)\conv hN(\left({\rm Re}\,\tau\right)h,\Sigma),
\end{equation} 
where $\tau=\Sigma F$.

Assume further that a uniformly bounded sequence $A^{(n)}$ of observables for $n\in\N\cup\{\infty\}$ 
satisfies 
\begin{equation*} %\label{eq:sand_weak_third}
 \lim_{n\to\infty}\Tr\rho_{\theta_{0}}^{(n)}
	e^{\i \xi^{i}X_{i}^{(n)}}A^{(n)}e^{\i\eta^{i}X_{i}^{(n)}}
 =\Tr\rho_{0}^{(\infty)}e^{\i \xi^{i}X_{i}^{(\infty)}}A^{(\infty)}e^{\i\eta^{i}X_{i}^{(\infty)}}
\end{equation*}
for all $\xi, \eta\in \Q^r$, where $\rho_0^{(\infty)} \sim N(0,\Sigma)$ and $X^{(\infty)}=(X_{1}^{(\infty)},\dots,X_{r}^{(\infty)})$ are the canonical observables.
Then it holds that
\begin{equation}\label{eq:lecam3_sand1}
\lim_{n\to\infty}\Tr\rho_{\theta_{0}}^{(n)}R_{h_{1}}^{(n)}A^{(n)}R_{h_{2}}^{(n)}=\Tr\rho_{0}^{(\infty)}R_{h_{1}}^{(\infty)}A^{(\infty)}R_{h_{2}}^{(\infty)}
\end{equation}
for any $h_{1},h_{2}\in\R^{d}$, where $R_{h}^{(n)}$ are square-root likelihood ratios defined by
\[ 
 R_{h}^{(n)}=\Ratio\left( \left. \rho_{\theta_{0}+h/\sqrt{n}}^{(n)} \,\right|\, \rho_{\theta_{0}}^{(n)}\right) 
 \;\;\mbox{and}\;\;
 R_{h}^{(\infty)}=\exp \left[ \frac{1}{2}\left((Fh)^{i}X_{i}^{(\infty)}-\frac{1}{2}(h^{\top}F^{\top}\Sigma Fh) I^{(\infty)} \right) \right].
\]
Specifically, 
\begin{equation}\label{eq:lecam3_sand2}
 \lim_{n\to\infty}\Tr\rho_{\theta_{0}+h/\sqrt{n}}^{(n)}A^{(n)} 
 =\lim_{n\to\infty}\Tr \rho_{\theta_0}^{(n)} R_{h}^{(n)}A^{(n)}R_{h}^{(n)} 
 =\Tr\rho_{h}^{(\infty)}A^{(\infty)}
\end{equation}
for any $h\in\R^{d}$, where $\rho_{h}^{(\infty)} \sim N(({\rm Re}\,\tau)h,\Sigma)$. 
\end{corollary}

\begin{proof} 
%See Section D of Supplementary Material \cite{qrep_supp}. 
See Appendix \ref{app:ProofLemmas}.
\end{proof}

Finally, the following asymptotic version of the Weyl CCR will turn out to be useful.

\begin{lemma}[Asymptotic Weyl CCR]\label{lem:asymptoticCCR}
Let $W^{(n)}(\xi):=e^{\i\xi^{i}X_{i}^{(n)}}$ for $\xi\in\R^d$, and assume that $(X^{(n)},\rho^{(n)})\conv{}N(0,J)$. 
Then 
\begin{equation*} %\label{eqn:asymptoticCCR}
 \lim_{n\to\infty}\left\| W^{(n)}(\xi)W^{(n)}(\eta)\sqrt{\rho^{(n)}}
 - e^{\i \xi^\top S\eta}\; W^{(n)}(\xi+\eta)\sqrt{\rho^{(n)}} \right\|_{\rm HS}=0,
\end{equation*}
for all $\xi,\eta\in\R^d$, where $S:={\rm Im}\,J$, and $\|\,\cdot\,\|_{\rm HS}$ denotes the Hilbert-Schmidt norm.
\end{lemma}

\begin{proof} 
%See Section D of Supplementary Material \cite{qrep_supp}. 
See Appendix \ref{app:ProofLemmas}.
\end{proof}

%----------------------------------------------------------------------------------------------------------------------------------------------------------------
\section{Proof of Theorem \ref{thm:qRep}} \label{sec:proof_qrep}
%----------------------------------------------------------------------------------------------------------------------------------------------------------------

Since the proof is somewhat lengthy, we first outline the proof. 
By choosing a suitable regular $r\times r$ matrix $K$, one finds another $D$-extension 
${X'}_i^{(n)}:=K_i^j X_j^{(n)}$ of $\Delta^{(n)}$ that exhibits 
\[ 
 X'^{(n)}\convd{\rho_{\theta_0}^{(n)}} N(0,\Sigma_{c}\oplus\Sigma_{q}), 
\]
where $\Sigma_{c}$ is a real $r_{c}\times r_{c}$ matrix and $\Sigma_{q}$ is a complex $r_{q}\times r_{q}$ matrix with $r_{c}+r_{q}=r$ so that the imaginary part $S_q:={\rm Im}\, \Sigma_{q}$ of $\Sigma_q$ is invertible%
\footnote{
In fact, for any $\Sigma \ge 0$ with ${\rm Re}\,\Sigma>0$, we have 
$({{\rm Re}\Sigma} )^{-1}\Sigma (\sqrt{{\rm Re}\Sigma} )^{-1}=I+\i S$, where 
$S:=(\sqrt{{\rm Re}\Sigma} )^{-1} ({\rm Im}\,\Sigma ) (\sqrt{{\rm Re}\Sigma} )^{-1}$ is a real skew-symmetric matrix. 
Further, by choosing a suitable real orthogonal matrix $P$, the matrix $S$ is transformed into the form 
$P^{\top} S P=0\oplus S_{q}$ with $\det S_{q}\neq 0$.
}. 
In what follows, we always adopt such a $D$-extension and simply denote $X'^{(n)}$ as $X^{(n)}$, omitting the prime. 
Further, we label the elements of $X^{(n)}$ as
\[ 
X_{c,1}^{(n)},\dots,X_{c,r_{c}}^{(n)},X_{q,1}^{(n)},\dots,X_{q,r_{q}}^{(n)},
\]
in accordance with the decomposition $\Sigma=\Sigma_{c}\oplus\Sigma_{q}$. 

We need to show that there exists a POVM $M$ on ${\rm CCR}(O_c) \otimes {\rm CCR}(S_q)$ that exhibits
$\phi_h(M(B))=\mathcal{L}_h(B)$ for any $B\in\B(R^s)$,
where $\phi_h\sim N(({\rm Re}\, \tau)h, \Sigma)$. 
To this end, we first formally enlarge $N(({\rm Re}\, \tau)h, \Sigma)$ to $N( ({\rm Re} \, \hat{\tau}) h, \hat{\Sigma} )$,  
where
\[
\hat{\tau}:=\begin{pmatrix}\tau\\ 0 \end{pmatrix},
\qquad
\hat{\Sigma}:=\begin{pmatrix}\Sigma_{c} & 0 & 0\\
0 & \Sigma_{q} & \Sigma_{q}\#\Sigma_{q}^{\top}\\
0 & \Sigma_{q}\#\Sigma_{q}^{\top} & \Sigma_{q}^{\top}
\end{pmatrix}.
\]
We next construct a POVM
on ${\rm CCR}(O_c) \otimes {\rm CCR}(S_q) \otimes {\rm CCR}(-S_q)$ and prove that it is a POVM on ${\rm CCR}(O_c) \otimes {\rm CCR}(S_q) \otimes I_a$. 
Finally, we prove that the POVM thus constructed enjoys the desired property. 

\begin{proof} 
As stated above, we divide the proof into three steps. 
In Step 1, we define a Hilbert space on which $N( ({\rm Re} \, \hat{\tau}) h, \hat{\Sigma} )$ is represented, and designate a fiducial  cyclic vector for ${\rm CCR}(O_c) \otimes {\rm CCR}(S_q) \otimes {\rm CCR}(-S_q)$. 
In Step 2, we construct (a precursor of) a POVM
on ${\rm CCR}(O_c) \otimes {\rm CCR}(S_q) \otimes {\rm CCR}(-S_q)$ and prove that it defines a POVM $M$ on ${\rm CCR}(O_c) \otimes {\rm CCR}(S_q)$. 
In Step 3, we prove that $M$ enjoys the desired property.

\subsection*{Step 1}
As a similar way to the prescription that precedes Lemma \ref{thm:commutant_qc}, we introduce a Hilbert space $\H=\H_c \otimes \H_q \otimes \H_a$ on which the von Neumann algebra ${\rm CCR}(\hat{S})={\rm CCR}(O_c) \otimes {\rm CCR}(S_q) \otimes {\rm CCR}(S_a)$ is represented, where 
\[ 
 \hat{S}={\rm Im}\, \hat{\Sigma}, \quad 
 O_c={\rm Im}\, \Sigma_c \,(=0), \quad 
 S_q={\rm Im}\, \Sigma_q, \quad\mbox{and}\quad 
 S_a={\rm Im}\,\Sigma_q^\top=-S_q.
\]
The canonical observables are 
\[
 \left\{ \hat X_{c,i}:=X_{c,i}\otimes I_{q}\otimes I_{a}\right\} _{i} \cup 
 \left\{ \hat X_{q,j}:=I_{c}\otimes X_{q,j}\otimes I_{a}\right\} _{j} \cup
 \left\{ \hat X_{a,k}:=I_{c}\otimes I_{q}\otimes X_{a,k}\right\} _{k},
\]
where $\{{X}_{c,i}\}_{i}$, $\{{X}_{q,j}\}_{j}$, and $\{{X}_{a,k}\}_{k}$ are the canonical observables of ${\rm CCR}(O_{c})$, ${\rm CCR}(S_{q})$, and ${\rm CCR}(S_{a})$, respectively.

In order to invoke the sandwiched coherent state representation for ${\rm CCR}(\hat S)$, we need a cyclic vector $\psi$ on $\H$. 
We first designate a cyclic vector $\psi_c\in\H_c$ for ${\rm CCR}(O_c)=L^\infty(\R^{r_c})$, 
in which each $\xi\in L^\infty(\R^{r_c})$ is identified with the bounded operator $T_\xi\in B(\H_c)$ defined by
\[  (T_\xi \varphi)(x):=\xi(x)\varphi(x),\qquad (\varphi \in \H_c,\; x\in\R^{r_c}).  \]
Let $\psi_{c}:=\sqrt{g(x)}$, where $g(x)$ is the density function of the (classical) Gaussian distribution $N(0,\Sigma_{c})$. 
Then, any function $f\in L^2(\R^{r_c})$ can be approximated by a series of functions $f_n:=T_{\xi_n} \psi_{c}$, 
where 
\[
  %\xi_n:=\frac{f(x)}{\sqrt{g(x)}} \,\mathbb{1}_{B_n} (x),
  \xi_n:=\frac{f(x)}{\sqrt{g(x)}} \,\mathbbm{1}_{B_n} (x),
\]
with $\mathbbm{1}_{B_n}$ being the indicator function of the ball $B_n$ of radius $n\in\N$ centered at the origin of $\R^{r_c}$.
As a consequence, $\psi_{c}$ is a cyclic vector of $\H_{c}$.

We next specify a cyclic vector $\psi_{qa}\in\H_q\otimes\H_a$ for ${\rm CCR}(S_q)\otimes {\rm CCR}(S_a)$.
Recall that 
\[
N\left(0,\begin{pmatrix}\Sigma_{q} & \Sigma_{q}\#\Sigma_{q}^{\top}\\
\Sigma_{q}\#\Sigma_{q}^{\top} & \Sigma_{q}^{\top}\end{pmatrix}\right)
\]
is a pure state on $\H_{q}\otimes\H_{a}$ (Corollary \ref{lem:pure_double}). 
Let $\psi_{qa}\in\H_{q}\otimes\H_{a}$ be a unit vector that corresponds to the above state. 
Then it is well known in the theory of coherent states that $\psi_{qa}$ is a cyclic vector for ${\rm CCR}(S_q)\otimes {\rm CCR}(S_a)$. 

Now we arrive at a cyclic vector
$\psi:=\psi_{c}\otimes\psi_{qa}\in\H_{c}\otimes\H_{q}\otimes\H_{a}$ for ${\rm CCR}(\hat S)$. 
This cyclic vector has the following nice property. 
Let 
\[
\hat{\Delta}_{i}:=F_{i}^{j}\hat{X_{j}}
\]
for $1\leq i \leq d$, where $F$ is the $r\times d$ real matrix introduced in Definition \ref{def:propertyD}, 
and let
\[
\hat{R}_{h}:=\exp\left[ \frac{1}{2}\left(h^{i}\hat{\Delta}_{i}-\frac{1}{2} (h^{\top}F^{\top}\Sigma Fh) \hat I \right)\right]
\]
for $h\in\R^d$, where $\hat I$ is the identity on ${\rm CCR}(\hat S)$. 
Then, for any $A\in {\rm CCR}(O_c)\otimes {\rm CCR}(S_q)$, the following identity holds:
\begin{equation} \label{eq:qrep_restrict}
\langle \psi, \hat{R}_h (A\otimes I_a) \hat{R}_h \psi \rangle = \phi_h(A),
\end{equation}
where $\phi_h \sim N(({\rm Re}\,\tau)\,h,\Sigma)$ with $\tau:=\Sigma F$.  
This relation will be used as a variant of the quantum Le Cam third lemma that goes back and forth between ${\rm CCR}(\hat S)$ and ${\rm CCR}(S)$.

To prove \eqref{eq:qrep_restrict},
let
\begin{equation}\label{eqn:Xbar}
\overline{X}_{i}:=\left\{
 \array{ll} X_{c,i}\otimes I_q & \;\;{\rm if} \;\; 1\le i \le r_c \\\\
  I_c\otimes X_{q,i-r_c}  & \;\;{\rm if} \;\; r_c+1\le i \le r_c+r_q
  \endarray \right.
\end{equation}
be canonical observables of ${\rm CCR}(O_c) \otimes {\rm CCR}(S_q)$.
Then by a direct computation using the quasi-characteristic function of the vector state%
\footnote{
If $\det\hat S=0$, the Hilbert space $\H_c$ is reducible under the action of ${\rm CCR}(O_c)$, and thus the vector state $\ket\psi\bra\psi$ is a mixed state. 
}
\[
 \ket\psi\bra\psi \sim N(0,\hat\Sigma), 
\]
we can verify that
\[
\langle \psi, \hat{R}_h (e^{\i \xi^i \overline{X_i}} \otimes I_a) \hat{R}_h \psi \rangle
=\langle \psi, \hat{R}_h \,e^{\i \sum_{i=1}^{r} \xi^i \hat{X_i}} \hat{R}_h \psi \rangle
=e^{\i \xi^\top ({\rm Re}\, \tau) h - \frac{1}{2} \xi^\top \Sigma \xi }.
\]
Since the last side is the characteristic function of $\phi_h \sim N(({\rm Re}\,\tau)\,h,\Sigma)$, we have
\[
\langle \psi, \hat{R}_h (e^{\i \xi^i \overline{X_i}} \otimes I_a) \hat{R}_h \psi \rangle
=\phi_h \left(e^{\i \xi^i \overline{X_i} } \right).
\]
Finally, since $\{e^{\i \xi^i \overline{X_i}} \}_{\xi\in\R^r}$ is SOT-dense in ${\rm CCR}(O_c) \otimes {\rm CCR}(S_q)$, the identity \eqref{eq:qrep_restrict} is proved. 

\subsection*{Step 2}

Given a pair of vectors $(\lambda,\mu)\in \R^{r_c}\times \R^{r_q}$, let 
\[
  W(\lambda, \mu):=e^{\i\left(\lambda^{i}X_{c,i}+\mu^{j}X_{q,j}\right)}
\]
be the corresponding Weyl operator on ${\rm CCR}(O_c)\otimes {\rm CCR}(S_q)$. 
By analogy to this operator, we introduce a unitary operator
\[
 W^{(n)}(\lambda, \mu):=e^{\i\left(\lambda^{i}X_{c,i}^{(n)}+\mu^{j}X_{q,j}^{(n)}\right)}
\]
on each $\H^{(n)}$.
We further define, for each $\xi=(\xi_{c},\xi_{q},\xi_{a})\in \R^{r_c}\times \R^{r_q}\times \R^{r_q}$, 
operators
\[
 A^{(n)}(\xi):=W^{(n)}(\xi_c,\xi_q) \sqrt{\rho^{(n)}_{\theta_0}}\, W^{(n)}(0,\xi_a)
\]
and
\[
 B^{(n)}(\xi):=W^{(n)}(0,\xi_q) \sqrt{\rho^{(n)}_{\theta_0}}\, W^{(n)}(\xi_c,\xi_a). 
\]

Note that these operators are asymptotically identified, in that 
\begin{equation}\label{eqn:A-B}
  \lim_{n\to\infty} \left\| A^{(n)}(\xi)-B^{(n)}(\xi) \right\|_{\rm HS} = 0.
\end{equation}
This is proved by observing  
\[
\left\| A^{(n)}(\xi)-B^{(n)}(\xi) \right\|_{\rm HS}^2
=2-2\,{\rm Re}\,\Tr A^{(n)}(\xi)^*B^{(n)}(\xi)
\]
and 
\begin{align*}
 &\lim_{n\to\infty} \Tr A^{(n)}(\xi)^*B^{(n)}(\xi) \\
 &\quad=\lim_{n\to\infty} \Tr W^{(n)}(0,\xi_a)^* \sqrt{\rho^{(n)}_{\theta_0}}\, W^{(n)}(\xi_c,\xi_q)^*
		  W^{(n)}(0,\xi_q) \sqrt{\rho^{(n)}_{\theta_0}}\, W^{(n)}(\xi_c,\xi_a)  \\
 &\quad=\lim_{n\to\infty} \Tr \left\{ W^{(n)}(\xi_c,\xi_a) W^{(n)}(0, -\xi_a)\sqrt{\rho^{(n)}_{\theta_0}}\right\} 
 		  \left\{W^{(n)}(-\xi_c,-\xi_q)W^{(n)}(0,\xi_q) \sqrt{\rho^{(n)}_{\theta_0}} \right\} \\
 &\quad=\lim_{n\to\infty}\Tr \left\{ \exp\left[\i \begin{pmatrix}\xi_{c} \\ \xi_{a} \end{pmatrix}^{\top} (O_c\oplus S_q) 
 		\begin{pmatrix} 0 \\ -\xi_{a} \end{pmatrix} \right]
 	W^{(n)}(\xi_c, 0)\sqrt{\rho^{(n)}_{\theta_0}} \right\} \\
 & \qquad\qquad\;\; \times \left\{ \exp\left[\i \begin{pmatrix} -\xi_{c} \\ -\xi_{q} \end{pmatrix}^{\top} (O_c\oplus S_q) 
 		\begin{pmatrix} 0 \\ \xi_{q} \end{pmatrix} \right]
 	W^{(n)}(-\xi_c, 0)\sqrt{\rho^{(n)}_{\theta_0}} \right\} \\
 &\quad=\lim_{n\to\infty} \Tr  
 	W^{(n)}(\xi_c, 0)\sqrt{\rho^{(n)}_{\theta_0}}\,W^{(n)}(-\xi_c, 0)\sqrt{\rho^{(n)}_{\theta_0}} \\
 &\quad = \exp\left\{ -\frac{1}{2}
 \begin{pmatrix}\xi_c \\ 0 \\ -\xi_c \\ 0 \end{pmatrix}^{\top}
 \begin{pmatrix}\Sigma_{c} & 0 & \Sigma_{c} & 0\\
	0 & \Sigma_{q} & 0 & \Sigma_{q}\#\Sigma_{q}^{\top}\\
	\Sigma_{c} & 0 & \Sigma_{c} & 0\\
	0 & \Sigma_{q}\#\Sigma_{q}^{\top} & 0 & \Sigma_{q}^{\top}
  \end{pmatrix}
 \begin{pmatrix}\xi_c \\ 0 \\ -\xi_c \\ 0 \end{pmatrix}
 \right\} \\
 &\quad= 1.
\end{align*}
Here, the asymptotic Weyl CCR (Lemma \ref{lem:asymptoticCCR}) was used in the third equality, 
and condition \eqref{eq:X_cond3} for $D$-extendibility was used in the second last equality. 

Now, given a POVM $M^{(n)}$ on $\H^{(n)}$ whose outcomes take values in $\R^s$, let
\[
  M_t^{(n)}:=M^{(n)}((-\infty, t])
\]
be the associated resolution of identity, where $t=(t_1,t_2,\dots,t_s)\in\R^s$ and $(-\infty, t]$ is the shorthand of the set
\[
 (-\infty, t_1]\times (-\infty, t_2]\times \cdots \times (-\infty, t_s].
\]
By using the resolution of identity $M_t^{(n)}$, we define the following function
\[
 \varphi_{t}^{(n)}(\xi;\eta):=\Tr A^{(n)}(\xi)^{*}M_{t}^{(n)}A ^{(n)}(\eta)
\]
for $\xi=(\xi_{c},\xi_{q},\xi_{a}), \eta=(\eta_{c},\eta_{q},\eta_{a}) \in \R^{r_c}\times \R^{r_q}\times \R^{r_q}$. 
Since $\varphi_{t}^{(n)}(\xi;\eta)$ is uniformly bounded, 
in that $|\varphi_{t}^{(n)}(\xi;\eta)|\leq 1$ for all $t\in \R^s$, $\xi,\eta\in \R^{r_c+2r_q}$, and $n\in\N$, 
the diagonal sequence trick \cite{ReedSimon} tells us that there is a subsequence $\{n_{m}\}_{m\in\N}\subset\{n\}_{n\in\N}$ such that $\varphi_{\alpha}^{(n_{m})}(\xi;\eta)$ are convergent for all countably many arguments $\alpha\in\Q^{s}$ and $\xi,\eta\in\Q^{r_{c}+2r_{q}}$, defining a limiting function 
\begin{equation}\label{eq:double_shift_lim}
 \varphi_{\alpha}(\xi;\eta):=\lim_{m\to\infty}\varphi_{\alpha}^{(n_{m})}(\xi;\eta). 
\end{equation}
We shall prove that this limiting function $\varphi_{\alpha}$ is the sandwiched coherent state representation of some operator $\tilde{M}_\alpha$ on $\H_c\otimes \H_q$.

First, we formally introduce the function $\varphi_\infty$ by 
\begin{align*}
\varphi_{\infty}(\xi;\eta) 
 & :=\left\langle e^{\i\xi^{i}\hat{X}_{i}} \psi, e^{\i\eta^{j}\hat{X}_{j}} \psi \right\rangle\\
 & =e^{-\i\xi^{\top}\hat{S}\eta} \left\langle \psi, e^{\i\left(\eta-\xi\right)^{i}\hat{X}_{i}} \psi \right\rangle\\
 & =e^{-\i\xi^{\top}\hat{S}\eta}\, e^{-\frac{1}{2}(\eta-\xi)^{\top}\hat{\Sigma}(\eta-\xi)}.
\end{align*}
Then it is shown that
\[
 \varphi_{\infty}(\xi;\eta)=\lim_{n\to\infty}\varphi_{\infty}^{(n)}(\xi;\eta),
\]
where $\varphi_{\infty}^{(n)}(\xi;\eta):=\lim_{t\to\infty} \varphi_t^{(n)}(\xi;\eta)=\Tr A^{(n)}(\xi)^{*}A^{(n)}(\eta)$.
In fact, 
\begin{align*}
 & \lim_{n\to\infty}\varphi_{\infty}^{(n)}(\xi;\eta) \\
 &\quad =\lim_{n\to\infty}\Tr A^{(n)}(\xi)^{*}A^{(n)}(\eta) \\
 &\quad=\lim_{n\to\infty} \Tr W^{(n)}(0,\xi_a)^* \sqrt{\rho^{(n)}_{\theta_0}}\, W^{(n)}(\xi_c,\xi_q)^*
		  W^{(n)}(\eta_c,\eta_q) \sqrt{\rho^{(n)}_{\theta_0}}\, W^{(n)}(0,\eta_a)  \\
 &\quad=\lim_{n\to\infty} \Tr \left\{ W^{(n)}(0,\eta_a) W^{(n)}(0, -\xi_a)\sqrt{\rho^{(n)}_{\theta_0}}\right\} 
 		  \left\{W^{(n)}(-\xi_c,-\xi_q)W^{(n)}(\eta_c,\eta_q) \sqrt{\rho^{(n)}_{\theta_0}} \right\} \\
 &\quad=\lim_{n\to\infty}\Tr \left\{ \exp\left[\i \begin{pmatrix} 0\\ \eta_{a} \end{pmatrix}^{\top} (O_c\oplus S_q) 
 		\begin{pmatrix} 0 \\ -\xi_{a} \end{pmatrix} \right]
 	W^{(n)}(0, \eta_a-\xi_a)\sqrt{\rho^{(n)}_{\theta_0}} \right\} \\
 & \qquad\qquad\;\; \times \left\{ \exp\left[\i \begin{pmatrix} -\xi_{c} \\ -\xi_{q} \end{pmatrix}^{\top} (O_c\oplus S_q) 
 		\begin{pmatrix} \eta_c \\ \eta_{q} \end{pmatrix} \right]
 	W^{(n)}(\eta_c-\xi_c, \eta_q-\xi_q)\sqrt{\rho^{(n)}_{\theta_0}} \right\} \\
 &\quad=e^{\i (-\eta_a^\top S_q \xi_a - \xi_q^\top S_q \eta_q)}
 	\lim_{n\to\infty} \Tr  
 	W^{(n)}(0, \eta_a-\xi_a) \sqrt{\rho^{(n)}_{\theta_0}}\,
	W^{(n)}(\eta_c-\xi_c, \eta_q-\xi_q) \sqrt{\rho^{(n)}_{\theta_0}} \\
 &\quad =e^{-\i \xi^\top \hat{S} \eta}\; 
 	\exp\left\{ -\frac{1}{2}
	 \begin{pmatrix} 0 \\ \eta_a-\xi_a \\ \eta_c-\xi_c \\ \eta_q-\xi_q \end{pmatrix}^{\top}
	 \begin{pmatrix}\Sigma_{c} & 0 & \Sigma_{c} & 0\\
		0 & \Sigma_{q} & 0 & \Sigma_{q}\#\Sigma_{q}^{\top}\\
		\Sigma_{c} & 0 & \Sigma_{c} & 0\\
		0 & \Sigma_{q}\#\Sigma_{q}^{\top} & 0 & \Sigma_{q}^{\top}
	 \end{pmatrix}
	 \begin{pmatrix} 0 \\ \eta_a-\xi_a \\ \eta_c-\xi_c \\ \eta_q-\xi_q \end{pmatrix}
	 \right\} \\
 &\quad= e^{-\i \xi^\top \hat{S} \eta}\; e^{-\frac{1}{2} (\eta-\xi)^\top \hat\Sigma (\eta-\xi)} \\
 &\quad=\varphi_{\infty}(\xi;\eta).
\end{align*}
As a consequence, by taking the limit $m\to\infty$ in $0\prec \varphi_{\alpha}^{(n_m)} \prec\varphi_{\infty}^{(n_m)}$, which follows from $0\le M_\alpha^{(n_m)} \le I^{(n_m)}$, we have
\begin{equation}\label{eqn:(i)}
 0\prec\varphi_{\alpha}\prec\varphi_{\infty}, \qquad (\forall \alpha\in\Q^d). 
\end{equation}

We can also prove the following identity: 
\begin{equation}\label{eqn:(ii)}
 \varphi_{\alpha}(\xi_{c},\xi_{q},\xi_{a};\eta_{c},\eta_{q},\eta_{a})
  =e^{-\i\xi_{a}^{\top}S_{a}\eta_{a}}\varphi_{\alpha}(\xi_{c}-\eta_{c},\xi_{q},\xi_{a}-\eta_{a};0,\eta_{q},0).
\end{equation}
In fact, by using the asymptotic identifiability of $A^{(n)}(\xi)$ and $B^{(n)}(\xi)$ established in \eqref{eqn:A-B}, 
\begin{align*}
 &\varphi_\alpha (\xi ; \eta) \\
 &\quad =\lim_{m\to\infty}\Tr A^{(n_m)}(\xi)^{*} M_\alpha^{(n_m)} A^{(n_m)}(\eta) \\
 &\quad =\lim_{m\to\infty}\Tr B^{(n_m)}(\xi)^{*} M_\alpha^{(n_m)} B^{(n_m)}(\eta) \\
 &\quad=\lim_{m\to\infty} \Tr W^{(n_m)}(\xi_c,\xi_a)^* \sqrt{\rho^{(n_m)}_{\theta_0}}\, W^{(n_m)}(0,\xi_q)^*
 		M_\alpha^{(n_m)} W^{(n_m)}(0,\eta_q) \sqrt{\rho^{(n_m)}_{\theta_0}}\, W^{(n_m)}(\eta_c,\eta_a)  \\
 &\quad=\lim_{m\to\infty} \Tr \left\{ W^{(n_m)}(\eta_c,\eta_a) W^{(n_m)}(-\xi_c, -\xi_a)
 		\sqrt{\rho^{(n_m)}_{\theta_0}}\right\} 
 		  \left\{W^{(n_m)}(0,-\xi_q) M_\alpha^{(n_m)} W^{(n_m)}(0,\eta_q) \sqrt{\rho^{(n_m)}_{\theta_0}} \right\} \\
 &\quad=\lim_{m\to\infty}\Tr \left\{ \exp\left[\i \begin{pmatrix} \eta_c\\ \eta_{a} \end{pmatrix}^{\top} (O_c\oplus S_q) 
 		\begin{pmatrix} -\xi_c \\ -\xi_{a} \end{pmatrix} \right]
 		W^{(n_m)}(\eta_c-\xi_c, \eta_a-\xi_a)\sqrt{\rho^{(n_m)}_{\theta_0}} \right\} \\
 & \qquad\qquad\;\; \times 
 		\left\{W^{(n_m)}(0,-\xi_q) M_\alpha^{(n_m)} W^{(n_m)}(0,\eta_q) \sqrt{\rho^{(n_m)}_{\theta_0}} \right\} \\
 &\quad=e^{-\i \eta_a^\top S_q \xi_a} 
 		\lim_{m\to\infty} \Tr W^{(n_m)}(\xi_c-\eta_c, \xi_a-\eta_a)^*\sqrt{\rho^{(n_m)}_{\theta_0}} \,
		W^{(n_m)}(0,\xi_q)^* M_\alpha^{(n_m)} W^{(n_m)}(0,\eta_q) \sqrt{\rho^{(n_m)}_{\theta_0}} \\
 &\quad=e^{-\i \eta_a^\top S_q \xi_a} 
 		\lim_{m\to\infty} \Tr B^{(n_m)}(\xi_c-\eta_c, \xi_q, \xi_a-\eta_a)^* M_\alpha^{(n_m)} 
			B^{(n_m)}(0,\eta_q,0) \\
 &\quad=e^{-\i \eta_a^\top S_q \xi_a} 
 		\lim_{m\to\infty} \Tr A^{(n_m)}(\xi_c-\eta_c, \xi_q, \xi_a-\eta_a)^* M_\alpha^{(n_m)} 
			A^{(n_m)}(0,\eta_q,0) \\
 &\quad=e^{-\i\xi_{a}^{\top}S_{a}\eta_{a}}\varphi_{\alpha}(\xi_{c}-\eta_{c},\xi_{q},\xi_{a}-\eta_{a};0,\eta_{q},0).
\end{align*}

Now that \eqref{eqn:(i)} and \eqref{eqn:(ii)} are verified, Lemmas \ref{thm:bochoner_dshift} and \ref{thm:commutant_qc} prove that there is a unique operator $\tilde{M}_{\alpha} \in {\rm CCR}(O_c) \otimes {\rm CCR}(S_q)$ satisfying $0\leq\tilde{M}_{\alpha}\leq I_c\otimes I_q$ and 
\begin{equation}\label{eq:M_alpha}
  \varphi_{\alpha}(\xi;\eta)
  =\left\langle e^{\i\xi^{i}\hat{X}_{i}} \psi, (\tilde{M}_{\alpha}\otimes I_{a} )e^{\i\eta^{j}\hat{X}_{j}} \psi \right\rangle,
\end{equation}
for all $\alpha\in\Q^{s}$ and $\xi,\eta\in\R^{r_{c}+2r_{q}}$. 

We are now ready to construct a POVM $M=\{M(B): B\in \B(\R^{s})\}$ from $\{ \tilde M_\alpha \}_{\alpha\in\Q^s}$.  
Since $\tilde M_\alpha$ is monotone in $\alpha\in\Q^s$, we can define, for each $t\in\R^s$,
\[
\bar{M}_{t}:=\inf_{\alpha>t,\alpha\in\Q^{s}}\tilde{M}_{\alpha},
\]
where the infimum is taken in the weak operator topology (WOT).
Since $t\mapsto\bar{M}_{t}$ is right-continuous, 
it uniquely determines a POVM 
$\bar{M}=\{\bar{M}(B): B\in \B(\bar{\R}^{s})\}$ 
over the extended reals $\bar{\R}^{s}$.
Finally, we transfer the `measure at infinity' $\bar M(\bar{\R}^{s}\setminus\R^{s})$ to the origin, to obtain
\[
 M(B):=\bar{M}(B)+\delta_{0}(B)\bar{M}(\bar{\R}^{s}\setminus\R^{s}),\qquad (B\in\B(\R^s)),
\]
where $\delta_0$ is the Dirac measure concentrated at the origin.

\subsection*{Step 3}

We prove that the POVM $M$ constructed in Step 2 is the desired one we have sought. 
Setting $\xi_a=\eta_a=0$ in \eqref{eq:double_shift_lim} and \eqref{eq:M_alpha}, we have
\begin{align*}
\varphi_\alpha(\xi_c,\xi_q,0 ; \eta_c, \eta_q,0)
&=\lim_{m\to\infty}\Tr \sqrt{\rho_{\theta_{0}}^{(n_{m})} } e^{-\i\xi^{i}X_{i}^{(n_{m})}}
	M_{\alpha}^{(n_{m})}e^{\i\eta^{i}X_{i}^{(n_{m})}} \sqrt{\rho_{\theta_{0}}^{(n_{m})} } \\
&=\left\langle \psi, e^{-\i\xi^{i} (\overline{X}_{i}\otimes I_{a}) }\left(\tilde{M}_{\alpha}\otimes I_{a}\right)
	e^{\i\eta^{i} (\overline{X}_{i}\otimes I_{a}) } \psi \right\rangle,
\end{align*}
or equivalently, 
\begin{align*}
\lim_{m\to\infty} \Tr \rho_{\theta_{0}}^{(n_{m})}  e^{-\i\xi^{i}X_{i}^{(n_{m})}}
	M_{\alpha}^{(n_{m})}e^{\i\eta^{i}X_{i}^{(n_{m})}}
=\left\langle \psi, \left( e^{-\i\xi^{i}\overline{X}_{i}} \tilde{M}_{\alpha}
	e^{\i\eta^{i} \overline{X}_{i} } \otimes I_{a} \right) \psi \right\rangle.
\end{align*}
Due to \eqref{eq:qrep_restrict}, this is further equal to
\[
  \phi_0 \left( e^{-\i\xi^{i} \overline{X}_{i}}\tilde{M}_{\alpha}e^{\i\eta^{i} \overline{X}_{i}} \right). 
\]
Therefore, the sandwiched Le Cam third lemma (Corollary \ref{cor:lecam3_sand}) yields 
\begin{equation}\label{eq:last_step4}
 \lim_{m\to\infty} \L^{(n_m)}_h (-\infty,\alpha]
 =\lim_{m\to\infty}\Tr\rho_{\theta_{0}+h/\sqrt{n_{m}}}^{(n_{m})}M_{\alpha}^{(n_{m})} 
 = \phi_h(\tilde{M}_{\alpha})
\end{equation}
for all $h\in\R^{d}$ and $\alpha\in\Q^{s}$. 

Fix $h\in\R^{d}$ arbitrarily. 
Due to assumption, $\L_h^{(n)}$ weakly converges to $\L_h$. 
Therefore, for any continuity point $t\in\R^{s}$ of $t\mapsto\L_{h}(-\infty,t]$, 
\begin{align*}
\L_{h}(-\infty,t] &= \lim_{m\to\infty} \L^{(n_m)}_h (-\infty,t] \\
%\label{eq:last_step1}\\
 & \leq \inf_{\alpha>t,\alpha\in\Q^s} \lim_{m\to\infty} \L^{(n_m)}_h (-\infty,\alpha] \\
% \label{eq:last_step2}\\
 & \leq \inf_{\alpha>t,\alpha\in\Q^s} \L_{h}(-\infty,\alpha] \\
 & = \L_{h}(-\infty,t].
 %\label{eq:last_step3}
\end{align*}
In the second inequality, we used the portmanteau lemma.
It then follows from \eqref{eq:last_step4} that
\begin{align*}
\L_{h}(-\infty,t] & =\inf_{\alpha>t,\alpha\in\Q^s} \lim_{m\to\infty} \L^{(n_m)}_h (-\infty,\alpha] 
  =\inf_{\alpha>t,\,\alpha\in\Q^{s}} \phi_h (\tilde{M}_{\alpha})
 = \phi_h(\bar{M}_{t}),
\end{align*}
and thus $\L_{h}(B)=\phi_h(\bar{M}(B))$ for all $B\in\B(\R^s)$: 
in particular, 
\[ \phi_h (\bar{M}(\R^{s}))=\L_{h}(\R^{s})=1. \]
Since $\L_{h}(\bar{\R}^{s}\setminus\R^{s})=0$, 
we have $\phi_h (\bar{M}(B))=\phi_h(M(B))$ for all $B\in\B(\R^s)$.

In summary,
\[ \L_{h}(B)=\phi_h(M(B)) \qquad (\forall B\in\B(\R^s)). \]
This completes the proof of Theorem \ref{thm:qRep}. 
\end{proof}

%----------------------------------------------------------------------------------------------------------------------------------------------------------------
\section{Applications}\label{sec:applications}
%----------------------------------------------------------------------------------------------------------------------------------------------------------------

In this section, we apply the asymptotic representation Theorem \ref{thm:qRep} to the analysis of asymptotic efficiency for sequences of quantum estimators.

\subsection{Quantum Hodges estimator}

In order to motivate ourselves to study asymptotic efficiency in the quantum domain, let us touch upon the issue of quantum superefficiency first.
In classical statistics, there was a well-known superefficient estimator called the Hodges estimator that asymptotically breaks the Cram\'er-Rao bound \cite{vaart}. 
An analogous estimator can be constructed in the quantum domain that asymptotically breaks the Holevo bound. 

Let us consider the pure state model 
\[
\S=\left\{ \rho_{\theta}=\frac{1}{2}\left(I+\theta^{1}\sigma_{1}+\theta^{2}\sigma_{2}
+\sqrt{1-\left(\theta^1\right)^{2}-\left(\theta^2\right)^{2}}\, \sigma_{3}\right)
: \theta\in\R^{2},\,\left(\theta^1\right)^{2}+\left(\theta^2\right)^{2}<1\right\},
\]
on $\H=\C^2$ having two-dimensional parameter $\theta=(\theta^1,\theta^2)$, where $\sigma_{1},\sigma_{2},\sigma_{3}$ are the Pauli matrices. 
It is well known \cite{holevo} that the weighted trace of the covariant matrix $V_{\theta}[M,\hat{\theta}]$ for a locally unbiased estimator $(M,\hat{\theta})$ with a weight matrix $G >0$ is bounded from below by the Holevo bound $c_G^{(H)}$ as 
\[
\Tr G V_\theta[M,\hat{\theta}]\geq c_G^{(H)}.
\]
If we set $G$ to be the SLD Fisher information matrix $J_{\theta}^{(S)}$, the Holevo bound $c_{J_{\theta}^{(S)}}^{(H)}$ is reduced to $4$, which is independent of $\theta$, and is achieved when and only when $V_{\theta}[M,\hat{\theta}]=(J_\theta^{(S)}/2)^{-1}$; specifically, it is achievable by a randomized measurement scheme without invoking any collective measurement \cite{YamagataTomo}. 

Now we construct a sequence of estimators that asymptotically breaks the Holevo bound. 
It is known that for the i.i.d.~model $\S^{(n)}=\{ \rho_{\theta}^{\otimes n}\} _{\theta}$, there is an adaptive estimation scheme $(\hat{M}^{(n)},\hat{\theta}^{(n)})$ in which $\sqrt{n}(\hat{\theta}^{(n)}-\theta)$ weakly converges to the (classical) normal distribution $N(0, (J_\theta^{(S)}/2)^{-1})$ for every $\theta$ \cite{Fujiwara:2006}.
Introduce a second estimator $T_n$ by
\begin{equation}\label{eq:q_hodges}
 T_n:=\begin{cases}
  \hat{\theta}^{(n)} & \text{if }\text{\ensuremath{\|\hat{\theta}^{(n)}\|}\ensuremath{\ensuremath{\geq1/\sqrt[4]{n}}}}\\
0 & \text{if }\text{\ensuremath{\|\hat{\theta}^{(n)}\|}\ensuremath{<\ensuremath{1/\sqrt[4]{n}}}}
\end{cases}.
\end{equation}
Then $\sqrt{n}(T_n-\theta)$ converges to $N(0, (J_\theta^{(S)}/2)^{-1})$ in distribution if $\theta\neq 0$, whereas it converges to $0$ in probability if $\theta=0$. 
At first sight, $T_n$ is an improvement on $\hat\theta^{(n)}$. 
However, as demonstrated below, this reasoning is a bad use of asymptotics \cite{vaart}.

In order to evaluate the asymptotic behavior of $\sqrt{n}(T_n-\theta)$ in more detail, we assume the following situation: 
through the first stage of estimation, the adaptive measurement $\hat M^{(k)}$ has converged to a measurement $M_\theta$ that is optimal at the true value of $\theta$ \cite{Fujiwara:2006}. 
Now we proceed to the second stage:  
fix the measurement to be the one that has been obtained through the first stage, i.e. $M_\theta$,
and take $\hat\theta^{(n)}$ to be the sample average of outcomes over $n$-i.i.d.~experiments, each being distributed as $N(\theta, V_\theta)$, where $V_\theta=(J_\theta^{(S)}/2)^{-1}$, so that
$
 \hat\theta^{(n)}\sim N(\theta, {V_\theta}/{n} ).
$
Under this situation, the weighted trace $\Tr J_\theta^{(S)} V_\theta[M_\theta, T_n]$ of covariance matrix of the quantum Hodges estimator $T_n$ can be evaluated as follows. 
Because of the rotational symmetry of the model $\S$ around the origin of the parameter space, we can assume without loss of generality that the true parameter $\theta$ lies on the plane $\theta^2=0$.
In this case, 
\[
 \Tr J_\theta^{(S)} V_\theta[M_\theta, T_n]
 =\int_0^{2\pi} d\phi \int_{1/\sqrt[4]{n}}^\infty w_\theta(r,\phi)  q_\theta(r,\phi) r dr
  +\int_0^{2\pi} d\phi \int_0^{1/\sqrt[4]{n}} w_\theta(0,\phi) q_\theta(r,\phi) r dr,
\]
where
\[
 w_\theta(r,\phi):=\frac{(r\cos\phi-\theta^1)^2}{1-(\theta^1)^2}+r^2\sin^2\phi
\]
is the weighted sum of squared errors, and
\[
 q_\theta(r,\phi) r dr d\phi:=\frac{n}{4\pi\sqrt{1-(\theta^1)^2}} \exp \left[-\frac{n}{4}\, w_\theta(r,\phi) \right] r dr d\phi
\]
is the probability density of $\hat\theta^{(n)}-\theta \sim N(0,V_\theta/n)$ in the polar coordinate system. 

\begin{figure}
\begin{center}
\includegraphics[keepaspectratio, scale=1]{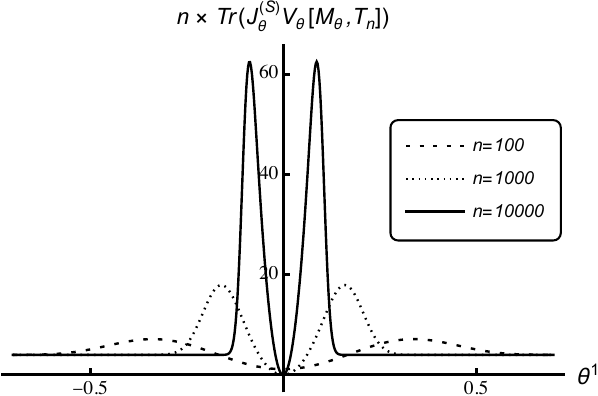}
\caption{Weighted trace of covariance matrix of the quantum Hodges estimator $T_n$ with the weight $J_\theta^{(S)}$ for the spin coherent state model $\S$ based on the means of samples of size 100 (dashed), 1000 (dotted), and 10000 (solid) observations. For reference, the corresponding Holevo bound is $c_{J_{\theta}^{(S)}}^{(H)}=4$.}
\label{fig:hodge}
\end{center}
\end{figure}

Figure \ref{fig:hodge} shows the graph of $n\times \Tr J_\theta^{(S)} V_\theta[M_\theta, T_n]$ for three different values of $n$. 
These functions are close to the Holevo bound $c_{J_{\theta}^{(S)}}^{(H)}=4$ on most of the domain but possess peaks close to zero. 
As $n\to\infty$, the location and widths of the peaks converge to zero but their heights to infinity. 
Because the values of $\theta$ at which $T_n$ behaves badly differ from $n$ to $n$, the pathological behavior of $\sqrt{n}(T_n-\theta)$ is not visible in the pointwise limit distributions under fixed $\theta$ as in the classical case \cite{vaart}.

\subsection{Quantum regular estimator}

In classical statistics, it is customary to restrict ourselves to a certain class of estimators 
in order to avoid pathological behavior like the Hodges estimator.
In this section, we shall extend such a strategy to the quantum domain. 

We begin with a standard estimation problem for a quantum Gaussian shift model. 
Our problem is to estimate the parameter $h\in\R^d$ of the quantum Gaussian shift model 
$\phi_h\sim N(\left({\rm Re}\, \tau\right)h,\Sigma)$, where $\Sigma$ is an $r\times r$ complex nonnegative matrix ($r\ge d$) with ${\rm Re}\,\Sigma>0$, and $\tau$ is an $r\times d$ complex matrix with  
${\rm rank} \left({\rm Re}\,\tau\right)=d$.

An estimator for the model $\phi_h$ is represented by a POVM $M$ over $\R^{d}$.
For each $h\in\R^d$, let $M-h$ denotes the shifted POVM in which the outcome $x$ of $M$ is transformed into $y=x-h$.
It is formally defined by
\[ \int_B f(y) \,\phi_h((M-h)(dy)) := \int_{B+h} f(x-h)\, \phi_h (M(dx))  \qquad (\forall B\in\B(\R^d)). \]
An estimator $M$ for $\phi_h$ is called {\em equivalent in law} if the probability distribution of the outcomes of the shifted POVM $M-h$ applied to $\phi_h$ is independent of $h\in\R^d$, in that
\[
 \phi_h((M-h)(B))=\phi_0(M(B))
\]
holds for all $B\in\B(\R^d)$.
The following result is standard.

\begin{lemma}\label{lem:equi_law}
Assume that an estimator $M$ for the shift parameter $h$ of a quantum Gaussian shift model 
$\phi_h \sim N(\left({\rm Re}\,\tau\right)h,\Sigma)$ is equivalent in law. 
Then, for any $d\times d$ weight matrix $G>0$, 
\[
  \int_{\R^{d}} G_{ij} (x-h)^{i}(x-h)^{j} \phi_{h} (M(dx))\geq c_{G}^{(H)},
\]
where $c_{G}^{(H)}$ is the Holevo bound. 
\end{lemma}

\begin{proof}
%See Section E of Supplementary Material \cite{qrep_supp}. 
See Appendix \ref{app:ProofApplications}.
\end{proof}

Now we introduce the notion of regular estimators%
\footnote{
In classical statistics, regularity is also called {\em asymptotically equivalent in law}. 
}
for q-LAN models. 
Suppose that we are given a sequence 
$\S^{(n)}=\{ \rho_{\theta}^{(n)} : \theta\in\Theta\subset\R^{d}\}$
of quantum statistical models that is q-LAN at $\theta_0\in\Theta$. 
A sequence $M^{(n)}$ of POVMs is called {\em regular at $\theta_0\in\Theta$} 
if the classical distribution $\L^{(n)h}$ of outcomes of the shifted POVM
\[
 M^{(n)h}:=\sqrt{n}\left\{ M^{(n)}-\left(\theta_{0}+h/\sqrt{n}\right)\right\}
\]
under $\rho_{\theta_{0}+h/\sqrt{n}}^{(n)}$ converges to a classical distribution $\L$ that is independent of $h$: 
\begin{equation}\label{eq:regular}
 \L^{(n)h}\conv{}\L\qquad(\forall h\in\R^{d})
\end{equation}

Note that $M^{(n)h}$ is a measurement in which the outcome $\hat{\theta}\in\R^{d}$ of $M^{(n)}$ is transformed into $\sqrt{n}\{ \hat{\theta}-\left(\theta_{0}+h/\sqrt{n}\right)\}$. 
Since 
\[
 \sqrt{n}\left\{ \hat{\theta}-\left(\theta_{0}+h/\sqrt{n}\right)\right\} \leq t
 \quad \Longleftrightarrow \quad
 \hat{\theta}\leq\theta_{0}+(h+t)/\sqrt{n},
\]
we see that
\[
 M^{(n)h}(-\infty,t]=M^{(n)} \left(-\infty,\theta_{0}+\frac{h+t}{\sqrt{n}} \right].
\]

When a sequence $\S^{(n)}$ of quantum statistical models is q-LAN and $D$-extendible at $\theta_{0}\in\Theta$, 
the next theorem is an immediate consequence of the asymptotic representation Theorem \ref{thm:qRep} and Lemma \ref{lem:equi_law}. 

\begin{theorem}[Bound for quantum regular estimator]\label{thm:regular}
Let $\S^{(n)}=\{ \rho_{\theta}^{(n)} : \theta\in\Theta\subset\R^{d}\}$ be a sequence of quantum statistical models that is q-LAN and $D$-extendible at $\theta_{0}\in\Theta$. 
For any estimator $M^{(n)}$ that is regular at $\theta_0$, and a $d\times d$ weight matrix $G>0$,
we have
\begin{equation}\label{eq:regular_bound}
 \int_{\R^{d}}G_{ij} x^{i}x^{j}\L(dx)\geq c_{G}^{(rep)},
\end{equation}
and hence
\begin{equation}\label{eq:regular_bound2}
 \liminf_{n\to\infty} \int_{\R^d} G_{ij} (x-h)^i (x-h)^j \;\Tr\rho_{\theta_{0}+h/\sqrt{n}}^{(n)} M^{(n)}(dx)
  \ge c_G^{(rep)},
\end{equation}
where 
$\L$ is the limit distribution of $M^{(n) h}$ under $\rho_{\theta_{0}+h/\sqrt{n}}^{(n)}$, and 
$c_{G}^{(rep)}$ is the asymptotic representation bound defined by \eqref{eq:rep_bound}. 
\end{theorem}

\begin{proof}
%See Section E of Supplementary Material \cite{qrep_supp}. 
See Appendix \ref{app:ProofApplications}.
\end{proof}

It is natural to inquire whether there exists a regular estimator $M^{(n)}$ that achieves the lower bound $c_{G}^{(rep)}$ in Theorem \ref{thm:regular}.
The answer is given by the following

\begin{theorem}[Achievability of asymptotic representation bound]\label{thm:achieve}
Assume that a quantum statistical model 
$\S^{(n)}=\{ \rho_{\theta}^{(n)} : \theta\in\Theta\subset\R^{d}\}$
is q-LAN and $D$-extendible at $\theta_0\in\Theta$. 
Given a $d\times d$ weight matrix $G>0$, there exist a regular estimator $M_\star^{(n)}$ and a $d\times d$ real strictly positive matrix $V_\star$ such that
\[
\left( M_\star^{(n)},\, \rho_{\theta_{0}+h/\sqrt{n}}^{(n)}\right)\conv hN(h,V_\star)
\]
and
\[
\Tr GV_\star=c_{G}^{(rep)}
\]
for all $h\in\R^{d}$. 
\end{theorem}

\begin{proof}
%See Section E of Supplementary Material \cite{qrep_supp}. 
See Appendix \ref{app:ProofApplications}.
\end{proof}

Theorem \ref{thm:achieve} implies that the asymptotic representation bound $c_{G}^{(rep)}$ is achievable, 
in that
\begin{align*}
   \sup_{L>0} \lim_{n\to\infty} \int_{\R^d}  L\wedge \left\{ G_{ij} (x-h)^i (x-h)^j \right\} \;\Tr\rho_{\theta_{0}+h/\sqrt{n}}^{(n)} M_\star^{(n)}(dx) 
   = c_G^{(rep)}.
\end{align*}
Moreover, in combination with Theorem \ref{thm:regular}, Theorem \ref{thm:achieve} tells us that the bound $c_{G}^{(rep)}$ gives the ultimate limit of estimation precision.
This fact has the following important consequence: 
since an achievable `scalar' lower bound for an estimation problem is necessarily unique, the bound $c_{G}^{(rep)}$ is uniquely determined. 
More precisely, we have the following

\begin{corollary}[Well-definedness of asymptotic representation bound]\label{cor:rep_bound}
For each $d\times d$ weight matrix $G>0$, the asymptotic representation bound $c_{G}^{(rep)}$ is independent of the choice of a $D$-extension. 
\end{corollary}

It should be emphasized here that Theorem \ref{thm:achieve} is valid for {\em all} $h\in\R^d$. 
This is a remarkable refinement of the former result \cite[Theorem 3.1]{qlan_first}, in which the Holevo bound $c_{G}^{(H)}$ for an i.i.d.~model was  achieved only on a {\em countable dense subset} of $\R^d$.

\subsection{Quantum minimax theorem}\label{sec:minimax}

We can also study efficiency in terms of minimax criteria. 
Let us begin with a minimax theorem for a quantum Gaussian shift model.

\begin{theorem}[Minimax theorem for quantum Gaussian shift model]\label{thm:q_gauss_minimax}
Suppose that we are given a quantum Gaussian shift model 
$\phi_h \sim N(\left({\rm Re}\,\tau\right)h,\Sigma)$. 
Then, for any estimator $M$ and a weight matrix $G>0$, 
\[
\sup_{h\in\R^{d}}\int_{\R^{d}}G_{ij}(x-h)^{i}(x-h)^{j}\phi_{h}(M(dx))\geq c_{G}^{(H)}.
\]
\end{theorem}

\begin{proof}
%See Section E of Supplementary Material \cite{qrep_supp}. 
See Appendix \ref{app:ProofApplications}.
\end{proof}

By using the asymptotic representation Theorem \ref{thm:qRep} as well as Theorem \ref{thm:q_gauss_minimax}, we can prove the following

\begin{theorem}[Local asymptotic minimax theorem]\label{thm:minimax_local}
Let $\S^{(n)}=\{ \rho_{\theta}^{(n)} : \theta\in\Theta\subset\R^{d}\}$
be a sequence of quantum statistical models that is q-LAN and $D$-extendible at $\theta_{0}\in\Theta$. 
Then, for any sequence $M^{(n)}$ of estimators
and $d\times d$ weight matrix $G>0$,
\begin{align}\label{eq:minimax1}
 & \lim_{\delta\to0}\liminf_{n\to\infty}\sup_{\left\Vert h\right\Vert \leq\delta\sqrt{n}}\int_{\R^{d}}G_{ij}(x-h)^{i}(x-h)^{j}
 	\; \Tr \rho_{\theta_{0}+h/\sqrt{n}}^{(n)}M^{(n)}(dx)  \\
 & \geq \sup_{H}
 	\liminf_{n\to\infty}\sup_{h\in H}\int_{\R^{d}}G_{ij}(x-h)^{i}(x-h)^{j}
 	\; \Tr \rho_{\theta_{0}+h/\sqrt{n}}^{(n)}M^{(n)}(dx)
 	\nonumber \\ 
 & \geq\sup_{L>0}\sup_{H}\liminf_{n\to\infty}\sup_{h\in H}\int_{\R^{d}} L\wedge
 	\left\{G_{ij}(x-h)^{i}(x-h)^{j}\right\} 
 	\; \Tr \rho_{\theta_{0}+h/\sqrt{n}}^{(n)}M^{(n)}(dx) 
 	\nonumber \\ 
 & \geq c_G^{(rep)}. \nonumber 
\end{align}
Here, $a \wedge b:=\min\{a,b\}$, and $H$  runs over all finite subsets of $\R^d$. 
Moreover, the last inequality is tight.
\end{theorem}

\begin{proof}
%See Section E of Supplementary Material \cite{qrep_supp}. 
See Appendix \ref{app:ProofApplications}.
\end{proof}

Note that the quantities appeared in the first and second lines of \eqref{eq:minimax1} 
correspond to the minimax theorems due to H\'ayak \cite{HajekMinimax} and in van der Vaart's book \cite{vaart}, respectively.

\subsection{Quantum James-Stein estimator}\label{sec:stein}

As the final topic of this section, we touch upon a superefficient estimator that uniformly breaks the asymptotic representation bound $c_G^{(rep)}$. 

Let us consider the i.i.d.~quantum statistical model $\S^{(n)}:=\{ \rho_{\theta}^{\otimes n}\}$ with the base model
\[
\rho_{\theta}=\frac{1}{2}\left(I+\theta^{1}\sigma_{1}+\theta^{2}\sigma_{2}+\theta^{3}\sigma_{3}\right),
\qquad \left( \theta=(\theta^1,\theta^2, \theta^3), \; \| \theta\|^2< 1 \right)
\]
on $\H=\C^{2}$.
We see from 
%Section C of Supplementary Material \cite{qrep_supp} 
Appendix \ref{app:D-extension}
that $\S^{(n)}$ is q-LAN and $D$-extendible at every point $\theta$.
In fact, since the linear span of SLDs at each $\theta$ is $\D_{\rho_\theta}$-invariant, the set of SLDs itself gives a $D$-extension. 

Here we focus our attention on the local asymptotic estimation at around the origin $\theta=0$. 
The SLDs at $\theta=0$ are $\sigma_{i}$ ($i=1,2,3$), and the corresponding SLD Fisher information matrix $J^S$ is the identity matrix. 
Let
\[
\Delta_{i}^{(n)}:=\frac{1}{\sqrt{n}}\sum_{k=1}^{n}I^{\otimes(k-1)}\otimes\sigma_{i}\otimes I^{\otimes(n-k)}.
\]
Then the asymptotic representation Theorem \ref{thm:qRep} allows us to convert the problem of estimating the local parameter $h$ of $\rho_{h/\sqrt{n}}^{\otimes n}$ into that of estimating the shift parameter $h$ of the limiting (classical) Gaussian shift model
\begin{equation}\label{eq:stein_lim_model}
\left\{ N(h,I):h\in\R^{3}\right\}.
\end{equation}
Specifically, for any regular POVM $M^{(n)}$ that satisfies
\[
\left(M^{(n)},\rho_{h/\sqrt{n}}^{\otimes n}\right)\conv{} \exists \L_{h},
\]
we see from Theorem \ref{thm:regular} and \eqref{eq:rep_bound} that 
\[
\int_{\R^{d}}\left\Vert x-h\right\Vert ^{2}\L_{h}(dx)\geq c_{I}^{(rep)}=\Tr I=3,
\]
where we have taken the weight $G$ to be the SLD Fisher information matrix $J^S=I$. 

Now we demonstrate that, if one discards the requirement of regularity, one can construct an estimator that breaks the above inequality for all $h$. 
An estimator on the classical Gaussian shift model \eqref{eq:stein_lim_model} that changes observed data $x\in\R^3$ into 
\begin{equation}\label{eqn:Stein}
 y=\left(1-\frac{1}{\left\Vert x\right\Vert }\right)x
\end{equation}
is called the James-Stein estimator \cite{JudgeBock}. 
Letting $\L_{h}^{(JS)}$ be the corresponding probability distribution of $y$, it is well known that
\[
 \int_{\R^{d}}\left\Vert y-h\right\Vert ^{2}\L_{h}^{(JS)}(dy)<3
\]
for all $h\in\R^3$. 
Now we see from Theorem \ref{thm:achieve} that there is a regular POVM $N^{(n)}$ that exhibits
\[
 \left(N^{(n)},\rho_{h/\sqrt{n}}^{\otimes n}\right)\conv{}N(h,I). 
\]
For each $n$, let $N^{(JS,\,n)}$ be a POVM that changes the outcome $x\in\R^3$ of $N^{(n)}$ into $y\in\R^3$ as \eqref{eqn:Stein}.
Then 
\[
 \left(N^{(JS,\,n)},\rho_{h/\sqrt{n}}^{\otimes n}\right)\conv{}\L_{h}^{(JS)}, 
\]
and thus $N^{(JS,\,n)}$ asymptotically breaks the asymptotic representation bound $c_{I}^{(rep)}$ for all $h\in\R^{3}$.

%----------------------------------------------------------------------------------------------------------------------------------------------------------------
\section{Conclusions} \label{sec:conclusions}
%----------------------------------------------------------------------------------------------------------------------------------------------------------------

In this paper, we derived a noncommutative analogue of asymptotic representation theorem for a $D$-extendible q-LAN model (Theorem \ref{thm:qRep}). 
This theorem converts an estimation problem for a local model $\{\rho^{(n)}_{\theta_0+h/\sqrt{n}} : h \in \R^d \}$ into another for the limiting quantum Gaussian shift model $\{ N( ({\rm Re}\, \tau) h, \Sigma) : h \in \R^d \}$. 
As a corollary, we arrived at a new bound $c_{G}^{(rep)}$ defined by the Holevo bound for the limiting model.
This bound turned out to have universal  importance in asymptotic quantum statistics. 
For example, it gave the ultimate limit of estimation precision for regular estimators (Theorems \ref{thm:regular} and \ref{thm:achieve}) and minimax estimators (Theorem \ref{thm:minimax_local}). 
Note that, since the bound $c_{G}^{(rep)}$ for an i.i.d.~model 
is reduced to the standard Holevo bound for the base model, 
the achievability theorem (Theorem \ref{thm:achieve}) gives a substantial refinement of the former result \cite[Theorem 3.1]{qlan_first} in which the Holevo bound was achieved only on a countable dense subset of the parameter space.
 
The key ingredient of Theorem \ref{thm:qRep} was the notion of $D$-extendibility. 
Its importance is first realized in the present paper; however, its trace can be found elsewhere. 
Gu\c{t}\u{a} and Kahn \cite{{guta_qubit}, {guta_qudit}} and Lahiry and Nussbaum \cite{low_rank} treated  i.i.d.~extensions of a quantum statistical model that has $\D_{\rho_\theta}$-invariant SLD-tangent space at every $\theta$ from the outset, and thus they did not need to care about the $D$-extendibility (Remark \ref{rem:DinvExt}). 
In their framework, the difficulty demonstrated in Example \ref{eg:noLimitPOVM} is automatically avoided by regarding the model as a submodel of its ambient full pure state model.
Yamagata {\em et al.} \cite{qlan_first} introduced the notion of joint q-LAN for $(X^{(n)}, \Delta^{(n)})$.
In view of the present paper, this was a forerunner of the $D$-extension $X^{(n)}$ of SLDs $\Delta^{(n)}$, whereby the achievability of the Holevo bound was proved. 
The notion of $D$-extendibility made it possible to generalize the Holevo bound to non-i.i.d.~models, providing a  proper perspective on the achievability of the asymptotic representation bound $c_{G}^{(rep)}$. 

We believe that the paper has established a solid foundation of the theory of (weak) quantum local asymptotic normality.
Nevertheless, its application has just begun, and many open problems are left for future study.
For example, it is not clear whether the $D$-extendibility condition can be replaced with a weaker one. 
One would convince oneself that there are quantum statistical models that are not i.i.d.~but are, nevertheless, q-LAN and $D$-extendible. 
Imagine a convergent sequence $\sigma_\theta^{(n)}\to\sigma_\theta^{(\infty)}$ of quantum statistical models on a fixed finite dimensional Hilbert space $\H$.
Then the tensor product models $\rho_\theta^{(n)}:=\bigotimes_{k=1}^n \sigma_\theta^{(k)}$ would  be q-LAN and $D$-extendible because they are `almost' i.i.d.~in the asymptotic limit. 
In fact, it is not difficult to realize this idea with some additional conditions
%(cf., Section C of Supplementary Material \cite{qrep_supp}). 
(cf., Appendix \ref{app:D-extension}).
In this way, the $D$-extendibility condition enables us to study quantum asymptotics beyond the i.i.d.~assumption. 
In view of applications, however, it would be nice if there were a more tractable weaker condition that establishes an asymptotic representation theorem.

It also remains to be investigated whether an asymptotically optimal statistical procedure for the local model indexed by the parameter $\theta_0+h/\sqrt{n}$ can be translated into useful statistical procedures for the real world case in which $\theta_0$ is unknown. 
Some authors \cite{GillMassar}, advocated two-step estimation procedures, in which one first measures a small portion of the quantum system, in number $n_1$ say, using some standard measurement scheme and constructs an initial estimate, say $\tilde\theta_1$, of the parameter.
One next applies the theory of q-LAN to compute the asymptotically optimal measurement scheme which corresponds to the situation $\theta_0=\tilde \theta_1$, 
and then proceeds to implement this measurement on the remaining $n_2\, (:=n-n_1)$ quantum systems collectively, estimating $h$ in the model $\theta=\tilde \theta_1+h/\sqrt{n_2}$. 
However such procedures are inherently limited to within the scope of weak consistency. 
Studying the strong consistency and asymptotic efficiency \cite{Fujiwara:2006} 
in the framework of collective quantum estimation scheme is an important open problem.

%----------------------------------------------------------------------------------------------------------------------------------
\section*{Acknowledgments}
%----------------------------------------------------------------------------------------------------------------------------------

%The present study was supported by JSPS KAKENHI Grant Numbers JP22340019, and JP17H02861.
%The present study was supported by JSPS KAKENHI Grants JP17H02861 and JP22K03466.

The present study was supported by JSPS KAKENHI Grant Numbers JP17H02861, JP22K03466, 
and MEXT Quantum Leap Flagship Program (MEXT Q-LEAP) Grant Number JPMXS0120351339.

\appendix

\section*{Appendix}

%----------------------------------------------------------------------------------------------------------------------------------------------------------------
\section{{ Asymptotic representation theorem for classical LAN}}\label{suppl:clan}
%----------------------------------------------------------------------------------------------------------------------------------------------------------------

This section gives a comprehensible proof of the asymptotic representation theorem for classical LAN models (Theorem \ref{thm:crep}). This also provides an alternative view for the `randomized' statistics appeared in the theorem. 

In constructing a statistic $T$ that enjoys $T^{(n)}\convd{h} T$ for all $h$, 
van der Vaart \cite{vaart} emphasized that one must invoke external information.
This prescription reminds us of a quantum POVM in which one makes use of an ancillary system in realizing it. 
In what follows, therefore, 
we identify the randomized statistic $T$ with a $\sigma$-finite measure on $\R^s\times\R^d$ that gives the desired limit distribution $\mathcal{L}_h$ for every $h\in\R^d$.

\begin{proof} 
For each $t\in \R^s$, let
\begin{equation}
 M^{(n)}(t, \omega):=
\ind _{{T^{(n)}}^{-1}( (-\infty,t])} (\omega),\qquad(n\in\N,\;\omega\in\Omega^{(n)}).
\end{equation}
Referring to the diagram
\begin{eqnarray*}
 \array{ccccc}
     & & \Omega^{(n)} & \displaystyle \mathop{\longrightarrow}^{\Delta^{(n)}}& \R^d  \\
     \mbox{[0,\,1]} & \displaystyle\mathop{\longleftarrow}^{P^{(n)}_{\theta_0}} & \mathcal{F}^{(n)}  & \displaystyle \mathop{\longleftarrow}^{{\Delta^{(n)}}^{-1}} &\mathcal{B}(\R^d)
 \endarray,
\end{eqnarray*}
we define, for each $t\in\R^s$, a finite Borel measure $\mu^{(n)}_t$ on $\R^d$ as follows:
\begin{equation}\label{muNT}
  \mu^{(n)}_t(B):=\int_{{\Delta^{(n)}}^{-1}(B)} M^{(n)}(t,\omega) dP^{(n)}_{\theta_0}(\omega),\qquad (B\in \mathcal{B}(\R^d) ).
\end{equation}
Note that the set $\{\mu^{(n)}_t\}_n$ is tight. 
In fact, since $\Delta^{(n)}\convd{0} N(0,J)$, the sequence $\Delta^{(n)}$ is tight under $P^{(n)}_{\theta_0}$, that is, for any $\varepsilon>0$, there exists a $K>0$ such that for all $n$,
\begin{equation*}
 P^{(n)}_{\theta_0}( \Delta^{(n)} \notin [-K,K]^d) < \varepsilon.
\end{equation*}
Consequently,
\begin{eqnarray*}
 \mu^{(n)}_t(\R^d\backslash [-K,K]^d) 
 &=&\int_{\Delta^{(n)} \notin [-K,K]^d} M^{(n)}(t,\omega) dP^{(n)}_{\theta_0}(\omega) \nonumber \\
 &\le& \int_{\Delta^{(n)} \notin [-K,K]^d} dP^{(n)}_{\theta_0}(\omega) \nonumber \\
 &=&P^{(n)}_{\theta_0}(\Delta^{(n)} \notin [-K,K]^d) < \varepsilon, \nonumber
\end{eqnarray*}
proving the tightness of $\{\mu^{(n)}_t\}_n$. 
It then follows from the Prohorov theorem that there is a subsequence $\{\mu^{(n_k)}_t\}_k$ that is weakly convergent for all $t\in\Q^s$, i.e.,
\begin{equation}\label{mu}
  \mu^{(n_k)}_t \convd{0} \,{}^\exists \mu_t,\qquad (\forall t\in\Q^s).
\end{equation}
Observe that for any continuity point $x\in\R^d$ of $\mu_t$, 
\begin{eqnarray*} 
 \mu_t((-\infty,x])
 &=&\lim_{k\to\infty} \int_{\Delta^{(n_k)}\le x} M^{(n_k)}(t,\omega) dP^{(n_k)}_{\theta_0}(\omega) \\
 &=&\lim_{k\to\infty} P^{(n_k)}_{\theta_0}(\{T^{(n_k)}\le t\}\cap\{\Delta^{(n_k)}\le x\}). \nonumber
\end{eqnarray*}

Let us extend $\mu_t$ to all $t\in\R^s$ so that $\mu_t ((-\infty,x])$ is right-continuous in $t$ for each $x\in\R^d$, and denote the extension by $\overline\mu_t$, that is,
\begin{equation}\label{muBar}
  \overline\mu_t((-\infty,x]):=\inf\{\mu_\alpha((-\infty,x])\,|\,\alpha\in\Q^s,\,  \alpha > t\}.
\end{equation}
Specifically, since $T^{(n)}\convd{0} \mathcal{L}_0$, the total mass of $\overline\mu_t$
for a continuity point $t$ of $\mathcal{L}_0$ is given by
\begin{eqnarray*}
 \overline\mu_t(\R^d)=\mu_t(\R^d) 
 =\lim_{k\to\infty} P_{\theta_0}^{(n_k)}(T^{(n_k)}\le t) 
 =\mathcal{L}_0((-\infty,t]).
\end{eqnarray*}
Further, since $\Delta^{(n)}\convd{0} N(0,J)$, we have from the joint tightness of $(\Delta^{(n)}, T^{(n)})$ that
\begin{eqnarray*}
\overline\mu_\infty(B)
 :=\lim_{t\to\infty}\mu_t(B)
 =\lim_{k\to\infty} P^{(n_k)}_{\theta_0}(\Delta^{(n_k)}\in B) 
 =\int_B g_0(x) dx. 
\end{eqnarray*}
Here, $g_h(x)$ denotes the density of $N(Jh,J)$ with respect to the Lebesgue measure $dx$. 
Put differently, $\overline\mu_\infty\sim N(0,J)$. 

Since $\overline\mu_t(B)\le \overline\mu_\infty(B)$ for all $t\in\R^s$, we find that $\overline\mu_t$ is absolutely continuous to $\overline\mu_\infty$, and hence to the Lebesgue measure. 
This guarantees the existence of the density
\begin{equation}\label{MT}
 M_t(x):=\frac{d \overline\mu_t}{d \overline\mu_\infty}(x)=\frac{1}{g_0(x)} \frac{d \overline\mu_t}{d x}(x).
\end{equation}
Note that $0\le M_t(x)\le 1$ and $M_t(x)\,\uparrow\, M_\infty(x)=1$ for each $x\in\R^d$.

We prove that this is the one that gives the desired limit distribution, in that 
\begin{eqnarray}  \label{conclusion}
\mathcal{L}_h((-\infty,t])
 =\int_{\R^d} g_h(x)\, M_t(x)\, dx 
\end{eqnarray}
for any $h\in\R^d$ and any continuity point $t\in\R^s$ of $\mathcal{L}_h$.

Because of \eqref{eq:lan},  we have
\begin{equation} \label{contiguous}
 P_{\theta_0+h/\sqrt{n}}^{(n)} \vartriangleleft \vartriangleright P_{\theta_0}^{(n)},
\end{equation}
which, in particular, entails that 
$\displaystyle \frac{dP_{\theta_0+h/\sqrt{n}}^{(n)}}{dP_{\theta_0}^{(n)}} $ is uniformly integrable under $P_{\theta_0}^{(n)}$, and hence under $\mu_t^{(n)}$ for any $t\in\R^s$.
Consequently, for any continuity point $t\in\R^s$ of $\mathcal{L}_h$, we have
\begin{eqnarray} \label{mainConv0}
\mathcal{L}_h((-\infty,t])
&=&
 \lim_{k\to\infty} \int_{\Omega^{(n_k)} } M^{(n_k)}(t,\omega)\, dP_{\theta_0+h/\sqrt{n_k}}^{(n_k)}(\omega) \\
 &=&\lim_{k\to\infty} \int_{\Omega^{(n_k)} } M^{(n_k)}(t,\omega)\, 
     \frac{dP_{\theta_0+h/\sqrt{n_k}}^{(n_k)}}{dP_{\theta_0}^{(n_k)}}\,{dP_{\theta_0}^{(n_k)}}(\omega) \nonumber  \\
 &=&
 \lim_{k\to\infty} \int_{\R^d }
     e^{h^i \Delta_i^{(n_k)}-\frac{1}{2}h^i h^j J_{ij}}\,{d\mu_t^{(n_k)}}(\Delta^{(n_k)}). \nonumber
\end{eqnarray}
Here, the first equality follows from the assumption that $T^{(n)}\convd{h} \;{}^\exists\mathcal{L}_h$, the second from \eqref{contiguous} and Lemma \ref{lem:contiguity} below, and the third from \eqref{eq:lan} and \eqref{muNT}.
Now we prove that
\begin{equation}\label{mainConv}
 \lim_{k\to\infty} \int_{\R^d } e^{h^i \Delta_i^{(n_k)}-\frac{1}{2}h^i h^j J_{ij}}\,{d\mu_t^{(n_k)}}(\Delta^{(n_k)}) 
=\int_{\R^d} \frac{g_h(x)}{g_0(x)}\, d \overline\mu_t(x)
\end{equation}
for any continuity point $t\in\R^s$ of $\mathcal{L}_h$.
Given $\varepsilon >0$ arbitrarily, take another continuity point $t'\in\R^s$ of $\mathcal{L}_h$ and a rational point $\alpha\in\Q^s$ 
$(t<\alpha<t')$ such that 
\begin{equation}\label{L}
 0\le \mathcal{L}_h((-\infty,t']) - \mathcal{L}_h((-\infty,t]) <\varepsilon
\end{equation}
and
\begin{equation}\label{alphaT}
  0\le \int_{\R^d} \frac{g_h(x)}{g_0(x)}\, d\mu_\alpha(x) -  \int_{\R^d} \frac{g_h(x)}{g_0(x)}\, d \overline\mu_t(x) <\varepsilon.
\end{equation}
The existence of such $t'$ is assured by the assumption that $t$ is a continuity point of $\mathcal{L}_h$,  and the existence of such $\alpha\in\Q^s$ by \eqref{muBar} and the monotone convergence theorem. 
Then 
\begin{align} \label{mainConv1}
&
 \left| \int_{\R^d }e^{h^i \Delta_i^{(n_k)}-\frac{1}{2}h^i h^j J_{ij}}\,{d\mu_t^{(n_k)}}(\Delta^{(n_k)})
 - \int_{\R^d} \frac{g_h(x)}{g_0(x)}\, d \overline\mu_t(x) \right| \\
&\qquad
 \le \left| \int_{\R^d }e^{h^i \Delta_i^{(n_k)}-\frac{1}{2}h^i h^j J_{ij}}\,{d\mu_t^{(n_k)}}(\Delta^{(n_k)})
 -\int_{\R^d }e^{h^i \Delta_i^{(n_k)}-\frac{1}{2}h^i h^j J_{ij}}\,{d\mu_\alpha^{(n_k)}}(\Delta^{(n_k)}) \right| 
 \nonumber \\ 
% \label{mainConv1}\\
&\qquad\quad
 + \left| \int_{\R^d }e^{h^i \Delta_i^{(n_k)}-\frac{1}{2}h^i h^j J_{ij}}\,{d\mu_\alpha^{(n_k)}}(\Delta^{(n_k)})
 - \int_{\R^d} \frac{g_h(x)}{g_0(x)}\, d\mu_\alpha(x) \right| \nonumber \\
 % \label{mainConv2}\\
&\qquad\quad
 +\left| \int_{\R^d} \frac{g_h(x)}{g_0(x)}\, d\mu_\alpha(x) 
 - \int_{\R^d} \frac{g_h(x)}{g_0(x)}\, d \overline\mu_t(x)\right|.\nonumber
 % \label{mainConv3}
\end{align}
Firstly, due to \eqref{mainConv0} and \eqref{L}, for sufficiently large $k$, the second line of \eqref{mainConv1} is evaluated from above by
\begin{align*}
 &\left| \int_{\R^d }e^{h^i \Delta_i^{(n_k)}-\frac{1}{2}h^i h^j J_{ij}}\,{d\mu_t^{(n_k)}}(\Delta^{(n_k)})
 -\int_{\R^d }e^{h^i \Delta_i^{(n_k)}-\frac{1}{2}h^i h^j J_{ij}}\,{d\mu_{t'}^{(n_k)}}(\Delta^{(n_k)}) \right|  \\
 &\qquad \le
 \left| \int_{\R^d }e^{h^i \Delta_i^{(n_k)}-\frac{1}{2}h^i h^j J_{ij}}\,{d\mu_t^{(n_k)}}(\Delta^{(n_k)}) - \mathcal{L}_h((-\infty, t]) \right|  \\
&\qquad\qquad
 + \left| \mathcal{L}_h((-\infty, t]) - \mathcal{L}_h((-\infty, t']) \right| \\
&\qquad\qquad
 +  \left|  \mathcal{L}_h((-\infty, t']) 
 	-\int_{\R^d }e^{h^i \Delta_i^{(n_k)}-\frac{1}{2}h^i h^j J_{ij}}\,{d\mu_{t'}^{(n_k)}}(\Delta^{(n_k)}) \right| \\
&\qquad< 3\varepsilon.
\end{align*}
Secondly, due to \eqref{mu} and the Lemma \ref{lem:pseudoL1conv} below, for sufficiently large $k$, the third line of \eqref{mainConv1} gets smaller than $\varepsilon$.
Finally, due to \eqref{alphaT}, the last line of \eqref{mainConv1} is bounded from above by $\varepsilon$.
Putting these evaluations together, we find that 
\begin{equation*}
\left| \int_{\R^d }e^{h^i \Delta_i^{(n_k)}-\frac{1}{2}h^i h^j J_{ij}}\,{d\mu_t^{(n_k)}}(\Delta^{(n_k)})
 - \int_{\R^d} \frac{g_h(x)}{g_0(x)}\, d \overline\mu_t(x) \right|
 < 5\varepsilon,
\end{equation*}
proving \eqref{mainConv}. 

Now that \eqref{mainConv0} and  \eqref{mainConv} are established, the desired identity \eqref{conclusion} follows immediately from \eqref{MT} and the assumption that $T^{(n)}\convd{h} \;{}^\exists\mathcal{L}_h$.
\end{proof}

%\medskip
\begin{lemma}\label{lem:contiguity}
Let the probability measures $P_n$ and $Q_n$ on $\Omega_n$ satisfy $Q_n\vartriangleleft P_n$. 
Then, for any measurable subset $F_n$ of $\Omega_n$, 
\begin{equation}\label{contiguity}
 \lim_{n\to\infty} E_{Q_n}[\ind_{F_n}] = \lim_{n\to\infty} E_{P_n}\left[\ind_{F_n} \frac{dQ_n}{dP_n} \right],
\end{equation}
provided either of the limits exists.
\end{lemma}

\begin{proof}
Let $Q_n=Q_n^{ac}+Q_n^\perp$ be the Lebesgue decomposition with respect to $P_n$, and let $A_n:={\rm supp}\, Q_n^{ac}$. Since $P_n(A_n^c)=0$ for all $n$, we have from $Q_n\vartriangleleft P_n$ that $Q_n(A_n^c)\to 0$. Therefore,
\begin{eqnarray}
 \int_{\Omega_n} \ind_{F_n}(\omega) dQ_n
 &=&\int_{A_n} \ind_{F_n}(\omega) dQ_n+\int_{A_n^c} \ind_{F_n}(\omega) dQ_n \nonumber \\
 &=&\int_{\Omega_n} \ind_{F_n}(\omega) \frac{dQ_n}{dP_n}\,dP_n+Q_n(A_n^c\cap F_n), \nonumber
\end{eqnarray}
from which \eqref{contiguity} immediately follows.
\end{proof}

\begin{lemma}\label{lem:pseudoL1conv}
Let $X_n\in L^1(P_n)$ for all $n$, and let $X\in L^1(P)$. 
Suppose that $\{X_n\}_n$ is uniformly integrable and $X_n\convd{} X$. 
Then
\begin{equation*}
 \lim_{n\to\infty} E_{P_n}[X_n]=E_P[X].
\end{equation*}
\end{lemma}

\begin{proof}
For $K\in [0,\infty)$, define a function $f_K:\R\to [-K,K]$ as follows:
\begin{equation*}
 f_K(x):=\left\{\array{ll} K,& (x>K) \\ x, & (-K\le x\le K) \\ -K, & (x<-K) \endarray\right.
\end{equation*}
Given $\varepsilon>0$, we can choose $K$ so that,  
by the uniform integrability, 
\begin{equation*}
 E_{P_n}[ \left| X_n- f_K(X_n) \right|] \le E_{P_n}[|X_n| \,;\, |X_n|>K]< \frac{\varepsilon}{3},\quad(\forall n)
\end{equation*}
and  
\begin{equation*}
 E_{P}[ \left| X- f_K(X) \right|] 
 \le E_P[|X| \,;\, |X|>K] < \frac{\varepsilon}{3}.
\end{equation*}
Further, since $X_n\convd{} X$, we can choose $n_0\in\N$ such that, for all $n\ge n_0$, 
\begin{equation*}
 \left| E_{P_n}[ f_K(X_n) ]- E_P[f_K(X)] \right| < \frac{\varepsilon}{3}.
\end{equation*}
The triangular inequality therefore implies that, for $n\ge n_0$, 
\begin{equation*}
 \left| E_{P_n}[ X_n ] - E_P[ X ] \right| < \varepsilon,
\end{equation*}
and the proof is complete.
\end{proof}

%----------------------------------------------------------------------------------------------------------------------------------------------------------------
\section{Gaussian states on degenerate CCR algebras}\label{app:degenerate_CCR}
%----------------------------------------------------------------------------------------------------------------------------------------------------------------

This section gives a brief account of degenerate canonical commutation relations (CCR) and hybrid classical/quantum Gaussian states. 

Let $V$ be a real symplectic space with nonsingular symplectic form $\Delta$. 
The unital $\ast$-algebra generated by elements of $V$ satisfying
\[ fg-gf = \sqrt{-1} \,\Delta(f,g),\quad f^*=f,\qquad (\forall f, g\in V) \]
is called the {\em canonical commutation relations (CCR) algebra}. 
There is a distinct, but closely related notion of the CCR. 
Let $\H$ be a separable Hilbert space and let $W: V\to B(\H)$ satisfy the relations
\[
 W(f) W(g)=e^{-\frac{\sqrt{-1}}{2}\,\Delta(f,g)} W(f+g),\quad W(f)^*=W(-f), \qquad (f,g\in V).
\]
These are called the {\em Weyl form} of the CCR. 
Specifically, the above relations imply that $W(f)$ is unitary and $W(0)=1$. 

One would like to represent the CCR by using selfadjoint operators on $\H$.
We first treat the case when $V$ is a two-dimensional symplectic space with symplectic basis $\{e_1, f_1\}$ satisfying $\Delta(e_1, f_1)=1$. 
Then the above relation reduces to 
\[
 W(t e_1) W(s f_1)=e^{-\frac{\sqrt{-1}}{2}\,st } \, W(t e_1+s f_1)=e^{-\sqrt{-1}\,st } \, W(s f_1) W(t e_1). 
\]
Let us regard $U(t):=W(t e_1)$ and $V(s):=W(s f_1)$ as one-parameter unitary groups acting on $\H$. 
By Stone's theorem, there is a one-to-one correspondence between selfadjoint operators and (strongly continuous) one-parameter unitary groups. 
Thus one defines a pair of selfadjoint operators $Q$ and $P$ by 
\[ U(t):=e^{\sqrt{-1}tQ}, \quad V(s)=e^{\sqrt{-1}s P}, \]
which fulfills the Weyl form of the CCR
\[ e^{\sqrt{-1}tQ} e^{\sqrt{-1}s P} =e^{-\sqrt{-1}st}  e^{\sqrt{-1}s P} e^{\sqrt{-1}tQ}. \]
Formally differentiating this identity with respect to $t$ and $s$ at $s=t=0$, one has the Heisenberg form of the CCR
\[  QP-PQ=\sqrt{-1}\,I. \] 
The operators $Q$ and $P$ are called the {\em canonical observables} of the CCR. 

There are variety of choices of Hilbert spaces $\H$ and irreducible representations of canonical observables on $\H$. 
However, according to the Stone-von Neumann theorem, they are unitarily equivalent \cite{holevo}.
Thus one may use any one of them. 
In this paper, we canonically use the Schr\"odinger representation on the Hilbert space $\H=L^2(\R)$. 
Note that the von Neumann algebra generated by $\{e^{\sqrt{-1} t Q}\}$ is $L^\infty(\R)$, and the von Neumann algebra generated by $\{e^{\sqrt{-1} (t Q+ s P)}\}$ is $B(\H)$. 

Extending the above formulation to a generic even-dimensional symplectic space $V$ is standard. 
This also allows us to use a more flexible formulation as follows. 
Given a regular $(2k)\times (2k)$ real skew-symmetric matrix $S=(S_{ij})$, let {\rm CCR}($S$) denote the 
von Neumann algebra generated by 
$\{e^{\sqrt{-1}\,t_1X_1},\dots,e^{\sqrt{-1}\,t_{2k}X_{2k}}\}$ that satisfy the CCR
\[  e^{\sqrt{-1}\, t_i X_i} e^{\sqrt{-1}\, t_j X_j} =e^{\sqrt{-1} \, t_i t_j S_{ij}}  e^{\sqrt{-1} (t_i X_i+ t_j X_j)}, \]
and call $X=(X_1,\dots, X_{2k})$ the canonical observables of the ${\rm CCR}(S)$. 
This is done by just finding a regular matrix $T$ satisfying 
\[
 T^\top S T
 =\frac{1}{2}
  \left[\array{ccccccc}
  0 & -1 &&&&& \\
  1 & 0 &&&&& \\
  && 0 & -1 &&& \\
  && 1 & 0 &&& \\
  &&&&\ddots &&\\
  &&&&& 0 & -1\\
  &&&&& 1 & 0
  \endarray\right]
\]
to obtain a suitable symplectic basis $\{e_i, f_i\}_{1\le i\le k}$ which generates $\{Q_i, P_i\}_{1\le i\le k}$ such that each $X_i$ belongs to an $\R$-linear span of $\{Q_i, P_i\}_{1\le i\le k}$. 

Now we formally extend this formulation to a generic $d\times d$ real skew-symmetric matrix $S=(S_{ij})$ as follows. 
We first find a regular matrix $T$ satisfying 
\[
  T^\top S T
 =\frac{1}{2}
  \left[\array{ccc|ccccccc}
  0 && &&&&&& \\
  & \ddots & &&&&&& \\
  && 0 &&&&&&\\
  \hline
  &&& 0 & -1 &&&&& \\
  &&& 1 & 0 &&&&& \\
  &&& && 0 & -1 &&& \\
  &&& && 1 & 0 &&& \\
  &&& &&&&\ddots &&\\
  &&& &&&&& 0 & -1\\
  &&& &&&&& 1 & 0
  \endarray\right],
\]
to obtain a basis $\{\tilde e_1,\dots, \tilde e_r\}\sqcup \{e_i, f_i\}_{1\le i\le k}$, where $r+2k=d$. 
We then extend $\{\tilde e_1,\dots, \tilde e_r\}$ to $\{\tilde e_i, \tilde f_i\}_{1\le i\le r}$ to form a symplectic basis $\{\tilde e_i, \tilde f_i\}_{1\le i\le r} \sqcup \{e_i, f_i\}_{1\le i\le k}$ of a $2(r+k)$-dimensional symplectic space $V$, 
which defines a von Neumann algebra
$\mathcal{A}$, the canonical observables of which are denoted by $\{\tilde Q_i, \tilde P_i\}_{1\le i\le r} \sqcup \{Q_i, P_i\}_{1\le i\le k}$.
Now we denote  ${\rm CCR}(S)$ to be the von Neumann subalgebra of $\mathcal{A}$ generated by 
\[ \{e^{\sqrt{-1}\, \tilde t_i \tilde Q_i}\}_{1\le i\le r} \sqcup \{e^{\sqrt{-1}\, t_i Q_i}, e^{\sqrt{-1}\, s_i P_i}\}_{1\le i\le k}.  \]

In summary, given a possibly degenerate $d\times d$ real skew-symmetric matrix $S=(S_{ij})$,
let ${\rm CCR}(S)$ denote the algebra generated by the observables $X=(X_1,\dots,X_d)$ that satisfy the following Weyl form of the CCR 
\[ e^{\i\xi^{i}X_{i}}e^{\i\eta^{j}X_{j}}=e^{\i\xi^{\top}S\eta}e^{\i\left(\xi+\eta\right)^{i}X_{i}}\qquad(\xi,\eta\in\R^{d}) \]
which is formally rewritten in the Heisenberg form
\[ \frac{\sqrt{-1}}{2}[X_i,X_j] = S_{ij} \qquad(1\leq i,j\leq d). \]

This formulation is useful in handling hybrid classical/quantum Gaussian states.
Given a possibly degenerate $d\times d$ real skew-symmetric matrix $S=(S_{ij})$, 
a state $\phi$ on ${\rm CCR}(S)$ with the canonical observables $X=(X_1,\dots,X_d)$ 
is called a {\em quantum Gaussian state}, denoted $\phi\sim N(\mu,\Sigma)$, if the characteristic function 
${\cal F}_{\xi}\{\phi\}:=\phi(e^{\sqrt{-1}\xi^{i}X_{i}})$ takes the form
\[ {\cal F}_{\xi}\{\phi\}=e^{\sqrt{-1}\xi^{i}\mu_{i}-\frac{1}{2}\xi^{i}\xi^{j}V_{ij}} \]
where $\xi=(\xi^i)_{i=1}^d\in\R^d$, $\mu=(\mu_i)_{i=1}^d\in\R^d$, and $V=(V_{ij})$ is a $d\times d$ real symmetric matrix such that the Hermitian matrix $\Sigma:=V+\sqrt{-1} S$ is positive semidefinite. 
When the canonical observables $X$ need to be specified, we also use the notation $(X,\phi)\sim N(\mu,\Sigma)$.

When we discuss relationships between a quantum Gaussian state $\phi$ on a CCR and a state on another algebra, we need to use the {\em quasi-characteristic function} \cite{qLevyCramer2}
\begin{equation}\label{eqn:q-GaussianCharFnc}
\phi\left(\prod_{t=1}^{r} e^{\sqrt{-1}\xi_{t}^{i}X_{i}}\right)
=
\exp\left(\sum_{t=1}^{r}\left(\sqrt{-1}\xi_{t}^{i}\mu_{i}-\frac{1}{2}\xi_{t}^{i}\xi_{t}^{j} \Sigma_{ji} \right)-\sum_{t=1}^{r}\sum_{u=t+1}^{r}\xi_{t}^{i}\xi_{u}^{j} \Sigma_{ji} \right)
\end{equation}
of a quantum Gaussian state, where $(X,\phi)\sim N(\mu,\Sigma)$ and $\{\xi_t\}_{t=1}^r \subset \R^d$.  
Note that \eqref{eqn:q-GaussianCharFnc} is analytically continued to $\{\xi_t\}_{t=1}^r \subset \C^d$.

The notion of quasi-characteristic function is exploited in discussing the quantum counterpart of the weak convergence \cite{{qLevyCramer},{qLevyCramer2},{qcontiguity}}.
For each $n\in\N$, let $\rho^{(n)}$ be a quantum state and $X^{(n)}=(X_1^{(n)},\dots,X_d^{(n)})$ be a list of observables on a finite dimensional Hilbert space $\H^{(n)}$.
We say the sequence $(X^{(n)},\rho^{(n)})$ {\em converges in distribution} to 
$(X,\phi)\sim N(\mu,\Sigma)$, in symbols
\[
(X^{(n)},\rho^{(n)}) \convd{} (X,\phi)
\quad \text{or} \quad
X^{(n)} \convd{\rho^{(n)}} N(\mu,\Sigma)  
\]
if
\[
\lim_{n\to\infty} \Tr \rho^{(n)} \left(\prod_{t=1}^{r} e^{\sqrt{-1}\xi_{t}^{i}X^{(n)}_{i}}\right)
=
\phi\left(\prod_{t=1}^{r} e^{\sqrt{-1}\xi_{t}^{i}X_{i}}\right)
\]
holds for any $r\in\N$ and subset $\{\xi_t\}_{t=1}^r$ of $\R^d$.

%----------------------------------------------------------------------------------------------------------------------------------------------------------------
\section{$D$-extendibility for i.i.d.~and non-i.i.d.~models}\label{app:D-extension}
%----------------------------------------------------------------------------------------------------------------------------------------------------------------

This section is a continuation of Remark \ref{rem:DinvExt},  demonstrating the $D$-extendibility of i.i.d.~models, the idea behind the terms `asymptotic $D$-invariance' and `$D$-extension', and a proper asymptotic treatment of the model presented in Example \ref{eg:noLimitPOVM}.
We also give an example of a sequence of quantum statistical models that is non-i.i.d.~but is, nevertheless, q-LAN and $D$-extendible.

Given a quantum state $\rho$ on a finite dimensional Hilbert space $\H$, let $\D_\rho: B(\H)\to B(\H)$ be Holevo's commutation operator \cite{holevo} with respect to $\rho$ defined by 
\[
 \D_\rho=\sqrt{-1}\,\frac{\mathcal{L}_{\rho}-\mathcal{R}_{\rho}}{\mathcal{L}_{\rho}+\mathcal{R}_{\rho}},
 \]
where $\mathcal{L}_{\rho}$ and $\mathcal{R}_{\rho}$ are superoperators defined by
\[ \mathcal{L}_{\rho} Z:=\rho Z,\qquad \mathcal{R}_{\rho} Z:=Z \rho,\qquad (Z\in B(\H)). \]
They are positive (selfadjoint) operators with respect to the Hilbert-Schmidt inner product $\langle A, B\rangle_{\rm HS}:=\Tr A^*B$ of $B(\H)$ because
\[ \langle Z, \mathcal{L}_{\rho} Z\rangle_{\rm HS}=\Tr Z^{*}\rho Z \ge 0 \]
and
\[ \langle Z, \mathcal{R}_{\rho} Z\rangle_{\rm HS}=\Tr Z^{*} Z \rho= \Tr Z \rho Z^*\ge 0 \]
for all $Z\in B(\H)$. 

When $\rho$ is not faithful, $\D_\rho$ is regarded as a superoperator acting on the quotient space $B(\H)/K_\rho$, where%
\footnote{
In \cite{holevo}, commutation operator $\D_\rho$ was defined on the space $L^2(\rho)$ of square-summable operators, which is the completion of $B(\H)$ with respect to the pre-inner product $\langle X,Y\rangle_\rho:=\frac{1}{2}\,\Tr\rho(X^*Y+YX^*)$. Note that $\langle K,K\rangle_\rho=0$ if and only if $K\rho=\rho K=0$. The `if' part is obvious, and the `only if' part is proved by observing $2\langle K,K\rangle_\rho=\Tr (K\sqrt{\rho})^*(K\sqrt{\rho})+\Tr (\sqrt{\rho} K)^*(\sqrt{\rho}K)$. }
\[ K_\rho:=\{ K\in B(\H) :  K\rho=\rho K=0\}. \]
Since $\D_\rho$ sends selfadjoint operators to selfadjoint operators, it is also regarded as a superoperator on  $B_{\rm sa}(\H)/K_\rho$, where $B_{\rm sa}(\H)$ is the set of selfadjoint operators. 

A subspace $V$ of $B_{\rm sa}(\H)$ is called {\em $\D_\rho$-invariant} if $V/K_\rho$ is $\D_\rho$-invariant. 
Given two lists of selfadjoint operators $\left(X_{1},\dots,X_{r}\right)$ and $\left(L_{1},\dots,L_{d}\right)$, the former is called a {\em $\D_\rho$-invariant extension} of the latter if 
${\rm Span}_{\R} \left\{ X_i \right\}_{i=1}^r \supset {\rm Span}_{\R} \left\{ L_i \right\}_{i=1}^d$ and 
${\rm Span}_{\R} \left\{ X_i \right\}_{i=1}^r$ is $\D_\rho$-invariant.

The following theorem motivated us to adopt the term `asymptotic $D$-invariance' in order to describe an asymptotic version of $\D_{\rho_{\theta_0}}$-invariance.

\begin{theorem}\label{thm:iid}
Given a quantum statistical model $\S:=\{ \rho_{\theta} : \theta\in\Theta\subset\R^{d}\}$ on a finite dimensional Hilbert space $\H$, let $(L_{1},\dots,L_{d})$ be its SLDs at $\theta_{0}\in\Theta$, and let $\S^{(n)}:=\{ \rho_{\theta}^{\otimes n} : \theta\in\Theta\subset\R^{d}\}$ be its i.i.d.~extensions. 
Take a linearly independent $\D_{\rho_{\theta_0}}$-invariant extension $(X_{1},\dots,X_{r})$ of $(L_{1},\dots,L_{d})$ satisfying $\Tr\rho_{\theta_0} X_i=0$ for all $i=1,\dots,r$, and let
\begin{equation*}
\Delta_{i}^{(n)}:=\frac{1}{\sqrt{n}}\sum_{k=1}^{n}I^{\otimes(k-1)}\otimes L_{i}\otimes I^{\otimes(n-k)},\qquad(1\leq i\leq d),
\end{equation*}
\begin{equation*}
X_{i}^{(n)}:=\frac{1}{\sqrt{n}}\sum_{k=1}^{n}I^{\otimes(k-1)}\otimes X_{i}\otimes I^{\otimes(n-k)},\qquad(1\leq i\leq r).
\end{equation*}
Then $X^{(n)}$ satisfies conditions \eqref{eq:X_cond2} -- \eqref{eq:X_cond1} in Definition \ref{def:propertyD}.
\end{theorem}

We prove Theorem \ref{thm:iid} in a series of lemmas.

\begin{lemma}\label{lem:Dinv_ineq}
Given a quantum state $\rho$ and a list of observables $(X_{1},\dots,X_{d})$ on a finite dimensional Hilbert space $\H$, let $A$ and $J$ be $d\times d$ nonnegative matrices whose $(i,j)$th entries are 
$A_{ij}=\Tr\sqrt{\rho}X_{j}\sqrt{\rho}X_{i}$ and $J_{ij}=\Tr\rho X_{j}X_{i}$. 
Then, both $A$ and $J\#J^{\top}$ are real matrices and satisfy
\begin{equation*}
A\leq J\#J^{\top},
\end{equation*}
where $\#$ denotes the operator geometric mean. 
\end{lemma}

\begin{proof}
$\overline{A}=A$ is obvious, and
\[ \overline{J\#J^\top}=\overline{J\#\overline{J}}=\overline{J}\#J=J\#\overline{J}=J\#J^\top. \]
Now recall that the operator geometric mean $P\#Q$ for positive operators $P$ and $Q$ is characterized as 
\cite{Ando}
\begin{equation*}
 P\#Q=\max\left\{ X\geq0 : \begin{pmatrix}P & X \\ X & Q \end{pmatrix} \geq 0 \right\}.
\end{equation*}
Since the Gram matrix for $\{\sqrt{\rho}X_{1},\dots,\sqrt{\rho}X_{d}\}\cup\{X_{1}\sqrt{\rho},\dots,X_{d}\sqrt{\rho}\}$
with respect to the Hilbert-Schmidt inner product is
\begin{equation*}
 \begin{pmatrix}J & A \\ A & J^{\top} \end{pmatrix}, 
\end{equation*}
the inequality $A\leq J\#J^{\top}$ immediately follows.
\end{proof}

\begin{lemma}\label{lem:concreteA}
Let $J=V+\i S$ be nonnegative matrix, and assume that $V={\rm Re}\,J$ is strictly positive. 
Then 
\begin{equation}
J\#J^\top=V^{1/2}\left\{ I+ \left(V^{-1/2}SV^{-1/2} \right)^{2} \right\}^{1/2} V^{1/2}.
\end{equation}
\end{lemma}

\begin{proof}
By changing $J$ into $J+\varepsilon I$ for $\varepsilon>0$ and considering the limit $\varepsilon \downarrow 0$, 
it suffices to treat the case when $J > 0$. 
Set 
\[ S_{V}:=V^{-1/2}SV^{-1/2}\quad\mbox{and}\quad X:=V^{1/2}\left\{I+S_{V}^{2}\right\}^{1/2} V^{1/2}. \] 
Then
\begin{align*}
X J^{-1} X & =X \left\{ V^{1/2}(I+\i S_{V})V^{1/2}\right\} ^{-1} X\\
 & =V^{1/2} \left\{ I+S_{V}^{2} \right\} \left\{ I+\i S_{V} \right\}^{-1}V^{1/2}\\
 & =V^{1/2} \left\{ I-\i S_{V} \right\} V^{1/2}\\
 & =J^{\top}.
\end{align*}
This proves that $X=J\# J^\top$. 
\end{proof}

\begin{lemma}\label{lem:Dinv_equiv}
Under the setting of Lemma \ref{lem:Dinv_ineq}, assume further that $V={\rm Re}\,J$ is strictly positive. 
Then the following conditions are equivalent. 
\begin{itemize}
\item [{(i)}] $A=J\#J^\top$
\item [{(ii)}] ${\rm Span}_{\C}\left\{ \sqrt{\rho}X_{i}+X_{i}\sqrt{\rho}\right\} _{i=1}^{d}\supset{\rm Span}_{\C}\left\{ \sqrt{\rho}X_{i}-X_{i}\sqrt{\rho}\right\} _{i=1}^{d}$
\item [{(iii)}] ${\rm Span}_{\C}\left\{ \rho X_{i}+X_{i}\rho\right\} _{i=1}^{d}\supset{\rm Span}_{\C}\left\{ \rho X_{i}-X_{i}\rho\right\} _{i=1}^{d}$
\item [{(iv)}] ${\rm Span}_{\C}\left\{ {X}_{i}\right\} _{i=1}^{d}$ is $\D_\rho$-invariant.
\item [{(v)}] ${\rm Span}_{\R}\left\{ {X}_{i}\right\} _{i=1}^{d}$ is $\D_\rho$-invariant.
\end{itemize}
\end{lemma}

\begin{proof}
We first prove that (i) $\Leftrightarrow$ (ii). Letting $J=V+\i S$, the Gram matrix $G$ for $\{ \sqrt{\rho}X_{i}+X_{i}\sqrt{\rho}\} _{i=1}^{d}\cup\{ \sqrt{\rho}X_{i}-X_{i}\sqrt{\rho}\} _{i=1}^{d}$ with respect to the Hilbert-Schmidt inner product is written as
\begin{equation*} %\label{eq:block_J}
G=
2\begin{pmatrix}V+A \;\; & \sqrt{-1}S \\ \sqrt{-1}S\;\; & V-A \end{pmatrix}
=\begin{pmatrix}I & I \\ I & -I \end{pmatrix}
  \begin{pmatrix}J & A\\ A & J^\top \end{pmatrix}
  \begin{pmatrix}I & I\\ I & -I \end{pmatrix}.
\end{equation*}
Condition (ii) is equivalent to saying that 
\[  {\rm rank}\, G={\rm rank}\, (V+A) =d. \]
Since 
\[
  \begin{pmatrix}J & A\\ A & J^\top \end{pmatrix}
  = \begin{pmatrix} V^{1/2} &  0 \\ 0 & V^{1/2} \end{pmatrix}
     \begin{pmatrix}I+\i S_{V} & A_{V} \\ A_{V} & I-\i S_{V} \end{pmatrix}
     \begin{pmatrix} V^{1/2} &  0 \\ 0 & V^{1/2} \end{pmatrix},
\]
where $A_{V}:=V^{-1/2}AV^{-1/2}$ and $S_{V}:=V^{-1/2}SV^{-1/2}$, 
condition (ii) is further equivalent to saying that the nonnegative matrix 
\[ 
 \begin{pmatrix}I+\i S_{V} & A_{V} \\ A_{V} & I-\i S_{V} \end{pmatrix}
 = \begin{pmatrix}  \,I\, &  0 \\ 0 & \,I\, \end{pmatrix}
 + \begin{pmatrix}\i S_{V} & A_{V} \\ A_{V} & -\i S_{V} \end{pmatrix}
\]
is of rank $d$, that is, the matrix 
\[
\begin{pmatrix}\i S_{V} & A_{V} \\ A_{V} & -\i S_{V} \end{pmatrix}
\]
has eigenvalues $-1$ and $+1$ each with multiplicity $d$. 
(Note that if $(x,y)^\top$ is an eigenvector corresponding to the eigenvalue $-1$, then $(y,-x)^\top$ is an eigenvector corresponding to the eigenvalue $+1$.) 
This is equivalent to 
\[
 \begin{pmatrix}\i S_{V} & A_{V} \\ A_{V} & -\i S_{V} \end{pmatrix}^{2}
 =\begin{pmatrix}A_{V}^{2}-S_{V}^{2} & \i(S_{V}A_{V}-A_{V}S_{V})\\
	\i(S_{V}A_{V}-A_{V}S_{V}) & A_{V}^{2}-S_{V}^{2} \end{pmatrix}
 =\begin{pmatrix} \,I\, & 0 \\ 0 & \,I\, \end{pmatrix}, 
\]
or 
\[ A_{V}=\left\{ I+S_{V}^{2} \right\}^{1/2}. \] 
Due to Lemma \ref{lem:concreteA}, this is further equivalent to 
\[ A=V^{1/2} \left\{ I+S_{V}^{2} \right\}^{1/2} V^{1/2}=J\#J^\top. \]

We next prove that (ii) $\Leftrightarrow$ (iii).
Condition (ii) says that ${\rm Span}_{\C}\{ {X}_{i}\} _{i=1}^{d}$ is invariant under the action of 
\[ \D_{1}=\frac{\sqrt{\mathcal{L}_\rho}-\sqrt{\mathcal{R}_\rho}}{\sqrt{\mathcal{L}_\rho}+\sqrt{\mathcal{R}_\rho}}, \]
while condition (iii) says that ${\rm Span}_{\C}\{ {X}_{i}\} _{i=1}^{d}$ is invariant under the action of
\[ \D_{2}=\frac{\mathcal{L}_\rho-\mathcal{R}_\rho}{\mathcal{L}_\rho+\mathcal{R}_\rho}. \]
Since both $\D_{1}$ and $\D_2$ are selfadjoint with respect to the Hilbert-Schmidt inner product, 
and $-I\leq\D_{1},\D_{2}\leq I$, continuous functional calculus shows that they are related as
\[ 
 \D_{1}=\frac{\sqrt{1+\D_{2}}-\sqrt{1-\D_{2}}}{\sqrt{1+\D_{2}}+\sqrt{1-\D_{2}}} 
 \quad\mbox{and}\quad
 \D_{2}=\frac{2\D_{1}}{1+\D_{1}^{2}}. 
\]
Consequently, $\D_1$-invariance and $\D_2$-invariance are equivalent. 

Further, since $\D_\rho=\i\,\D_{2}$, we have (iii) $\Leftrightarrow$ (iv).
Finally, (iv) $\Leftrightarrow$ (v) is obvious. 
\end{proof}

\begin{proof}[Proof of Theorem \ref{thm:iid}]
Firstly, condition \eqref{eq:X_cond1} is obvious because $(X_{1},\dots,X_{r})$ is a  $\D_{\rho_{\theta_0}}$-invariant extension of $(L_{1},\dots,L_{d})$. 
Secondly, condition \eqref{eq:X_cond2} follows from the quantum central limit theorem for sums of i.i.d.~observables \cite{qLevyCramer2} (cf., Lemma \ref{lem:qclt_niid} below), in that
\begin{equation*}
X^{(n)}\conv{\rho_{\theta_{0}}^{\otimes n}}N(0,\Sigma),
\end{equation*}
where $\Sigma_{ij}=\Tr \rho_{\theta_0} X_{j}X_{i}$. 
Now we prove the key condition \eqref{eq:X_cond3}. 

Let us regard $\hat{\H}:=B(\H)$ as a Hilbert space endowed with the Hilbert-Schmidt inner product. 
We introduce selfadjoint operators
$\superL_{X_i}$ and $\superR_{X_i}$ on $\hat\H$ for $i=1,\dots, r$ by
\[ \superL_{X_i} Z:=X_i Z,\qquad \superR_{X_i} Z:=Z X_i,\qquad (Z\in B(\H)). \]
Further, let $\psi_0:=\sqrt{\rho_{\theta_0}}$ be a reference vector in $\hat{\H}$. 
Note that
\[ 
 \langle \psi_0, \superL_{X_{i}} \psi_0 \rangle_{\rm HS}
 =\Tr \sqrt{\rho_{\theta_0}} (X_{i} \sqrt{\rho_{\theta_0}})
 =\Tr \rho_{\theta_0} X_{i}=0,
\]
and $\langle \psi_0, \superR_{X_{i}} \psi_0 \rangle_{\rm HS}=0$ likewise. 
Now consider the operators on $\hat\H^{\otimes n}$ defined by
\begin{align*}
\superL_{X_{i}}^{(n)} & :=\frac{1}{\sqrt{n}}\sum_{k=1}^{n}I^{\otimes(k-1)}\otimes\superL_{X_{i}}\otimes I^{\otimes(n-k)}, \\
\superR_{X_{i}}^{(n)} & :=\frac{1}{\sqrt{n}}\sum_{k=1}^{n}I^{\otimes(k-1)}\otimes\superR_{X_{i}}\otimes I^{\otimes(n-k)},
\end{align*}
and apply the quantum central limit theorem to the i.i.d.~extension states 
$\left(\ket{\psi_0}\bra{\psi_0}\right)^{\otimes n}$, to obtain
\begin{equation}\label{eq:LR_qclt}
\left(\superL_{X_1}^{(n)},\dots,\superL_{X_r}^{(n)},\superR_{X_1}^{(n)},\dots,\superR_{X_r}^{(n)}\right)\conv{\left(\ket{\psi_0}\bra{\psi_0}\right)^{\otimes n}}N\left(0,\begin{pmatrix}\Sigma & A\\
A & \Sigma^{\top}
\end{pmatrix}\right),
\end{equation}
where
\[
A_{ij} 
= \langle \psi_0, \superL_{X_{i}} \superR_{X_{j}} \psi_0 \rangle_{\rm HS}
  =\Tr\sqrt{\rho_{\theta_{0}}}X_{i}\sqrt{\rho_{\theta_{0}}}X_{j}. 
\]
Since ${\rm Span}_{\R}\left\{ X_{i}\right\} _{i=1}^{r}$ is $\D_{\rho_0}$-invariant, we see from 
Lemma \ref{lem:Dinv_equiv} that 
\begin{equation*}
A=\Sigma\#\Sigma^{\top}. 
\end{equation*}
Further, for all $\xi,\eta\in\R^{d}$, we have
\begin{eqnarray*}
 && \lim_{n\to\infty}\Tr\sqrt{\rho_{\theta_0}^{\otimes n}}e^{\sqrt{-1}\xi^{i} X_{i}^{(n)} }
	\sqrt{\rho_{\theta_0}^{\otimes n}}e^{\sqrt{-1}\eta^{i} X_{i}^{(n)} }  \\
 &&\qquad =\lim_{n\to\infty} \Tr \left(\ket{\psi_0}\bra{\psi_0}\right)^{\otimes n}
	e^{\sqrt{-1}\left(\xi^{i}\superL_{X_{i}}^{(n)}+\eta^{i}\superR_{X_{i}}^{(n)}\right)} \\
 &&\qquad =\exp\left[ {-\frac{1}{2}
 	\begin{pmatrix}\xi \\ \eta \end{pmatrix}^\top
 	\begin{pmatrix}\Sigma & A \\ A & \Sigma^{\top} \end{pmatrix}
	\begin{pmatrix}\xi \\ \eta \end{pmatrix}} \right],
\end{eqnarray*}
where (\ref{eq:LR_qclt}) is used in the second equality.
This proves (\ref{eq:X_cond3}). 
\end{proof}

\begin{remark}
Several remarks on the $D$-extendibility of the one-dimensional pure state model $\rho_\theta$ treated in Example \ref{eg:noLimitPOVM} are now in order. 
Let us first show that 
\[
\Delta^{(n)}:=\frac{1}{\sqrt{n}}\sum_{k=1}^{n}I^{\otimes(k-1)}\otimes\sigma_x\otimes I^{\otimes(n-k)}
\]
is not asymptotically $D$-invariant at $\theta=0$.
To this end, it suffices to prove that $\Delta^{(n)}$ does {\em not} satisfy the identity
\[
 \lim_{n\to\infty}
 \Tr\sqrt{\rho_0^{\otimes n}}e^{\sqrt{-1}\xi \Delta^{(n)}}\sqrt {\rho_0^{\otimes n}}e^{\sqrt{-1}\eta \Delta^{(n)}} 
=e^{-\frac{1}{2}
 \begin{pmatrix} \xi \\ \eta \end{pmatrix}^\top
 \begin{pmatrix} J & J \\ J & J \end{pmatrix}
 \begin{pmatrix} \xi \\ \eta \end{pmatrix}}
\]
for $\xi,\eta\in\R$, 
where $J:=\Tr \rho_0\, \sigma_x^2=1$ is the SLD Fisher information of the model at $\theta=0$.
In fact, since $\Tr\sqrt{\rho_0}\,\sigma_x\,\sqrt{\rho_0}\,\sigma_x=0$, we can compute in a quite similar way to the proof of Theorem \ref{thm:iid} that
\[
 \lim_{n\to\infty}
 \Tr\sqrt{\rho_0^{\otimes n}}e^{\sqrt{-1}\xi \Delta^{(n)}}\sqrt{ \rho_0^{\otimes n}}e^{\sqrt{-1}\eta \Delta^{(n)}} 
=e^{-\frac{1}{2}
 \begin{pmatrix} \xi \\ \eta \end{pmatrix}^\top
 \begin{pmatrix} 1 & 0 \\ 0 & 1 \end{pmatrix}
 \begin{pmatrix} \xi \\ \eta \end{pmatrix}}.
\]
This proves the claim.

We next verify that $\Delta^{(n)}$ has a $D$-extension and therefore the model is $D$-extendible at  $\theta=0$. 
While this is a straightforward consequence of Theorem \ref{thm:iid}, we demonstrate this by a direct computation. 
Let $(X_1, X_2):=(\sigma_x, \sigma_y)$, which is a $\D_{\rho_0}$-invariant extension of the SLD $\sigma_x$ at $\rho_0$, and let
\[
X_{i}^{(n)}:=\frac{1}{\sqrt{n}}\sum_{k=1}^{n}I^{\otimes(k-1)}\otimes X_{i}\otimes I^{\otimes(n-k)},\qquad (i=1,2). 
\]
Then by a direct computation similar to the above identity, we have
\[
 \lim_{n\to\infty}
 \Tr\sqrt{\rho_0^{\otimes n}}e^{\sqrt{-1}\xi^i X_i^{(n)}}\sqrt{\rho_0^{\otimes n}}e^{\sqrt{-1}\eta^j X_j^{(n)}} 
=e^{-\frac{1}{2}
 \begin{pmatrix} \xi \\ \eta \end{pmatrix}^\top
 \begin{pmatrix} \Sigma & 0 \\ 0 & \Sigma^\top \end{pmatrix}
 \begin{pmatrix} \xi \\ \eta \end{pmatrix}}
\]
for $\xi,\eta\in\R^2$, where
\[
 \Sigma=\left[ \Tr \rho_0 X_j X_i \right]_{ij}=\begin{pmatrix} 1 & -\i \\ \i & 1 \end{pmatrix}. 
\]
Since $\Sigma\#\Sigma^\top=0$, we see that $X^{(n)}$ is a $D$-extension of $\Delta^{(n)}\,(=X_1^{(n)} )$ 
with $F=(1,0)^\top$.

Finally, we demonstrate a proper perspective on the local parameter i.i.d.~model $\rho_{h/\sqrt{n}}^{\otimes n}$.
As shown in Example \ref{eg:noLimitPOVM}, the sequence $M^{(n)}=\{\rho_0^{\otimes n},\,I^{(n)}-\rho_0^{\otimes n} \}$ of binary POVMs does not have a binary POVM on the `classical' Gaussian shift model $N(h,1)$ that gives the limiting distribution $\L_h=( e^{-\frac{1}{4}h^2},\, 1- e^{-\frac{1}{4}h^2} )$. 
This fact nullifies the naive conjecture presented just before Example \ref{eg:noLimitPOVM}, 
but it does not rule out the existence of a POVM on {\em another} CCR algebra that gives the above limiting distribution $\L_h$. 
In fact, Theorem \ref{thm:qRep} tells us that $M^{(n)}$ has a limiting binary POVM $M^{(\infty)}=\{M^{(\infty)}(0), M^{(\infty)}(1)\}$ on the `quantum' Gaussian shift model 
\[
\phi_h\sim N(({\rm Re}\,\Sigma F)\,h, \Sigma)
=N\left( \begin{pmatrix} h\\ 0 \end{pmatrix},\,\begin{pmatrix} 1 & -\i \\ \i & 1 \end{pmatrix} \right)
\]
that satisfies $\phi_h(M^{(\infty)}(0))=\L_{h}(0)$ for every $h\in\R$.
To be specific, let $\H^{(\infty)}$ be a separable Hilbert space that irreducibly represents the ${\rm CCR}({\rm Im}\, \Sigma)$, and let $\rho^{(\infty)}_{h}$ be the density operator of the quantum Gaussian state $\phi_h$ on $\H^{(\infty)}$.
Then, from the noncommutative Parseval identity \cite{holevo}, we see that the POVM $M^{(\infty)}:=\{ \rho^{(\infty)}_{0}, \, I^{(\infty)}-\rho^{(\infty)}_{0} \}$ fulfills
\begin{align*}
\Tr \rho^{(\infty)}_h M^{(\infty)}(0) 
& = \sqrt{\frac{\det ({\rm Im}\, \Sigma)}{\pi^2}} \int_{\R^2} \, \overline{\F_\xi[\rho^{(\infty)}_h]}\, \F_\xi[M^{(\infty)}(0)]  \, d\xi\\
& = \frac{1}{\pi} \int_{\R^2} \, e^{-\i \, \xi^1 h -\frac{1}{2} \|\xi \|^2} \cdot  e^{-\frac{1}{2} \|\xi \|^2} d\xi
= e^{-\frac{1}{4}h^2}
\end{align*}
for every $h\in\R$. 
\end{remark}

Let us proceed to the issue of handling non-i.i.d.~quantum statistical models.
We start with a slightly generalized version of the quantum central limit theorem. 

\begin{lemma}[Quantum central limit theorem for sums of non-i.i.d.~observables]\label{lem:qclt_niid}
For each $k \in\N$, let $\H^{(k)}$ be a finite dimensional Hilbert space, 
and let $\sigma^{(k)}$ and $A^{(k)}=(A_{1}^{(k)},\dots,A_{r}^{(k)})$ be a quantum state and a list of observables on $\H^{(k)}$. 
Assume that $A^{(k)}$ are zero-mean:
\[
\Tr\sigma^{(k)}A_{i}^{(k)}=0\qquad(1\leq i\leq r),
\]
uniformly bounded: 
\[
\sup_{k\in\N,1\leq i\leq r}\left\Vert A_{i}^{(k)}\right\Vert <\infty,
\]
and there is an $r\times r$ nonnegative matrix $\Sigma$ such that
\[
 \lim_{n\to\infty}
 \Tr\sigma^{(k)}A_{j}^{(k)}A_{i}^{(k)}=\Sigma_{ij}
 \qquad(1\leq i,j\leq r). 
\]
Then under the tensor product states:
\[
\rho^{(n)}:=\bigotimes_{k=1}^{n}\sigma^{(k)},
\]
the scaled sums of observables
\[
X_{i}^{(n)}:=\frac{1}{\sqrt{n}}\sum_{k=1}^{n}I^{\otimes(k-1)}\otimes A_{i}^{(k)}\otimes I^{\otimes(n-k)}
\]
exhibit
\[
X^{(n)}\conv{\rho^{(n)}}N(0,\Sigma). 
\]
\end{lemma}

\begin{proof}
We need only check the convergence of quasi-characteristic functions
\[
 \lim_{n\to\infty}\Tr\rho^{(n)}\prod_{t=1}^{T}e^{\i\xi_{t}^{i}X_{i}^{(n)}}
 =\exp\left\{ -\frac{1}{2}\sum_{t=1}^{T}\xi_{t}^{\top}\Sigma\xi_{t}
 	-\sum_{t=1}^{T}\sum_{s=t+1}^{T}\xi_{s}^{\top}\Sigma\xi_{t} \right\},
\]
for all $T\in\N$ and $\{\xi_{t}\}_{t=1}^{T}\subset\R^{r}$.
Observe 
\begin{align*}
 & \Tr\rho^{(n)}\prod_{t=1}^{T}e^{\i\xi_{t}^{i}X_{i}^{(n)}}  \\
 &\qquad =\prod_{k=1}^{n}\left\{ \Tr\sigma^{(k)}\prod_{t=1}^{T}e^{\frac{\i}{\sqrt{n}}\xi_{t}^{i}A_{i}^{(k)}}\right\} \\
 &\qquad =\prod_{k=1}^{n}\left\{ \Tr\sigma^{(k)}\sum_{m\in\Z_{+}^{T}}\prod_{t=1}^{T}\frac{1}{m_{t}!}\left(\frac{\i}{\sqrt{n}}\xi_{t}^{i}A_{i}^{(k)}\right)^{m_{t}}\right\} \\
 &\qquad =\prod_{k=1}^{n}\left\{ 1-\frac{1}{n}\left(\frac{1}{2}\sum_{t=1}^{T}\xi_{t}^{\top}\Sigma^{(k)}\xi_{t}
 	+\sum_{t=1}^{T}\sum_{s=t+1}^{T}\xi_{s}^{\top}\Sigma^{(k)}\xi_{t} \right)+c^{(k)}(n)\right\} 
\end{align*}
where $\Z_{+}=\{0\}\cup\N$, $\Sigma_{ij}^{(k)}:=\Tr\sigma^{(k)}A_{j}^{(k)}A_{i}^{(k)}$,  and
\[
c^{(k)}(n)
 :=\sum_{m_{1}+\cdots+m_{T}\geq3}
 \frac{(\i)^{m_1+\cdots + m_T}}{n^{(m_1+\cdots + m_T)/2}} \Tr\sigma^{(k)}
 \prod_{t=1}^{T} \frac{1}{m_{t}!} \left(\xi_{t}^{i}A_{i}^{(k)}\right)^{m_{t}}. 
\]
Note that, since $\{A_{i}^{(k)}\}_{i,k}$ are assumed to be uniformly bounded,  
\[
\max_{1\le k \le n}\left|c^{(k)}(n)\right|=O\left(\frac{1}{n\sqrt{n}}\right). 
\]
Consequently, we can further evaluate the quasi-characteristic function as
\begin{align*}
 &\log \left\{ \Tr\rho^{(n)}\prod_{t=1}^{T}e^{\i\xi_{t}^{i} X_{i}^{(n)}}  \right\} \\
 &\qquad =\sum_{k=1}^{n}\log\left\{ 1-\frac{1}{n}\left(\frac{1}{2}\sum_{t=1}^{T}\xi_{t}^{\top}\Sigma^{(k)}\xi_{t}
 	+\sum_{t=1}^{T}\sum_{s=t+1}^{T}\xi_{s}^{\top}\Sigma^{(k)}\xi_{t} \right)+O\left(\frac{1}{n\sqrt{n}}\right) \right\} \\
 &\qquad =\sum_{k=1}^{n}\left\{ -\frac{1}{n}\left(\frac{1}{2}\sum_{t=1}^{T}\xi_{t}^{\top}\Sigma^{(k)}\xi_{t}
 	+\sum_{t=1}^{T}\sum_{s=t+1}^{T}\xi_{s}^{\top}\Sigma^{(k)}\xi_{t} \right)+O\left(\frac{1}{n\sqrt{n}}\right)\right\}.
\end{align*}
Taking the limit $n\to\infty$, we have
\[
 \lim_{n\to\infty}\log \left\{ \Tr\rho^{(n)}\prod_{t=1}^{T}e^{\i\xi_{t}^{i} X_{i}^{(n)}} \right\}
 =-\frac{1}{2}\sum_{t=1}^{T}\xi_{t}^{\top}\Sigma\xi_{t}
 	-\sum_{t=1}^{T}\sum_{s=t+1}^{T}\xi_{s}^{\top}\Sigma\xi_{t}.
\]
This proves the claim. 
\end{proof}

We are now ready to give an example of a sequence of quantum statistical models that is non-i.i.d.~but is, nevertheless, q-LAN and $D$-extendible.

\begin{example}
Given a sequence 
$\{ \sigma_{\theta}^{(k)}:\theta\in\Theta\subset\R^{d}\}_{k\in\N}$
of quantum statistical models on a fixed finite dimensional Hilbert space $\H$, 
let us consider their tensor products on $\H^{\otimes n}$ defined by
\[
 \rho_{\theta}^{(n)}:=\bigotimes_{k=1}^{n}\sigma_{\theta}^{(k)}.
\]
If $\sigma_{\theta}^{(k)}$ converges to a model $\sigma_\theta^{(\infty)}$ 
in a certain mode of convergence as $k\to\infty$, 
it is expected that 
the model $\rho_{\theta}^{(n)}$ will be q-LAN and $D$-extendible, because it is almost i.i.d.~in the asymptotic limit.
In what follows, we demonstrate a sufficient condition for realizing this scenario.

Assume that, for some $\theta_0\in\Theta$,
\begin{equation}\label{eqn:niid_Dext1}
 \lim_{k\to\infty} \sigma_{\theta_0}^{(k)}= \sigma_{\theta_0}^{(\infty)} \qquad 
\end{equation}
and the SLDs $\{L^{(k)}_i\}_{i=1}^d$ of $\sigma^{(k)}_{\theta}$ at $\theta_0\in\Theta$ is convergent:
\begin{equation}\label{eqn:niid_Dext2}
\lim_{k\to\infty} L_{i}^{(k)}= L_{i}^{(\infty)}\qquad (i=1,\dots, d). 
\end{equation}
Then $\rho_{\theta}^{(n)}$ with 
\[
\Delta_{i}^{(n)}=\frac{1}{\sqrt{n}}\sum_{k=1}^{n}I^{\otimes(k-1)}\otimes L_{i}^{(k)}\otimes I^{\otimes(n-k)}
\]
is $D$-extendible at $\theta_0$.

Assume further that the square-root likelihood ratios
$R_{h}^{(k)}:=\Ratio(\sigma_{\theta_{0}+h}^{(k)}\mid\sigma_{\theta_{0}}^{(k)})$
around $\theta_0$ satisfy
\begin{equation}\label{eq:niid_qlan_diff1}
\sup_{k\in\N\cup{\{\infty\}}} \left\Vert R_h^{(k)}-I-\frac{1}{2} h^i L_i^{(k)}\right\Vert = o(\Vert h \Vert)
\end{equation}
and
\begin{equation}\label{eq:niid_qlan_diff2}
\sup_{k\in\N\cup{\{\infty\}}} \left( 1-\Tr \sigma^{(k)}_{\theta_0} R_h^{(k)^2} \right)=o(\Vert h \Vert^2).
\end{equation}
In the left-hand side of \eqref{eq:niid_qlan_diff1}, the norm $\| \cdot \|$ stands for the operator norm.
Then $\rho_{\theta}^{(n)}$ is q-LAN at $\theta_0$. 
\end{example}

\begin{proof}
Let $D^{(\infty)}=(D_{1}^{(\infty)},\dots,D_{r}^{(\infty)})$ be a $\D_{\sigma^{(\infty)}}$-invariant extension of 
$L^{(\infty)}=(L_1^{(\infty)},\dots, L_d^{(\infty)})$ such that $D_{i}^{(\infty)}=L_{i}^{(\infty)}$ for $i=1,\dots,d$.
Accordingly, we define, for each $k\in\N$, a set of observables $D^{(k)}=(D_{1}^{(k)},\dots,D_{r}^{(k)})$ by
\[
 D_{i}^{(k)}
 =\begin{cases}
 L_{i}^{(k)} & \quad(1\leq i\leq d)\\
 D_{i}^{(\infty)}-(\Tr\sigma_{\theta_0}^{(k)}D_{i}^{(\infty)})I & \quad(d+1 \le i \le r)
 \end{cases}.
\]
It then follows from \eqref{eqn:niid_Dext1} and \eqref{eqn:niid_Dext2} that
\[
 \Sigma_{ij}^{(k)}:=\Tr\sigma_{\theta_0}^{(k)}D_{j}^{(k)}D_{i}^{(k)},\qquad 
 A_{ij}^{(k)}=:\Tr\sqrt{\sigma_{\theta_0}^{(k)}}D_{j}^{(k)}\sqrt{\sigma_{\theta_0}^{(k)}}D_{i}^{(k)}
\]
for $k\in\N\cup\{\infty\}$ satisfy 
\[
 \lim_{k\to\infty}\Sigma^{(k)}=\Sigma^{(\infty)},\qquad
 \lim_{k\to\infty}A^{(k)}=A^{(\infty)}.
\]
Moreover, since $D^{(\infty)}$ is $\D_{\sigma^{(\infty)}}$-invariant, we see from Lemma \ref{lem:Dinv_equiv} that $A^{(\infty)}=\Sigma^{(\infty)}\#\Sigma^{(\infty)^{\top}}$.

Let $X^{(n)}=(X^{(n)}_1,\dots, X^{(n)}_r)$, where
\[
 X_{i}^{(n)}=\frac{1}{\sqrt{n}}\sum_{k=1}^{n}I^{\otimes(k-1)}\otimes D_{i}^{(k)}\otimes I^{\otimes(n-k)}. 
\]
In order to prove the $D$-extendibility, it suffices to verify the following: 
\begin{itemize}
\item[(i)] There is an $r\times d$ matrix $F$ satisfying $\Delta_k^{(n)}=F_k^i X_i^{(n)}$ for all $n$.\\
\item[(ii)] $X^{(n)}\conv{\rho_{\theta_0}^{(n)}}N(0,\Sigma^{(\infty)})$.\\
\item[(iii)] For all $\xi,\eta\in\R^{r}$, 
\[
\lim_{n\to\infty}\Tr\sqrt{\rho_{\theta_{0}}^{(n)}}e^{\sqrt{-1}\xi^{i}X_{i}^{(n)}}\sqrt{\rho_{\theta_{0}}^{(n)}}e^{\sqrt{-1}\eta^{i}X_{i}^{(n)}}
=e^{-\frac{1}{2}
\begin{pmatrix}\xi \\ \eta \end{pmatrix}^\top
\begin{pmatrix}\Sigma^{(\infty)} & \Sigma^{(\infty)}\#\Sigma^{(\infty)^{\top}}\\
	\Sigma^{(\infty)}\#\Sigma^{(\infty)^{\top}} & \Sigma^{(\infty)^{\top}}\end{pmatrix}
\begin{pmatrix}\xi \\ \eta\end{pmatrix}.
}
\]
\end{itemize}
Firstly, by definition of $D^{(k)}$, (i) is satisfied by the following matrix
\[
 F=\left(\array{cc} I \\  O \endarray\right),
\]
where $I$ is the $d\times d$ identity matrix and $O$ is the $(r-d)\times d$ zero matrix.
We next show (ii) and (iii) simultaneously by modifying the proof of Theorem \ref{thm:iid}. 
Let us regard $\hat{\H}:=\B(\H)$ as a Hilbert space endowed with the Hilbert-Schmidt inner product
$\left\langle A,B\right\rangle_{\rm HS}:=\Tr A^{*}B$,
and let us introduce, for each $X\in B(\H)$, linear operators $\superL_{X}$ and
$\superR_{X}$ $(1\leq i\leq r)$ on $\hat{\H}$
by
\[
\superL_{X}Z=XZ,\qquad\superR_{X}Z=ZX,\qquad(Z\in B(\hat\H)).
\]
Further, let $\psi_{0}^{(k)}:=\sqrt{\sigma_{\theta_0}^{(k)}}\in \hat{\H}$ and introduce operators on $\hat{\H}^{\otimes n}$ by
\begin{align*}
\superL_{X_{i}}^{(n)} 
	& :=\frac{1}{\sqrt{n}} \sum_{k=1}^n I^{\otimes(k-1)}\otimes\superL_{X_{i}^{(k)}}\otimes I^{\otimes(n-k)},\\
\superR_{X_{i}}^{(n)} 
	& :=\frac{1}{\sqrt{n}} \sum_{k=1}^n I^{\otimes(k-1)}\otimes\superR_{X_{i}^{(k)}}\otimes I^{\otimes(n-k)}.
\end{align*}
Then, applying the quantum central limit theorem (Lemma \ref{lem:qclt_niid})
to the product states
\[ \hat{\rho}_{0}^{(n)}:=\bigotimes_{k=1}^{n}\ \ket{\psi_{0}^{(k)}}\bra{\psi_{0}^{(k)}}, \]
we obtain
\[
\left(\superL_{X^{(n)}}^{(n)},\superR_{X^{(n)}}^{(n)}\right)\conv{\hat{\rho}_{0}^{(n)}}
N\left(0,
\begin{pmatrix}
 \Sigma^{(\infty)} & \Sigma^{(\infty)}\#\Sigma^{(\infty)^{\top}}\\
 \Sigma^{(\infty)}\#\Sigma^{(\infty)^{\top}} & \Sigma^{(\infty)^{\top}}
\end{pmatrix}
\right).
\]
This proves (ii) and (iii). 

We next prove that, with additional assumptions \eqref{eq:niid_qlan_diff1} and \eqref{eq:niid_qlan_diff2}, the model $\rho_{\theta}^{(n)}$ is q-LAN at $\theta_0$.
Let $J_{ij}^{(k)}:=\Tr\sigma_{\theta_0}^{(k)}L_{j}^{(k)}L_{i}^{(k)}$ 
for $k\in\N\cup\{\infty\}$.
Since 
\[ \Delta^{(n)} \conv{\rho_{\theta_0}^{(n)}} N(0,J^{(\infty)}) \]
has been shown in (ii) of the proof of $D$-extendibility, 
it suffices to prove that 
\begin{equation}\label{eq:infinitesimal_niid}
\lim_{n\to\infty}\Tr\rho_{\theta_{0}}^{(n)}\left\{ e^{\frac{1}{2}\left(h^{i}\Delta_{i}^{(n)}-\frac{1}{2} h^{\top} J^{(\infty)} h \right)}-\bar{R}_{h}^{(n)}\right\} ^{2}=0
\end{equation}
for $\bar{R}_{h}^{(n)}:=\Ratio(\rho_{\theta_{0}+h/\sqrt{n}}^{(n)}\mid\rho_{\theta_{0}}^{(n)})$. 
The sequence appeared in the left-hand side of \eqref{eq:infinitesimal_niid} is rewritten as
\begin{align} \label{eq:niid_qlan_proof}
 \Tr \rho_{\theta_{0}}^{(n)}e^{h^{i}\Delta_{i}^{(n)}-\frac{1}{2}h^{\top} J^{(\infty)} h} 
 +\Tr\rho_{\theta_{0}}^{(n)} \bar{R}_{h}^{(n)^2}
 -2\,{\rm Re}\,\Tr\rho_{\theta_{0}}^{(n)}\bar{R}_{h}^{(n)}e^{\frac{1}{2}\left(h^{i}\Delta_{i}^{(n)}-\frac{1}{2}h^{\top} J^{(\infty)} h\right)}.
\end{align}
In order to prove \eqref{eq:infinitesimal_niid}, therefore, it suffices to verify the following:
\begin{itemize}
\item[(iv)] $\displaystyle \lim_{n\to\infty}  \Tr \rho_{\theta_{0}}^{(n)}e^{h^{i}\Delta_{i}^{(n)}-\frac{1}{2}h^{\top} J^{(\infty)} h} =1.$\\
\item[(v)] $\displaystyle \lim_{n\to\infty} \Tr\rho_{\theta_{0}}^{(n)} \bar{R}_{h}^{(n)^2}=1$.\\
\item[(vi)] $\displaystyle \lim_{n\to\infty} \Tr\rho_{\theta_{0}}^{(n)}\bar{R}_{h}^{(n)}e^{\frac{1}{2}\left(h^{i}\Delta_{i}^{(n)}-\frac{1}{2}h^{\top} J^{(\infty)} h\right)}=1$.
\end{itemize}
 
Firstly, because of \eqref{eqn:niid_Dext2}, 
the SLDs $\{L_{i}^{(k)} \}_{k\in\N, 1\le i\le d}$ are uniformly bounded, and thus
\begin{align*}
\Tr\rho_{\theta_{0}}^{(n)}e^{h^i \Delta_{i}^{(n)}} 
 & =\prod_{k=1}^{n}\Tr\sigma_{\theta_{0}}^{(k)}e^{\frac{1}{\sqrt{n}}h^{i}L_{i}^{(k)}}\\
 & =\prod_{k=1}^{n}\left(1+\frac{1}{2n}h^{\top}J^{(k)} h+O\left(\frac{1}{n \sqrt{n}}\right)\right).
\end{align*}
In the second line, $\Tr\sigma_{\theta_0}^{(k)}L_{i}^{(k)}=0$ was used, and the remainder term $O(1/n\sqrt{n})$ is uniform in $k$. 
Consequently, 
\begin{align*}
\lim_{n\to\infty}\log\Tr\rho_{\theta_{0}}^{(n)}e^{h^{i}\Delta_{i}^{(n)}} & =\lim_{n\to\infty}\sum_{k=1}^{n}\log\left(1+\frac{1}{2n}h^{\top}J^{(k)}h+O\left(\frac{1}{n\sqrt{n}}\right)\right)\\
 & =\lim_{n\to\infty}\sum_{k=1}^{n}\left(\frac{1}{2n}h^{\top}J^{(k)}h+O\left(\frac{1}{n\sqrt{n}}\right)\right) \\
 & =\frac{1}{2}h^{\top}J^{(\infty)} h,
\end{align*}
proving (iv). 

Secondly, taking the logarithm of
\[
\Tr\rho_{\theta_{0}}^{(n)} {\bar{R}_{h}^{(n)^2}}
=\prod_{k=1}^{n}\Tr\sigma_{\theta_{0}}^{(k)} {R_{h/\sqrt{n}}^{(k)^2}}, 
\]
we have
\begin{align*}
  \log\Tr\rho_{\theta_{0}}^{(n)} {\bar{R}_{h}^{(n)^2}}
 & =\sum_{k=1}^{n}\log\Tr\sigma_{\theta_{0}}^{(k)} {R_{h/\sqrt{n}}^{(k)^2}}\\
 & =\sum_{k=1}^{n}\log\left\{ 1-\left(1-\Tr\sigma_{\theta_{0}}^{(k)} {R_{h/\sqrt{n}}^{(k)^2}}\right)\right\} \\
 &= \sum_{k=1}^{n}\log\left\{ 1- o\left(\frac{1}{n}\right) \right\}.
\end{align*}
In the last equality, we used \eqref{eq:niid_qlan_diff2}. 
Since the last line converges to $0$ as $n\to\infty$, we have (v). 

In order to prove (vi), we need to show that
\[
 B^{(k)}(h):=I+\frac{1}{2}h^{i}L_{i}^{(k)}-R_{h}^{(k)}
\]
satisfies 
\begin{equation}\label{eqn:non-iid_rhoB}
 \Tr\sigma_{\theta_{0}}^{(k)}B^{(k)}(h/\sqrt{n})=\frac{1}{8n}h^{\top}J^{(k)}h+o\left(\frac{1}{n}\right), 
\end{equation}
where the remainder term $o(1/n)$ is uniform in $k$. 
This is shown by observing the identity
\begin{align*}
1-\Tr\sigma_{\theta_{0}}^{(k)} {R_{h/\sqrt{n}}^{(k)^2}} & =1-\Tr\sigma_{\theta_{0}}^{(k)}\left(I+\frac{1}{2\sqrt{n}}h^{i}L_{i}^{(k)}-B^{(k)}(h/\sqrt{n})\right)^{2}\\
 & =-\frac{1}{4n}h^{\top}J^{(k)}h-\Tr\sigma_{\theta_{0}}^{(k)} {B^{(k)}(h/\sqrt{n})}^2 \\
 & \qquad +2\,\Tr\sigma_{\theta_{0}}^{(k)}B^{(k)}(h/\sqrt{n})
 	+\frac{h^{i}}{\sqrt{n}}\,{\rm Re}\,\Tr\sigma_{\theta_{0}}^{(k)}L_{i}^{(k)}B^{(k)}(h/\sqrt{n})\\
 & =-\frac{1}{4n}h^{\top}J^{(k)}h + 2\,\Tr\sigma_{\theta_{0}}^{(k)}B^{(k)}(h/\sqrt{n})
 	+o\left(\frac{1}{n}\right).
\end{align*}
Here, \eqref{eq:niid_qlan_diff1} was used in the last equality.
Since the above quantity is of order $o(1/n)$ due to \eqref{eq:niid_qlan_diff2},
the equality \eqref{eqn:non-iid_rhoB} is proved. 

Now we are ready to prove (vi), i.e.,
\[
\Tr\rho_{\theta_{0}}^{(n)}\bar{R}_{h}^{(n)}e^{\frac{1}{2}\left(h^{i}\Delta_{i}^{(n)}-\frac{1}{2}h^{\top}J^{(\infty)} h\right)}
=e^{-\frac{1}{4}h^{\top}J^{(\infty)} h}\prod_{k=1}^{n}\Tr\sigma_{\theta_{0}}^{(k)}R_{h/\sqrt{n}}^{(k)} \; e^{\frac{1}{2 \sqrt{n}}h^iL_{i}^{(k)}}
\to 1 \quad (n\to\infty).
\]
We have from \eqref{eq:niid_qlan_diff1} and \eqref{eqn:non-iid_rhoB} that
\begin{align*}
 & \Tr\sigma_{\theta_{0}}^{(k)}R_{h/\sqrt{n}}^{(k)}\; e^{\frac{1}{2 \sqrt{n}}h^i L_i^{(k)}}\\
 &\qquad =\Tr\sigma_{\theta_{0}}^{(k)}\left(I+\frac{1}{2\sqrt{n}}h^i L_i^{(k)}-B^{(k)}(h/\sqrt{n})\right) \\
 &\qquad\qquad\qquad \times \left(I+\frac{1}{2\sqrt{n}}h^i L_i^{(k)}+\frac{1}{8n}\left(h^i L_i^{(k)}\right)^{2}+o\left(\frac{1}{n}\right) \right)\\
 &\qquad =1+\frac{1}{8n}h^{\top} J^{(k)}h+\frac{1}{4n}h^{\top}J^{(k)}h-\Tr\sigma_{\theta_{0}}^{(k)}B^{(k)}(h/\sqrt{n}) + o\left(\frac{1}{n}\right) \\
 &\qquad =1+\frac{1}{4n}h^{\top}J^{(k)}h+o\left(\frac{1}{n}\right).
\end{align*}
Therefore, 
\begin{align*}
\lim_{n\to\infty}\log\prod_{k=1}^{n}\Tr\sigma_{\theta_{0}}^{(k)}R_{h/\sqrt{n}}^{(k)}e^{\frac{1}{2 \sqrt{n}}h^{i}L_{i}^{(k)}} 
 & =\lim_{n\to\infty}\sum_{k=1}^{n}\log\left(1+\frac{1}{4n}h^{\top}J^{(k)}h+o\left(\frac{1}{n}\right)\right)\\
 & =
 \lim_{n\to\infty}\sum_{k=1}^{n}\left(\frac{1}{4n}h^{\top}J^{(k)}h+o\left(\frac{1}{n}\right)\right)\\
 & =\frac{1}{4}h^{\top}J^{(\infty)} h
\end{align*}
or equivalently,
\[
 \lim_{n\to\infty}\Tr\rho_{\theta_{0}}^{(n)}\bar{R}_{h}^{(n)}e^{\frac{1}{2}h^{i}\Delta_{i}^{(n)}}
 =e^{\frac{1}{4}h^{\top}J^{(\infty)} h}.
\]
This proves (vi), and the proof of \eqref{eq:infinitesimal_niid} is complete. 
\end{proof}

%----------------------------------------------------------------------------------------------------------------------------------------------------------------
\section{Proofs of Lemmas in Section \ref{sec:preliminary}}\label{app:ProofLemmas}
%----------------------------------------------------------------------------------------------------------------------------------------------------------------

In this section, we give detailed proofs of lemmas presented in Section \ref{sec:preliminary}. 

\subsection{Proof of Lemma \ref{lem:pure_gaussian}}

\begin{proof}
Let $\rho$ be the density operator of $N(0,J)$ on the irreducible representation Hilbert space $\H$. 
Then $\text{\ensuremath{\rho}}$ is pure if and only if $\Tr\rho^{2}=1$. 
On the other hand, due to the noncommutative Parseval identity \cite{holevo}, 
\begin{align*}
\Tr\rho^{2}  =\sqrt{\frac{\det S}{\pi^{d}}}\int_{\R^{d}}d\xi\left|\F_{\xi}[\rho]\right|^{2} 
 =\sqrt{\frac{\det S}{\pi^{d}}}\int_{\R^{d}}d\xi e^{-\xi^{\top}V\xi}
 =\sqrt{\frac{\det S}{\pi^{d}}}\sqrt{\frac{\pi^{d}}{\det V}}=\sqrt{\frac{\det S}{\det V}}.
\end{align*}
As a consequence, $\rho$ is pure if and only if $\det V=\det S$.  
\end{proof}

\subsection{Proof of Corollary \ref{lem:pure_double}}

\begin{proof}
We first remark that the dimension $d$ of the matrix $J$ is even because the skew-symmetric matrix $S={\rm Im}\,J$ is invertible. 
Now we set
\[ \hat{J}:=\begin{pmatrix}J & J\#J^{\top}\\ J\#J^{\top} & J^{\top} \end{pmatrix}. \] 
Then
\begin{equation*}
\det({\rm Im}\,\hat{J})
=\det \begin{pmatrix} S & 0\\ 0 & S^\top \end{pmatrix}
=(\det S)^{2}.
\end{equation*}
On the other hand, $\hat J$ is rewritten as
\[ \hat J =\frac{1}{2}\begin{pmatrix}I & I\\
I & -I
\end{pmatrix}\begin{pmatrix}V+J\#J^{\top} & \sqrt{-1}S\\
\sqrt{-1}S & V-J\#J^{\top}
\end{pmatrix}\begin{pmatrix}I & I\\
I & -I
\end{pmatrix}.
\]
Therefore, setting $S_{V}:=V^{-1/2}SV^{-1/2}$, we have
\begin{align*}
\det ({\rm Re}\,\hat{J}) & =\det \begin{pmatrix}V+J\#J^{\top} & 0\\
0 & V-J\#J^{\top}
\end{pmatrix}\\
 & =\det (V+J\#J^{\top})\det (V-J\#J^{\top})\\
 & =(\det V)^{2}\det \left(I+\sqrt{I+S_{V}^{2}} \right) \det \left(I-\sqrt{I+S_{V}^{2}} \right)\\
 & =(\det V)^{2}\det (-S_{V}^{2}) \\
 & =(\det S)^{2}. 
\end{align*}
Here we used Lemma \ref{lem:concreteA} in the third equality.
It then follows from Lemma \ref{lem:pure_gaussian} that $N(0,\hat{J})$ is a pure state. 
\end{proof}

\subsection{Proof of Lemma \ref{thm:bochoner_dshift}}

\begin{proof}
We first verify that the subspace $\mathring{\H}:={\rm Span}_{\C}\{\psi(\xi)\}_{\xi\in D}$ is dense in $\H$. 
Since $\psi$ is a cyclic vector, the subspace ${\rm Span}_{\C}\{\psi(\xi)\}_{\xi\in\R^{d}}$ is dense in $\H$. 
Further, given $\xi\in\R^{d}$, take an arbitrary sequence $\xi^{(n)}\in D$ that is convergent to $\xi$. 
Then
\begin{align*}
 \lim_{n\to\infty}\left\Vert \psi(\xi)-\psi(\xi^{(n)})\right\Vert ^{2} 
 & =2-2\, \lim_{n\to\infty}{\rm Re}\,\langle\psi(\xi), \psi(\xi^{(n)}) \rangle \\
 & =2-2\, \lim_{n\to\infty}{\rm Re\,}\, \left\{ e^{- \i\xi^{\top}S\xi^{(n)}}
 	\left\langle \psi, e^{- \i\left(\xi-\xi^{(n)}\right)^i X_i} \psi \right\rangle\right\} \\
 & =0.
\end{align*}
This proves that $\mathring{\H}$ is dense in $\H$. 

We next introduce a sesquilinear functional $F:\mathring{\H}\times\mathring{\H}\to\C$ by
\[
 F\left( \sum_{i=1}^{n}a_{i}\, \psi(\xi^{(i)}),\sum_{j=1}^{m}b_{j}\, \psi(\eta^{(j)}) \right)
 :=\sum_{i=1}^{n}\sum_{j=1}^{m}\bar{a}_{i}b_{j} \,\varphi(\xi^{(i)};\eta^{(j)}).
\]
We need to verify that $F$ is well-defined.  
Let  
\[
 \sum_{i=1}^{n}a_{i} \,\psi(\xi^{(i)})=\sum_{i=1}^{n'}a'_{i} \,\psi(\xi'^{(i)})
 \quad\mbox{and}\quad
 \sum_{j=1}^{m}b_{j}\, \psi(\eta^{(j)})=\sum_{j=1}^{m'}b'_{j}\, \psi(\eta'^{(j)}). 
\]
be different representations of the same vectors in $\mathring{\H}$. 
The well-definedness of $F$ is proved by showing the following series of equalities: 
\begin{equation}\label{eqn:proof_Lemma3.3_1}
\sum_{i=1}^{n}\sum_{j=1}^{m}\bar{a}_{i}b_{j} \, \varphi(\xi^{(i)};\eta^{(j)})
 =\sum_{i=1}^{n}\sum_{j=1}^{m'}\bar{a}_{i}b'_{j} \, \varphi(\xi^{(i)};\eta'^{(j)})
 =\sum_{i=1}^{n'}\sum_{j=1}^{m'}\bar{a}'_{i}b'_{j} \, \varphi(\xi'^{(i)};\eta'^{(j)}). 
\end{equation}
The first equality in \eqref{eqn:proof_Lemma3.3_1} is equivalent to
\[
 \sum_{j=1}^{m}b_{j} \, \varphi(\xi^{(i)};\eta^{(j)})-\sum_{j=1}^{m'}b'_{j} \, \varphi(\xi^{(i)};\eta'^{(j)})=0,
 \qquad (\forall \xi^{(i)}\in D),
\]
which is further equivalent to the following proposition: 
\begin{equation}\label{eqn:sandProp}
 \sum_{k=1}^{r} c_{k} \, \psi(\xi^{(k)})=0
 \quad \Longrightarrow \quad
 \sum_{k=1}^{r} c_{k} \, \varphi(\xi^{(0)};\xi^{(k)})=0, \qquad (\forall \xi^{(0)}\in D). 
\end{equation}
Since $0\prec \varphi \prec \varphi_I$, the antecedent of the above proposition \eqref{eqn:sandProp} implies that for any $r\in\N$, $\{\xi^{(i)}\}_{0\le i\le r}\subset D$, and $\{c_i\}_{0\le i\le r}\subset \C$ with $c_0=0$, 
\[
 0 \leq \sum_{i=0}^{r} \sum_{j=0}^{r} \bar{c}_{i} c_{j} \, \varphi(\xi^{(i)};\xi^{(j)})
    \leq \sum_{i=0}^{r} \sum_{j=0}^{r} \bar{c}_{i} c_{j} \, \varphi_{I}(\xi^{(i)};\xi^{(j)})
    =\left\Vert \sum_{k=0}^{r} c_{k} \, \psi(\xi^{(k)}) \right\Vert ^{2}=0.
\]
This shows that the vector $(c_{0}, c_1, \dots,c_{r})^\top$ belongs to the kernel of the positive-semidefinite  matrix 
\[ \left[\varphi(\xi^{(i)};\xi^{(j)})\right]_{0\leq i,j \leq r}. \]
As a consequence,  
\[ \sum_{k=1}^{r} c_{k}\,\varphi(\xi^{(0)};\xi^{(k)})=\sum_{k=0}^{r} c_{k} \,\varphi(\xi^{(0)};\xi^{(k)})=0, \]
proving the proposition \eqref{eqn:sandProp}.
The second equality in \eqref{eqn:proof_Lemma3.3_1} is proved in the same way. 

Now, fix an element $\phi_0\in\mathring{\H}$ arbitrarily.
Then the map $\phi_1 \mapsto F(\phi_1, \phi_0)$ is a bounded conjugate-linear functional on a dense subset $\mathring{\H}$ of $\H$, so that it is continuously extended to the totality of $\H$, and there is a vector $\phi_0^F \in\H$ such that
\[
\left\langle \phi_1, \phi_0^F \right\rangle =F(\phi_1, \phi_0). 
\]
Since the map $\phi_0\mapsto \phi_0^F$ is a bounded linear transformation on $\mathring{\H}$, it is continuously extended to $\H$, and there is a bounded operator $A$ satisfying $\phi_0^F=A\phi_0$ for all $\phi_0\in\H$. 
In summary, 
\[  F(\phi_1,\phi_0)=\langle \phi_1, A\phi_0\rangle. \]
Since $F(\phi_1,\phi_0)=\overline{F(\phi_0,\phi_1)}$, the operator $A$ is selfadjoint.  
Further, since $0 \prec \varphi \prec \varphi_I$, we see that $0\le A\le I$. 
Finally, since
\[ \varphi(\xi ; \eta)=F(\psi(\xi), \psi(\eta))=\langle \psi(\xi), A \psi(\eta) \rangle \]
on $D\times D$, it is continuously extended to $\R^d\times \R^d$. 
\end{proof}

\subsection{Proof of Lemma \ref{thm:commutant_qc}}

\begin{proof}
In view of (\ref{eq:commutant}), it suffices to show that $V(\varphi)$ and 
$e^{\i(\zeta_{c}^{i} \hat{X}_{c,i}+\zeta_{a}^{j} \hat{X}_{a,j})}$
commute for any $\zeta_{c}\in\R^{d_{c}}$ and $\zeta_{a}\in\R^{d_{a}}$.
For any $\xi_{c},\eta_{c}\in\R^{d_{c}}$, $\xi_{q},\eta_{q}\in\R^{d_{q}}$, and $\xi_{a},\eta_{a}\in\R^{d_{a}}$, 
\begin{align*}
 & \left\langle \psi(\xi_{c},\xi_{q},\xi_{a}),
 	e^{-\i(\zeta_{c}^{i} \hat{X}_{c,i}+\zeta_{a}^{j} \hat{X}_{a,j})}V(\varphi)
	 	e^{\i(\zeta_{c}^{i} \hat{X}_{c,i}+\zeta_{a}^{j} \hat{X}_{a,j})}
	 \psi(\eta_{c},\eta_{q},\eta_{a}) \right\rangle \\
 &\quad =e^{\i( -\zeta_{a}^{\top}S_{a}\xi_{a} +\zeta_{a}^{\top}S_{a}\eta_{a})}
 	\left\langle \psi(\xi_{c}+\zeta_{c},\xi_{q},\xi_{a}+\zeta_{a}), 
	V(\varphi) \psi(\eta_{c}+\zeta_{c},\eta_{q},\eta_{a}+\zeta_{a}) \right\rangle\\
 &\quad =e^{\i( -\zeta_{a}^{\top}S_{a}\xi_{a} +\zeta_{a}^{\top}S_{a}\eta_{a})}\varphi(\xi_{c}+\zeta_{c},\xi_{q},\xi_{a}+\zeta_{a};\eta_{c}+\zeta_{c},\eta_{q},\eta_{a}+\zeta_{a})\\
 &\quad =e^{\i( -\zeta_{a}^{\top}S_{a}\xi_{a} +\zeta_{a}^{\top}S_{a}\eta_{a})}e^{-\i(\xi_{a}+\zeta_{a})^{\top}S_{a}(\eta_{a}+\zeta_{a})}\varphi(\xi_{c}-\eta_{c},\xi_{q},\xi_{a}-\eta_{a};0,\eta_{q},0)\\
 &\quad =e^{-\i\xi_{a}^{\top}S_{a}\eta_{a}}\varphi(\xi_{c}-\eta_{c},\xi_{q},\xi_{a}-\eta_{a};0,\eta_{q},0)
\end{align*}
Since the last line is independent of $\zeta_{c}$ and $\zeta_{a}$, the proof is complete.
\end{proof}

\subsection{Proof of Lemma \ref{thm:sand_bounded}}

\begin{proof}
Let $L^\infty_{c.a.e.}(\R)$ denote the set of real-valued bounded Borel functions on $\R$ that are continuous almost everywhere. 
To prove Lemma \ref{thm:sand_bounded}, it suffices to verify that 
\begin{align}\label{eq:sand_func_simple}
&\lim_{n\to\infty}\Tr\rho^{(n)}e^{ \i \xi^{i}X_{i}^{(n)}}f(\zeta^{i}X_{i}^{(n)})A^{(n)}e^{\i\eta^{i}X_{i}^{(n)}} \\
&\qquad =\Tr\rho^{(\infty)}e^{ \i \xi^{i}X_{i}^{(\infty)}}f(\zeta^{i}X_{i}^{(\infty)})A^{(\infty)}e^{\i\eta^{i}X_{i}^{(\infty)}}
\nonumber
\end{align}
for all $\xi,\eta,\zeta\in\R^{d}$ and $f\in L^\infty_{c.a.e.}(\R)$.
In fact, since $f(\zeta^{i}X_{i}^{(n)})A^{(n)}$ are also uniformly bounded, \eqref{eq:sand_func} can be derived by applying \eqref{eq:sand_func_simple} and its complex conjugate recursively.

We first show that \eqref{eq:sand_weak} can be extended to all $\xi\in\R^{d}$.
For any $\varepsilon>0$ and $\xi\in\R^d$, there exists a $\tilde{\xi}\in\Q^{d}$ such that 
\[
\Tr\rho^{(\infty)}\left| \left( e^{\i\xi^{i}X_{i}^{(\infty)}}
 -e^{\i\tilde{\xi}^{i}X_{i}^{(\infty)}} \right)^* \right|^{2}<\varepsilon
\]
Then, by using the Schwarz inequality, 
\begin{align*}
 & \limsup_{n\to\infty}\left|\Tr\rho^{(n)}
 	\left(e^{\i\xi^{i}X_{i}^{(n)}}-e^{\i\tilde{\xi}^{i}X_{i}^{(n)}}\right)
	A^{(n)}e^{\i\eta^{i}X_{i}^{(n)}}\right|^{2} \\
 &\qquad \leq\limsup_{n\to\infty}\Tr\rho^{(n)}
 	\left| \left( e^{\i\xi^{i}X_{i}^{(n)}}-e^{\i\tilde{\xi}^{i}X_{i}^{(n)}} \right)^* \right|^{2}
 	\times\Tr\rho^{(n)}\left|A^{(n)}e^{\i\eta^{i}X_{i}^{(n)}}\right|^{2}\\
 &\qquad \leq M^{2} \;\Tr\rho^{(\infty)}
 	\left| \left(e^{\i\xi^{i}X_{i}^{(\infty)}}-e^{\i\tilde{\xi}^{i}X_{i}^{(\infty)}} \right)^* \right|^{2}\\
 &\qquad < M^{2} \varepsilon,
\end{align*}
where $M:=\sup_n \| A^{(n)} \|$. 
Similarly, we have
\[
 \left|\Tr\rho^{(\infty)}
 	\left(e^{\i\xi^{i}X_{i}^{(\infty)}}-e^{\i\tilde{\xi}^{i}X_{i}^{(\infty)}}\right)
	A^{(\infty)}e^{\i\eta^{i}X_{i}^{(\infty)}}\right|^{2}
 < M^{2} \varepsilon.
\]
It then follows from Lemma \ref{lem:extend_conv} that \eqref{eq:sand_weak} holds for all $\xi\in\R^{d}$ and 
$\eta\in\Q^{d}$. 
By applying a similar argument to $\eta$, we see that \eqref{eq:sand_weak} holds for all $\xi, \eta \in\R^{d}$. 

We next show that 
\begin{align}\label{eq:snad_weakR2}
&  \lim_{n\to\infty}\Tr\rho^{(n)}e^{\i\xi^{i}X_{i}^{(n)}}e^{\i\zeta^{i}X_{i}^{(n)}}A^{(n)}e^{\i\eta^{i}X_{i}^{(n)}} \\
 &\qquad\qquad
 =\Tr\rho^{(\infty)}e^{\i\xi^{i}X_{i}^{(\infty)}}e^{\i\zeta^{i}X_{i}^{(\infty)}}A^{(\infty)}e^{\i\eta^{i}X_{i}^{(\infty)}}
 \nonumber
\end{align}
for all $\xi,\eta,\zeta\in\R^{d}$. In fact, letting $S={\rm Im}\,J$,
\begin{align*}
 &\limsup_{n\to\infty}\left|\Tr\rho^{(n)}\left( e^{\i\xi^{i}X_{i}^{(n)}}e^{\i\zeta^{i}X_{i}^{(n)}}-e^{\i\xi^{\top}S\zeta}e^{\i\left(\xi+\zeta\right)^{i}X_{i}^{(n)}}\right) A^{(n)}e^{\i\eta^{i}X_{i}^{(n)}}\right|^{2}\\
 &\qquad \leq\limsup_{n\to\infty}\Tr\rho^{(n)}
 	\left| \left( e^{\i\zeta^{i}X_{i}^{(n)}}e^{\i\xi^{i}X_{i}^{(n)}}-e^{\i\xi^{\top}S\zeta}e^{\i\left(\xi+\zeta\right)^{i}X_{i}^{(n)}} \right)^* \right|^{2} \\
 &\qquad\qquad\qquad \times \Tr\rho^{(n)}\left|A^{(n)}e^{\i\eta^{i}X_{i}^{(n)}}\right|^{2}\\
 &\qquad \leq  M^{2} \;  \Tr\rho^{(\infty)}
 	\left| \left( e^{\i\zeta^{i}X_{i}^{(\infty)}}e^{\i\xi^{i}X_{i}^{(\infty)}}-e^{\i\xi^{\top}S\zeta}e^{\i\left(\xi+\zeta\right)^{i}X_{i}^{(\infty)}} \right)^* \right|^{2} \\
 &\qquad =0. 
\end{align*}
By using this, \eqref{eq:snad_weakR2} is proved as follows. 
\begin{align*}
 & \lim_{n\to\infty}\Tr\rho^{(n)}e^{\i\xi^{i}X_{i}^{(n)}}e^{\i\zeta^{i}X_{i}^{(n)}}A^{(n)}e^{\i\eta^{i}X_{i}^{(n)}} \\
 &\qquad\qquad =e^{\i\xi^{\top}S\zeta}\; \lim_{n\to\infty} \Tr\rho^{(n)}e^{\i\left(\xi+\zeta\right)^{i}X_{i}^{(n)}}A^{(n)}e^{\i\eta^{i}X_{i}^{(n)}} \\
 &\qquad\qquad =e^{\i\xi^{\top}S\zeta}\; \Tr\rho^{(\infty)}e^{\i\left(\xi+\zeta\right)^{i}X_{i}^{(\infty)}}A^{(\infty)}e^{\i\eta^{i}X_{i}^{(\infty)}}  \\
 &\qquad\qquad =\Tr\rho^{(\infty)}e^{\i\xi^{i}X_{i}^{(\infty)}}e^{\i\zeta^{i}X_{i}^{(\infty)}}A^{(\infty)}e^{\i\eta^{i}X_{i}^{(\infty)}}.
\end{align*}
In the second equality, \eqref{eq:sand_weak} is used. 

Now we are ready to prove (\ref{eq:sand_func_simple}).
Let $Z^{(n)}:=\zeta^{i}X_{i}^{(n)}$ for each $n\in\N\cup\{\infty\}$ and $\zeta\in\R^{d}$. 
According to (\ref{eq:snad_weakR2}), for any $f\in{\rm Span}_{\C}\{ e^{\i tx} \} _{t\in\R}$, we have
\begin{align}\label{eq:sand_bound_conv}
&\lim_{n\to\infty}\Tr\rho^{(n)}e^{\i\xi^{i}X_{i}^{(n)}} f(Z^{(n)})A^{(n)}e^{\i\eta^{i}X_{i}^{(n)}}  \\
&\qquad\qquad =\Tr\rho^{(\infty)}e^{\i\xi^{i}X_{i}^{(\infty)}} f(Z^{(\infty)})A^{(\infty)} e^{\i\eta^{i}X_{i}^{(\infty)}}.
 \nonumber
\end{align}
Our goal is to prove this identity for all $f\in L_{c.a.e.}^{\infty}(\R)$. 

Let 
\[
 \rho_\xi^{(n)}:=e^{ - \i\xi^{i}X_{i}^{(n)}}\rho^{(n)}e^{\i\xi^{i}X_{i}^{(n)}},
\]
and let $\mu_\xi^{(n)}$ be the classical probability measure on $\R$ that has the characteristic function 
\[
 \varphi_\xi^{(n)}(t) := \Tr \rho_\xi^{(n)} e^{\i t Z^{(n)}}.
\]
It then follows from (\ref{eq:sand_weak_origin}) that, for all $t\in\R$, 
\begin{align*}
\lim_{n\to\infty} \varphi_\xi^{(n)}(t)
&= \lim_{n\to\infty} \Tr \rho^{(n)} e^{ \i\xi^{i}X_{i}^{(n)}}e^{\i t\,\zeta^{i}X_{i}^{(n)}}e^{ -\i\xi^{i}X_{i}^{(n)}}  \\
 &=\Tr\rho^{(\infty)}e^{ \i\xi^{i}X_{i}^{(\infty)}}e^{\i t\,\zeta^{i}X_{i}^{(\infty)}}e^{ -\i\xi^{i}X_{i}^{(\infty)}}\\
 &=e^{-2\i t\,\zeta^{\top}S\xi}\;\Tr\rho^{(\infty)}e^{\i t\,\zeta^{i}X_{i}^{(\infty)}}\\
 &=\exp\left[ \i t\left(\zeta^{\top} h - 2\zeta^{\top}S\xi\right)-\frac{t^{2}}{2}\zeta^{\top}J\zeta \right].
\end{align*}
This shows that that
\[
  \mu_\xi^{(n)}\conv{}N\left(\zeta^{\top} h - 2\zeta^{\top}S\xi,\,\zeta^{\top}J\zeta \right).
\]
Let $p_{\xi}$ be the density function of the classical Gaussian distribution $N(\zeta^{\top} h - 2\zeta^{\top}S\xi,\,\zeta^{\top}J\zeta)$, and let $\psi_\xi:=\sqrt{p_\xi} \in L^2(\R)$. 
Then, the portmanteau lemma shows that, for all $f \in L_{c.a.e.}^{(\infty)}(\R)$,
\[
 \lim_{n\to\infty}\Tr \rho_\xi^{(n)} f(Z^{(n)})
 =\lim_{n\to\infty} \int_\R f(z) \mu_\xi^{(n)}(dz)
 =\int_\R f(z) p_\xi(z) dz,
 =\left\langle \psi_\xi, f \psi_\xi \right\rangle 
\]
where
\[ 
 \left\langle \psi_1, \psi_2 \right\rangle:=\int_\R \overline{\psi_1(z)} \psi_2(z) dz
 \qquad ( \psi_1, \psi_2\in L^2(\R) ).
\]
Now recall that
\[
 \overline{{\rm Span}}^{\rm SOT}_{\C}\left\{  e^{\i tx} \right\}_{t\in\R}=L^{\infty}(\R). 
\]
Thus, for all $\varepsilon>0$ and $f \in L_{c.a.e.}^{(\infty)}(\R)$, there exists a real-valued function
$\tilde{f}\in{\rm Span}_{\C}\{  e^{\i tx} \} _{t\in\R}$ such that 
\[
 \| (f-\tilde{f}) \psi_\xi \|^{2}=\langle \psi_\xi, ( f-\tilde{f} )^{2}  \psi_\xi\rangle <\varepsilon.
\]
Then by using the Schwarz inequality, 
\begin{align} \label{eqn:sand_bound_conv1}
 & \limsup_{n\to\infty}\left|\Tr\rho^{(n)}e^{\i\xi^{i}X_{i}^{(n)}}\left\{ f(Z^{(n)})-\tilde{f}(Z^{(n)})\right\} A^{(n)}e^{\i\eta^{i}X_{i}^{(n)}}\right|^{2}\\
 & \leq\limsup_{n\to\infty}\Tr \rho_\xi^{(n)} \left\{ f(Z^{(n)})-\tilde{f}(Z^{(n)})\right\}^{2}\times\Tr\rho^{(n)}\left|A^{(n)}e^{\i\eta^{i}X_{i}^{(n)}}\right|^{2}  \nonumber \\
 & \leq M^{2} \langle \psi_\xi, ( f-\tilde{f} )^{2} \psi_\xi \rangle  \nonumber\\
 & < M^{2} \varepsilon. \nonumber
\end{align}
Likewise, 
\begin{equation} \label{eqn:sand_bound_conv2}
\left|\Tr\rho^{(\infty)}e^{\i\xi^{i}X_{i}^{(\infty)}}\left\{ f(Z^{(\infty)})-\tilde{f}(Z^{(\infty)})\right\} A^{(\infty)}e^{\i\eta^{i}X_{i}^{(\infty)}}\right|^{2}< M^{2} \varepsilon.
\end{equation}
Now that \eqref{eq:sand_bound_conv} \eqref{eqn:sand_bound_conv1} \eqref{eqn:sand_bound_conv2} have been verified, the identity \eqref{eq:sand_func_simple} is an immediate consequence of Lemma \ref{lem:extend_conv}. \end{proof}

\subsection{Proof of Lemma \ref{thm:sand_uni_int}}

We begin with a brief review of the uniform integrability. 
Given sequences of quantum states $\{ \rho^{(n)}\} _{n\in\N}$ and observables $\{ B^{(n)}\} _{n\in\N}$ on Hilbert spaces $\{ \H^{(n)}\} _{n\in\N}$, 
we say that $B^{(n)}$ is {\em uniformly integrable} with respect to $\rho^{(n)}$ 
if for all $\varepsilon>0$, there exists $L>0$ that satisfies
\begin{equation} \label{eq:uniform_int_def}
\Tr\rho^{(n)} \left|B^{(n)}-h_L(B^{(n)})\right| <\varepsilon
\end{equation}
for all $n$. 
Here, the function $h_L$ is defined by
\begin{equation} \label{eqn:cutoffFnc}
h_L(x)=\begin{cases}
x & (|x|\leq L)\\
0 & (|x|>L)
\end{cases}.
\end{equation}
When $\Tr \rho^{(n)} |B^{(n)}|<\infty$ for all $n\in\N$, the uniform integrability is equivalent to saying that
\[
\limsup_{n\to\infty}\Tr\rho^{(n)} \left|B^{(n)}-h_L(B^{(n)}) \right| <\varepsilon.
\]
Note that \eqref{eq:uniform_int_def} implies $\Tr\rho^{(n)} |B^{(n)}|<L+\varepsilon$ for all $n$.
In other words, uniform integrability entails uniform boundedness of $\Tr\rho^{(n)} |B^{(n)}|$. 

\begin{proof}[Proof of Lemma \ref{thm:sand_uni_int}]
Set
\[
 \tilde{A}^{(n)}:=\left\{ \prod_{s=2}^{r_{1}}f_{s}(\xi_{s}^{i}X_{i}^{(n)})\right\} A^{(n)}\left\{ \prod_{t=2}^{r_{2}}g_{t}(\eta_{t}^{i}X_{i}^{(n)})\right\} ^{*}.
\]
Then $\tilde{A}^{(n)}$ is uniformly bounded, i.e., there is an $\tilde{M}>0$ such that
$\Vert \tilde{A}^{(n)}\Vert <\tilde{M}$
for all $n\in\N\cup\{\infty\}$. 
Further, set
\[
 Y^{(n)}:=\xi_{1}^{i}X_{i}^{(n)}, \quad
 Z^{(n)}:=\eta_{1}^{i}X_{i}^{(n)}, \quad
 \tilde{Y}^{(n)}:=Y^{(n)}+o_{1}^{(n)}, \quad
 \tilde{Z}^{(n)}:=Z^{(n)}+o_{2}^{(n)}.
\]
It then follows from the proof of Lemma \ref{thm:sand_bounded} that, for any $s, t \in\R$, 
\[
 \lim_{n\to\infty}\Tr\rho^{(n)}e^{\i s Y^{(n)}}\tilde{A}^{(n)}e^{\i t Z^{(n)}}
 =\Tr\rho^{(\infty)}e^{\i s Y^{(\infty)}}\tilde{A}^{(\infty)}e^{\i t Z^{(\infty)}}.
\]
and therefore
\[
 \lim_{n\to\infty}\Tr\rho^{(n)}e^{\i s \tilde{Y}^{(n)}}\tilde{A}^{(n)}e^{\i t\tilde{Z}^{(n)}}
 =\Tr\rho^{(\infty)}e^{\i s Y^{(\infty)}}\tilde{A}^{(\infty)}e^{\i t Z^{(\infty)}}.
\]
We can further deduce from Lemma \ref{thm:sand_bounded} that, for any $L>0$, 
\begin{equation}\label{eqn:proof_lem3.7_1}
 \lim_{n\to\infty}\Tr\rho^{(n)}f_L (\tilde{Y}^{(n)}) \tilde{A}^{(n)} g_L(\tilde{Z}^{(n)})
 =\Tr\rho^{(\infty)}f_L(Y^{(\infty)}) \tilde{A}^{(\infty)} g_L (Z^{(\infty)}),
\end{equation}
where $f_L:=h_L\circ f_1$ and $g_L:=h_L\circ g_1$ are bounded functions. 
Our goal is to prove that $f_L$ and $g_L$ in \eqref{eqn:proof_lem3.7_1} can be replaced with $f_1$ and $g_1$ if both $f_1(\tilde Y^{(n)})^2$ and $g_1(\tilde Z^{(n)})^2$ are uniformly integrable under $\rho^{(n)}$. 

As stated in the preliminary remark of this subsection, there exists a $K>0$ that fulfills 
\begin{equation} \label{eqn:proof_lem3.7_2}
 \max\left\{ \Tr\rho^{(n)} f_{1}(\tilde Y^{(n)})^{2},\; \Tr\rho^{(n)} g_{1}(\tilde Z^{(n)})^{2} \right\} \leq K
\end{equation}
for all $n\in\N\cup\{\infty\}$.
In addition, for any $\varepsilon>0$, there exists an $L>0$ such that 
\begin{equation}\label{eqn:proof_lem3.7_3}
 \max\left\{ 
	\Tr\rho^{(n)}\left( f_{1}({\tilde Y}^{(n)})-f_L(\tilde{Y}^{(n)}) \right)^{2},\; 
	\Tr\rho^{(n)}\left( g_{1}({\tilde Z}^{(n)})-g_L({\tilde Z}^{(n)}) \right)^{2}
 \right\} <\varepsilon 
\end{equation}
for all $n\in\N\cup\{\infty\}$.
Observe that, for all $n\in\N$, 
\begin{align}\label{eqn:proof_lem3.7_4}
 & \left|\Tr\rho^{(n)}f_{1}({\tilde Y}^{(n)})\tilde{A}^{(n)}g_{1}({\tilde Z}^{(n)})
 	-\Tr\rho^{(\infty)}f_{1}(Y^{(\infty)})\tilde{A}^{(\infty)}g_{1}(Z^{(\infty)})\right| \\
 & \qquad \leq\left|\Tr\rho^{(n)}f_{1}({\tilde Y}^{(n)})\tilde{A}^{(n)}g_{1}({\tilde Z}^{(n)})
 	-\Tr\rho^{(n)}f_L({\tilde Y}^{(n)})\tilde{A}^{(n)}g_L({\tilde Z}^{(n)})\right|  \nonumber \\
 & \qquad\quad +\left|\Tr\rho^{(n)}f_L({\tilde Y}^{(n)})\tilde{A}^{(n)}g_L({\tilde Z}^{(n)})
 	-\Tr\rho^{(\infty)}f_L(Y^{(\infty)})\tilde{A}^{(\infty)}g_L(Z^{(\infty)})\right|  \nonumber \\
 & \qquad\quad +\left|\Tr\rho^{(\infty)}f_L(Y^{(\infty)})\tilde{A}^{(\infty)}g_L(Z^{(\infty)})
 	-\Tr\rho^{(\infty)}f_{1}(Y^{(\infty)})\tilde{A}^{(\infty)}g_{1}(Z^{(\infty)})\right|. \nonumber
\end{align}
The second line in \eqref{eqn:proof_lem3.7_4} is evaluated as follows.
For any $\varepsilon>0$, take $L>0$ satisfying \eqref{eqn:proof_lem3.7_3}. 
Then by using \eqref{eqn:proof_lem3.7_2},
\begin{align} \label{eqn:proof_lem3.7_5}
 & \left|\Tr\rho^{(n)}f_{1}({\tilde Y}^{(n)})\tilde{A}^{(n)}g_{1}({\tilde Z}^{(n)})
 	-\Tr\rho^{(n)}f_L ({\tilde Y}^{(n)})\tilde{A}^{(n)}g_L({\tilde Z}^{(n)})\right|  \\
 & \qquad \leq\left|\Tr\rho^{(n)}\left\{ f_{1}({\tilde Y}^{(n)})-f_L({\tilde Y}^{(n)})\right\}
 	\tilde{A}^{(n)}g_{1}({\tilde Z}^{(n)})\right|   \nonumber \\
 & \qquad\qquad
	+\left|\Tr\rho^{(n)}f_L({\tilde Y}^{(n)})\tilde{A}^{(n)}
 	\left\{ g_{1}({\tilde Z}^{(n)})-g_L({\tilde Z}^{(n)})\right\} \right|  \nonumber \\
 & \qquad
 	\leq \sqrt{ \Tr\rho^{(n)} \left\{  f_{1}({\tilde Y}^{(n)})-f_L({\tilde Y}^{(n)}) \right\}^2 
	\times\Tr\rho^{(n)}\left|\tilde{A}^{(n)}g_{1}({\tilde Z}^{(n)})\right|^2}   \nonumber \\
 & \qquad\qquad
 	+ \sqrt{ \Tr\rho^{(n)} \left| \left( f_L ({\tilde Y}^{(n)}) \tilde{A}^{(n)} \right)^* \right|^{2}
		\times \Tr\rho^{(n)} \left\{ g_{1}({\tilde Z}^{(n)})-g_L({\tilde Z}^{(n)})\right\}^{2} }  \nonumber \\
 & \qquad < 2\tilde{M} \sqrt{ \varepsilon K}.  \nonumber
\end{align}
The last line in \eqref{eqn:proof_lem3.7_4} is evaluated just by setting $n=\infty$ in \eqref{eqn:proof_lem3.7_5}. 
Finally, the third line in \eqref{eqn:proof_lem3.7_4} is evaluated as follows: 
because of \eqref{eqn:proof_lem3.7_1}, for any $\varepsilon>0$, there is an $N\in\N$ such that $n\ge N$ implies
\[
 \left|\Tr\rho^{(n)}f_L({\tilde Y}^{(n)})\tilde{A}^{(n)}g_L({\tilde Z}^{(n)})
 	-\Tr\rho^{(\infty)}f_L(Y^{(\infty)})\tilde{A}^{(\infty)}g_L(Z^{(\infty)})\right| 
 <\varepsilon.
\]
Putting these evaluations together, we have
\[
  \left|\Tr\rho^{(n)}f_{1}({\tilde Y}^{(n)})\tilde{A}^{(n)}g_{1}({\tilde Z}^{(n)})
 	-\Tr\rho^{(\infty)}f_{1}(Y^{(\infty)})\tilde{A}^{(\infty)}g_{1}(Z^{(\infty)})\right|
 < 4\tilde{M} \sqrt{ \varepsilon K}+\varepsilon.
\]
Since $\varepsilon>0$ is arbitrary, the proof is complete. 
\end{proof}

\subsection{Proof of Corollary \ref{cor:lecam3_sand}}

\begin{proof}
The first assertion (\ref{eq:lecam3_sand4}) immediately follows from the conventional quantum Le Cam third lemma \cite[Corollary 7.5]{qcontiguity}.
We focus our attention on the proof of \eqref{eq:lecam3_sand1} and \eqref{eq:lecam3_sand2}. 

The basic observation for the proof of \eqref{eq:lecam3_sand1} is that the square-root likelihood ratio $R_{h}^{(n)}$ is an (unbounded) function of $X^{(n)}$.
Therefore, in order to invoke the extended version of the sandwiched L\'evy-Cram\'er continuity theorem (Lemma \ref{thm:sand_uni_int}), we need to show the uniform integrability.
As a matter of fact, uniform integrability of $(R_{h}^{(n)}+o_{L^{2}}(\rho_{\theta_{0}}^{(n)}) )^{2}$ under $ \rho_{\theta_{0}}^{(n)}$ has been shown in \cite[Theorem 6.2]{qcontiguity}. 
For the sake of the reader's convenience, however, we give a simplified proof.

Let $\sigma^{(\infty)}$ be the density operator of $N(0,J)$ and let $\{ \Delta_{i}^{(\infty)}\} _{i=1}^{d}$ be the corresponding canonical observables. 
Since the model $\rho_\theta^{(n)}$ is q-LAN at $\theta_0$, for any $\varepsilon>0$, there exists an $L>0$ such that
\begin{align*}
 & \limsup_{n\to\infty}\Tr\rho_{\theta_{0}}^{(n)}\left\{ \left( R_{h}^{(n)}+o_{L^{2}}(\rho_{\theta_{0}}^{(n)})\right)^{2}-h_L\left(\left\{ R_{h}^{(n)}+o_{L^{2}}(\rho_{\theta_{0}}^{(n)})\right\} ^{2}\right)\right\} \\
 &\qquad \leq\limsup_{n\to\infty}\left\{ 1-\Tr\rho_{\theta_{0}}^{(n)}h_L \left(\left\{ R_{h}^{(n)}+o_{L^{2}}(\rho_{\theta_{0}}^{(n)})\right\} ^{2}\right)\right\} \\
 &\qquad =\limsup_{n\to\infty}\left\{ 1-\Tr\rho_{\theta_{0}}^{(n)}h_L \left(e^{h^{i}\Delta_{i}^{(n)}-\frac{1}{2}h^{\top}Jh+o_{D}(h^{i}\Delta_{i}^{(n)},\,\rho_{\theta_{0}}^{(n)})}\right)\right\} \\
 &\qquad =
 	1-\Tr \sigma^{(\infty)} h_L\left(e^{h^{i}\Delta_{i}^{(\infty)}-\frac{1}{2}h^{\top}Jh}\right) \\
 &\qquad <\varepsilon. 
\end{align*}
Here, the function $h_L$ is defined by \eqref{eqn:cutoffFnc}, and the last equality is guaranteed by
the quantum L\'evy-Cram\'er continuity theorem 
(cf., Lemma \ref{thm:sand_bounded} with $A^{(n)}=I^{(n)}$).
This proves that $(R_{h}^{(n)}+o_{L^{2}}(\rho_{\theta_{0}}^{(n)}) )^{2}$ is uniformly integrable.

Now we prove \eqref{eq:lecam3_sand1}. 
Since $\{X_k^{(n)}\}_{1\le k \le r}$ is a $D$-extension of $\{\Delta_i^{(n)}\}_{1\le i \le d}$, 
\[ \Delta_i^{(n)}=F_i^k X_k^{(n)}. \] 
It then follows from the definition of q-LAN that, for all $h\in\R^{d}$,
\[
 R_{h}^{(n)}=\exp\left\{ \frac{1}{2}\left((Fh)^{i}X_{i}^{(n)}-\frac{1}{2}h^{\top}Jh+o_{h}^{(n)}\right)\right\} 
 -o_{L^{2}}(\rho_{\theta_{0}}^{(n)}),
\]
where $o_{h}^{(n)}=o_{D}((Fh)^{i}X_{i}^{(n)},\rho_{\theta_{0}}^{(n)})$, and $J=F^{\top}\Sigma F$.
Since $(R_{h}^{(n)}+o_{L^{2}}(\rho_{\theta_{0}}^{(n)}))^{2}$ is uniformly integrable for all $h$ under $\rho_{\theta_{0}}^{(n)}$, we can conclude from the extended version of the sandwiched L\'evy-Cram\'er continuity theorem (Lemma \ref{thm:sand_uni_int}) that
\begin{align*}
& \lim_{n\to\infty}\Tr\rho_{\theta_{0}}^{(n)}R_{h_{1}}^{(n)}A^{(n)}R_{h_{2}}^{(n)} \\
&\qquad =\lim_{n\to\infty}\Tr\rho_{\theta_{0}}^{(n)}e^{\frac{1}{2}\left((Fh_{1})^{i}X_{i}^{(n)}-\frac{1}{2}h_{1}^{\top}Jh_{1}+o_{h_{1}}^{(n)}\right)}A^{(n)}e^{\frac{1}{2}\left((Fh_{2})^{i}X_{i}^{(n)}-\frac{1}{2}h_{2}^{\top}Jh_{2}+o_{h_{2}}^{(n)}\right)}\\
 &\qquad =\Tr\rho_{0}^{(\infty)}e^{\frac{1}{2}\left((Fh_{1})^{i}X_{i}^{(\infty)}-\frac{1}{2}h_{1}^{\top}Jh_{1}\right)}A^{(\infty)}e^{\frac{1}{2}\left((Fh_{2})^{i}X_{i}^{(\infty)}-\frac{1}{2}h_{2}^{\top}Jh_{2}\right)}\\
 &\qquad =\Tr\rho_{0}^{(\infty)}R_{h_{1}}^{(\infty)}A^{(\infty)}R_{h_{2}}^{(\infty)}.
\end{align*}

Finally, \eqref{eq:lecam3_sand2} immediately follows from \eqref{eq:lecam3_sand1} and the fact that the singular parts asymptotically vanish \cite[Corollary 7.5]{qcontiguity},
 in that, 
\begin{equation*} %\label{eq:singular}
\lim_{n\to\infty}\Tr\left|\rho_{\theta_{0}+h/\sqrt{n}}^{(n)}-R_{h}^{(n)}\rho_{\theta_{0}}^{(n)}R_{h}^{(n)}\right|=\lim_{n\to\infty}\left\{ 1-\Tr\rho_{\theta_{0}}^{(n)}R_{h}^{(n)^{2}}\right\} =0.
\end{equation*}
The proof is complete. 
\end{proof}

\subsection{Proof of Lemma \ref{lem:asymptoticCCR}}

\begin{proof}
The Hilbert-Schmidt norm under consideration is calculated as 
\begin{align*}
& \left\| W^{(n)}(\xi)W^{(n)}(\eta)\sqrt{\rho^{(n)}}
	-e^{\i \xi^\top S\eta}\; W^{(n)}(\xi+\eta)\sqrt{\rho^{(n)}} \right\|_{\rm HS}^2 \\
&\qquad  =\Tr \sqrt{\rho^{(n)}} W^{(n)}(\eta)^* W^{(n)}(\xi)^* W^{(n)}(\xi) W^{(n)}(\eta) \sqrt{\rho^{(n)}} \\
&\qquad\qquad   +\Tr \sqrt{\rho^{(n)}} W^{(n)}(\xi+\eta)^* W^{(n)}(\xi+\eta) \sqrt{\rho^{(n)}} \\
&\qquad\qquad -2\,{\rm Re}\left\{ e^{\i \xi^\top S\eta}\cdot
	\Tr \sqrt{\rho^{(n)}} W^{(n)}(\eta)^* W^{(n)}(\xi)^* W^{(n)}(\xi+\eta) \sqrt{\rho^{(n)}}  \right\} \\
&\qquad =2 -2\, {\rm Re}\left\{ e^{\i \xi^\top S\eta}\cdot
	\Tr \rho^{(n)} W^{(n)}(-\eta) W^{(n)}(-\xi) W^{(n)}(\xi+\eta) \right\}.
\end{align*}
Letting $W(\xi):=e^{\i\xi^{i}X_{i}}$, we see from the assumption $(X^{(n)},\rho^{(n)}) \conv{} N(0,J)$ that
\begin{align*}
 & \lim_{n\to\infty} \Tr \rho^{(n)} W^{(n)}(-\eta) W^{(n)}(-\xi) W^{(n)}(\xi+\eta) \\
 &\qquad = \Tr\rho^{(\infty)} W(-\eta) W(-\xi) W(\xi+\eta) \\
 &\qquad = e^{\i (-\eta)^i(-\xi)^j S_{ij}}\; \Tr\rho^{(\infty)} W(-\eta-\xi) W(\xi+\eta) \\
 &\qquad = e^{\i \eta^\top S\xi},
\end{align*}
where $\rho^{(\infty)}$ is the density operator of $N(0,J)$. Since $S$ is real skew-symmetric, 
\begin{align*}
 & \lim_{n\to\infty}\left\| W^{(n)}(\xi)W^{(n)}(\eta)\sqrt{\rho^{(n)}}
	-e^{\i \xi^\top S\eta}\; W^{(n)}(\xi+\eta)\sqrt{\rho^{(n)}} \right\|_{\rm HS}^2 \\
 &\qquad =2 -2\, {\rm Re}\left\{ e^{\i \xi^\top S\eta}\cdot e^{\i \eta^\top S\xi} \right\}
 =0.
\end{align*}
This proves the claim.
\end{proof}

%----------------------------------------------------------------------------------------------------------------------------------------------------------------
\section{Proofs of Theorems in Section \ref{sec:applications}}\label{app:ProofApplications}
%----------------------------------------------------------------------------------------------------------------------------------------------------------------

In this section, we give detailed proofs of theorems presented in Section \ref{sec:applications}. 

\subsection{Chain of convergence}

We begin with the following Lemma, which is elementary but is useful in later applications.

\begin{lemma}[Chain of convergence]\label{lem:extend_conv}
Let $X$ and $Y$ be sets and $\left\langle Z,d\right\rangle$ be a metric space. 
Suppose that sequences of functions $F_{n}:X\to Z$ and $G_{n}:Y\to Z$ for $n\in\N\cup\{\infty\}$ satisfy the following conditions: 
\begin{equation*}%\label{eq:extend_conv1}
 \lim_{n\to\infty}F_{n}(x)=F_{\infty}(x),\qquad (\forall x\in X),
\end{equation*}
and for all $\varepsilon>0$ and $y\in Y$, there exists $x\in X$ satisfying 
\begin{equation*} %\label{eq:extend_conv2}
 \limsup_{n\to\infty}d\left(G_{n}(y), F_{n}(x)\right)<\varepsilon
 \quad\mbox{and}\quad 
 d\left(G_{\infty}(y), F_{\infty}(x)\right)<\varepsilon.
\end{equation*}
Then
\begin{equation*} %\label{eq:extend_conv0}
\lim_{n\to\infty}G_{n}(y)=G_{\infty}(y), \qquad (\forall y\in Y). 
\end{equation*}
\end{lemma}

\begin{proof}
Take $\varepsilon>0$ and $y\in Y$ arbitrarily.
Then, there exist $x\in X$ and $N\in\N$ such that 
\begin{equation*} %\label{eq:extend_conv4}
d\left(F_{n}(x),F_{\infty}(x)\right)<\varepsilon
\quad \mbox{and} \quad 
d\left(G_{n}(y),F_{n}(x)\right)<2\varepsilon 
\end{equation*}
for all $n\geq N$, and
\[ d\left(G_{\infty}(y), F_{\infty}(x)\right)<\varepsilon. \]
Thus
\[
d\left(G_{n}(y),G_{\infty}(y)\right)
\leq d\left(G_{n}(y),F_{n}(x)\right)+d\left(F_{n}(x),F_{\infty}(x)\right)+d\left(F_{\infty}(x),G_{\infty}(y)\right)
<4\varepsilon, 
\]
proving the claim. 
\end{proof}

\subsection{Proof of Lemma \ref{lem:equi_law}}

\begin{proof}
Let $\L$ be the classical distribution obtained by applying the shifted POVM $M-h$ to $\phi_{h}$, that is, $\L(B):=\phi_h((M-h)(B))$, which is independent of $h$ by assumption. 
Let $m$ be the first moment of $\L$.
Then, for each $i=1,\dots, d$, 
\begin{align*}
 \int_{\R^{d}} (z-h)^{i} \phi_h ((M-m)(dz))
 &= \int_{\R^{d}+m}\left(x-m-h\right)^{i} \phi_h(M(dx))  \\
 &= \int_{\R^{d}+m-h} (y-m)^{i} \phi_h ((M-h)(dy)) \\
 &= \int_{\R^{d}}(y-m)^{i} \L(dy)=0.
\end{align*}
This implies that $M-m$ is an unbiased estimator for the parameter $h$ of $\phi_h$. 
It then follows from the quantum Cram\'er-Rao type inequality \cite{holevo} that
\begin{align*}
c_{G}^{(H)}
 &\le G_{ij}\int_{\R^{d}}(z-h)^i (z-h)^j \phi_h ((M-m)(dz)) \\ 
 &= G_{ij}\int_{\R^{d}+m-h}(y-m)^{i}(y-m)^{j} \phi_h ((M-h)(dy))  \\
 &= G_{ij}\left\{ \int_{\R^{d}}y^{i}y^{j}\L(dy)-m^{i}m^{j}\right\}. 
\end{align*}
As a consequence,
\begin{align*}
  G_{ij}\int_{\R^{d}}\left(x-h\right)^{i} \left(x-h\right)^{j} \phi_h(M(dx))
 = G_{ij}\int_{\R^{d}-h} y^{i}y^{j} \phi_h ((M-h)(dy)) 
 \geq c_{G}^{(H)}, 
\end{align*}
proving the claim.
\end{proof}

Note that this result is closely related to what Holevo established in \cite{holevo} within the framework of group covariant measurement, where the achievability of the lower bound was also discussed.

\subsection{Proof of Theorem \ref{thm:regular}}

\begin{proof}
By applying the representation Theorem \ref{thm:qRep} to the sequence
\[
 N^{(n)}:=M^{(n)h}+h =\sqrt{n} (M^{(n)}-\theta_{0} )
\]
of POVMs that is independent of $h\in\R^{d}$, we see that there exists a POVM $N$ on $\phi_{h}\sim N(({\rm Re}\,\tau)h,\Sigma)$ such that 
\[
 \left(N^{(n)},\rho_{\theta_{0}+h/\sqrt{n}}^{(n)}\right)\conv h\left(N,\phi_{h}\right)
 \qquad (\forall h\in\R^d).
\]
Let $\L_h$ denote the classical probability distribution of outcomes of $N$ applied to $\phi_h$. 
Then, by construction, for any $B\in\B(\R^d)$ that satisfies $\L_h(\partial B)=0$,
\begin{align*}
  \L_h(B)
  &=\phi_h(N(B)) \\
  &=\lim_{n\to\infty} \Tr \rho_{\theta_{0}+h/\sqrt{n}}^{(n)} N^{(n)}(B) \\
  &=\lim_{n\to\infty} \Tr \rho_{\theta_{0}+h/\sqrt{n}}^{(n)} (M^{(n)h}+h)(B) \\
  &=\lim_{n\to\infty} \Tr \rho_{\theta_{0}+h/\sqrt{n}}^{(n)} M^{(n)h}(B-h) \\
  &=\L(B-h).
\end{align*}
Here, $\L$ is the limit distribution of $M^{(n) h}$ under $\rho_{\theta_{0}+h/\sqrt{n}}^{(n)}$, which is independent of $h$ by regularity. 
As a consequence,
\[
 \phi_h((N-h)(B))=\phi_h(N(B+h))=\L_h(B+h)=\L(B).
\]
Since the last side is independent of $h$, $N$ is equivalent in law. 
Thus, Lemma \ref{lem:equi_law} yields
\[
 \int_{\R^d} G_{ij} (x-h)^i (x-h)^j \phi_h(N(dx)) \ge c_G^{(rep)},
\]
which implies \eqref{eq:regular_bound}, 
and the portmanteau lemma proves \eqref{eq:regular_bound2}. 
\end{proof}

\subsection{Proof of Theorem \ref{thm:achieve}}

\begin{proof}
Let $\{ X_{i}\} _{i=1}^{r}$ be the canonical observables of 
$\phi_{h} \sim N(({\rm Re}\,\tau)h,\Sigma)$, and let
\begin{equation}\label{eqn:XtoY}
 Y_{\star\, i}:=(K_\star)_{i}^{j}X_{j},
\end{equation}
where $K_\star$ is the $r\times d$ matrix $K$ that achieves the minimum in the definition \eqref{eq:rep_bound} of $c_{G}^{(rep)}$. 
Thus, 
\[
 c_{G}^{(rep)}= \Tr G\,{\rm Re}\,Z_\star+\Tr \left|\sqrt{G}\,{\rm Im}\,Z_\star\,\sqrt{G}\right| =\Tr G V_\star,
\]
where
\[
 Z_\star=K_\star^\top \Sigma K_\star
\]
is the complex $d \times d$ matrix whose $(i,j)$th entry is $\phi_{0}(Y_{\star\, j}Y_{\star\, i})$, 
and
\[
 V_\star:={\rm Re}\,Z_\star+\sqrt{G}^{-1} \left|\sqrt{G}\,{\rm Im}\,Z_\star\,\sqrt{G}\right| \sqrt{G}^{-1}, 
\]
is a real $d \times d$ matrix. 
Note that $V_\star\ge Z_\star$, since
\[
 V_\star \ge {\rm Re}\,Z_\star+ \i \sqrt{G}^{-1} \sqrt{G}\,{\rm Im}\,Z_\star\,\sqrt{G} \sqrt{G}^{-1}=Z_\star.
\]
By analogy to \eqref{eqn:XtoY}, we introduce a sequence of transformed observables
\[
 Y_{\star\, i}^{(n)}:=(K_\star)_{i}^{j}X_{j}^{(n)}
\]
on $\H^{(n)}$. 
 
Let us consider another quantum Gaussian state $\tilde{\phi}\sim N(0,\tilde{Z}_\star)$ with $\tilde{Z}_\star:=V_\star-Z_\star$ and canonical observables $\tilde{Y}=(\tilde{Y}_{1},\dots,\tilde{Y}_{d})$ on an ancillary Hilbert space $\tilde{\H}$. 
Accordingly, for each $n\in\N$, we introduce a quantum state $\sigma_\star^{(n)}$ and observables $\tilde{Y_\star}^{(n)}=(\tilde{Y}_{\star\,1}^{(n)},\dots,\tilde{Y}_{\star\,d}^{(n)})$ on an ancillary Hilbert space $\tilde{\H}^{(n)}$ satisfying
\begin{equation} \label{eqn:achieve_toZ*}
 (\tilde{Y}_\star^{(n)},\sigma_\star^{(n)} )\conv{}N(0,\tilde{Z_\star}).
\end{equation}

A key observation is that the series of observables%
\footnote{
This construction was inspired by the optical heterodyne measurement \cite{holevo}.
}
\[
 \bar{Y}_{\star\,i}^{(n)}:=Y_{\star\,i}^{(n)}\otimes I+I\otimes\tilde{Y}_{\star\,i}^{(n)}\qquad(1\leq i\leq d)
\]
on the enlarged Hilbert spaces $\H^{(n)}\otimes \tilde\H^{(n)}$ exhibits
\begin{equation} \label{eq:achieve_toV}
 (\bar{Y}_\star^{(n)},\rho_{\theta_{0}+h/\sqrt{n}}^{(n)}\otimes\sigma_\star^{(n)})\conv{h} N(h,V_\star). 
\end{equation}
This can be verified by calculating the limit of the quasi-characteristic function
\begin{align} \label{eqn:quasiCharFnc}
 & \Tr\left(\rho_{\theta_{0}+h/\sqrt{n}}^{(n)}\otimes\sigma_\star^{(n)}\right)
 	\prod_{t=1}^T e^{\i\xi_{t}^{i}\bar{Y}_{\star\,i}^{(n)}} \\
 &\qquad =
    \left\{ \Tr \rho_{\theta_{0}+h/\sqrt{n}}^{(n)}
  	\prod_{t=1}^T e^{\i\xi_{t}^{i} (K_\star)_i^j X_{j}^{(n)}} \right\}  \times 
	 \left\{ \Tr \sigma_\star^{(n)}
  	\prod_{t=1}^T e^{\i\xi_{t}^{i}\tilde{Y}_{\star\,i}^{(n)}} \right\}, \nonumber
\end{align}
where $\{ \xi_{t}\} _{t=1}^T \subset\R^{d}$. 
In fact, since 
\[  (X^{(n)},\rho_{\theta_{0}+h/\sqrt{n}}^{(n)})\conv hN(({\rm Re}\,\tau)h,\Sigma), \]
which follows from the quantum Le Cam third lemma (Corollary \ref{cor:lecam3_sand}), 
the first factor in the second line of \eqref{eqn:quasiCharFnc} has the limit
\begin{align} \label{eqn:quasiCharFnc1}
 & \exp\left[ \sum_{t=1}^T \left\{ \i \, \xi_t^i (K_\star)_i^j  ({\rm Re}\,\tau)_{jk}  h^k 
 	-\frac{1}{2} \xi_t^i (K_\star)_i^k \cdot \xi_t^j (K_\star)_j^\ell \, \Sigma_{\ell k}\right\} \right. \\
 &\qquad\qquad\qquad\qquad  \left.
	-\sum_{t=1}^T \sum_{u=r+1}^T \xi_t^i (K_\star)_i^k \cdot \xi_u^j (K_\star)_j^\ell \, \Sigma_{\ell k}
	\right] \nonumber \\
 &\qquad = 
 	\exp\left[ \sum_{t=1}^T \left\{ \i \, \xi_t^i \delta_{ij} h^j
 	-\frac{1}{2} \xi_t^i \xi_t^j (Z_\star)_{ji} \right\}
	-\sum_{t=1}^T \sum_{u=r+1}^T \xi_t^i \xi_u^j (Z_\star)_{ji} \right]. \nonumber
\end{align}
Here, the equalities $K_\star^\top ({\rm Re}\,\tau)=I$ and $K_\star^\top \Sigma K_\star=Z_\star$ have been used.
On the other hand, due to the assumption \eqref{eqn:achieve_toZ*}, the second factor in the second line of \eqref{eqn:quasiCharFnc} has the limit
\begin{align} \label{eqn:quasiCharFnc2}
 \exp\left[ \sum_{t=1}^T \left\{ -\frac{1}{2} \xi_t^i \xi_t^j (\tilde Z_\star)_{ji} \right\}
	-\sum_{t=1}^T \sum_{u=r+1}^T \xi_t^i \xi_u^j (\tilde Z_\star)_{ji} \right].
\end{align}
Since $Z_\star+\tilde Z_\star=V_\star$, \eqref{eqn:quasiCharFnc}, \eqref{eqn:quasiCharFnc1}, and \eqref{eqn:quasiCharFnc2} yield
\begin{align*}
 & \lim_{n\to\infty}
 	\Tr\left(\rho_{\theta_{0}+h/\sqrt{n}}^{(n)}\otimes\sigma_\star^{(n)}\right)
 	\prod_{t=1}^T e^{\i\xi_{t}^{i}\bar{Y}_{\star\,i}^{(n)}} \\
 &\qquad = 
 	\exp\left[ \sum_{t=1}^T \left\{ \i \, \xi_t^i  \delta_{ij} h^j 
 	-\frac{1}{2} \xi_t^i \xi_t^j (V_\star)_{ji} \right\}
	-\sum_{t=1}^T \sum_{u=r+1}^T \xi_t^i \xi_u^j (V_\star)_{ji} \right].
\end{align*}
This is nothing but the quasi-characteristic function of the classical Gaussian shift model $N(h,V_\star)$, proving \eqref{eq:achieve_toV}. 

We next construct a sequence $M_\star^{(n)}$ of POVMs by means of functional calculus for $\bar{Y}_\star^{(n)}$. 
Since $\{ \bar{Y}_{\star\,i}^{(n)}\} _{i=1}^{d}$ do not in general commute, we need some elaboration.
For each positive integer $m\in\N$, define an indicator function $S^{(m)}:\R\to\{0,1\}$
by
\[
S^{(m)}(x):=\indicate_{(-\frac{1}{2m},\frac{1}{2m}]}(x)=\begin{cases}
1, & \text{if }x\in(-\frac{1}{2m},\frac{1}{2m}]\\
0, & \text{if }x\not\in(-\frac{1}{2m},\frac{1}{2m}]
\end{cases}.
\]
Then, for each $x\in\R$ and $n\in\N$, the map enjoys the identity
\begin{equation}
\sum_{k\in\Z}S^{(m)}\left(x-\frac{k}{m}\right)=1.\label{eq:S_achieve}
\end{equation}
For each pair $(n,m)$ of positive integers, define
\[
M_\star^{(n,m)}(\omega):=\left(\prod_{i=1}^{d}S^{(m)}\left(\bar{Y}_{\star\,i}^{(n)}-\omega_{i}\right)\right)\left(\prod_{i=1}^{d}S^{(m)}\left(\bar{Y}_{\star\,i}^{(n)}-\omega_{i}\right)\right)^{*},
\]
where
\[
\omega=\left(\omega_{1},\dots,\omega_{d}\right)\in\Omega^{(m)}:=\left\{ \left(\frac{z_{1}}{m},\dots,\frac{z_{d}}{m}\right)\mid z_{1},\dots,z_{d}\in\Z\right\} .
\]
It then follows from (\ref{eq:S_achieve}) that $M_\star^{(n,m)}$ is a
POVM on $\H^{(n)}\otimes\tilde{\H}$ whose outcomes take values on
$\Omega^{(m)}$. 

Note that, due to (\ref{eq:achieve_toV}) and the quantum L\'evy-Cram\'er continuity theorem 
(cf., Lemma \ref{thm:sand_bounded} with $A^{(n)}=I^{(n)}$),
as well as the fact that the set of
discontinuity points of $S^{(m)}$ has Lebesgue measure zero, the
following equality holds:
\[
\lim_{n\to\infty}\Tr\left(\rho_{\theta_{0}+h/\sqrt{n}}^{(n)}\otimes\sigma_\star^{(n)}\right)M_\star^{(n,m)}(\omega)
=\int_{\R^{d}}p_{h}(x)\prod_{i=1}^{d}S^{(m)}(x_{i}-\omega_{i})dx,
\]
where $p_{h}(x)$ denotes the probability density function of $N(h,V_\star)$.
Note also that for each $t=(t_{i})\in\R^{d}$, 
the indicator function $\chi_{t}(x):=\indicate_{(-\infty,t]}(x)$
fulfills the following equality
\[
 \lim_{m\to\infty}\sum_{\omega\in\Omega^{(m)}}\chi_{t}(\omega)\prod_{i=1}^{d}S^{(m)}(x_{i}-\omega_{i})
 =\chi_{t}(x)
\]
for all $x\in\R^{d}$ but $x=t$. Combining these equalities, we have
\[
 \lim_{m\to\infty}\lim_{n\to\infty}\sum_{\omega\in\Omega^{(m)}}\chi_{t}(\omega)\Tr\left(\rho_{\theta_{0}+h/\sqrt{n}}^{(n)}\otimes\sigma_\star^{(n)}\right)M_\star^{(n,m)}(\omega)
 =\int_{\R^{d}}\chi_{t}(x)p_{h}(x)dx
\]
for all $h\in\R^{d}$ and $t\in\R^{d}$. 

As a consequence, the diagonal sequence trick shows that there exists a subsequence $\{m(n)\}_{n\in\N}$
such that
\[
\lim_{n\to\infty}\sum_{\omega\in\Omega^{(m(n))}}\chi_{t}(\omega)\Tr\left(\rho_{\theta_{0}+h/\sqrt{n}}^{(n)}\otimes\sigma_\star^{(n)}\right)M_\star^{(n,m(n))}(\omega)
=\int_{\R^{d}}\chi_{t}(x)p_{h}(x)dx
\]
for all $h\in\Q^{d}$ and $t\in\Q^{d}$. 
Setting $\bar{M}_\star^{(n)}:=M_\star^{(n,m(n))}$
and $\bar{\Omega}^{(n)}:=\Omega^{(m(n))}$, we get
\begin{equation}\label{eq:achieve_Q1}
\lim_{n\to\infty}\sum_{\omega\in\bar{\Omega}^{(n)}}\chi_{t}(\omega)\Tr\left(\rho_{\theta_{0}+h/\sqrt{n}}^{(n)}\otimes\sigma_\star^{(n)}\right)\bar{M}_\star^{(n)}(\omega)
=\int_{\R^{d}}\chi_{t}(x)p_{h}(x)dx
\end{equation}
for all $h\in\Q^{d}$ and $t\in\Q^{d}$. Moreover, since both sides
of (\ref{eq:achieve_Q1}) are monotone increasing in $t$, and the
right-hand side is continuous in $t$, (\ref{eq:achieve_Q1}) holds
for all $t\in\R^{d}$ and all $h\in\Q^{d}$. 

Now let, for each $t\in\R^{d}$, 
\[
M_{\star\,t}^{(n)}
:=\Tr_{\tilde{\H}^{(n)}}\left\{ \sum_{\omega\in\bar{\Omega}^{(n)}}\chi_{t}(\omega)\bar{M}_\star^{(n)}(\omega)\right\} \left(I^{(n)}\otimes\sigma_\star^{(n)}\right),
\]
where $\Tr_{\tilde{\H}^{(n)}}$ denotes the partial trace on $\tilde\H^{(n)}$. 
Then, $M_{\star\,t}^{(n)}$ is a resolution of identity on $\H^{(n)}$ satisfying
\begin{align}\label{eq:achieve_Q2}
\lim_{n\to\infty}\Tr\rho_{\theta_{0}+h/\sqrt{n}}^{(n)}M_{\star\,t}^{(n)} 
 & =\lim_{n\to\infty}\sum_{\omega\in\bar{\Omega}^{(n)}}\chi_{t}(\omega)\Tr\left(\rho_{\theta_{0}+h/\sqrt{n}}^{(n)}\otimes\sigma_\star^{(n)}\right)\bar{M}_\star^{(n)}(\omega) \\
 & =\int_{\R^{d}}\chi_{t}(x)p_{h}(x)dx  \nonumber
\end{align}
for all $t\in\R^{d}$ and all $h\in\Q^{d}$. 

Finally, we extend the identity \eqref{eq:achieve_Q2} to all $h\in\R^{d}$. 
Let 
\[ 
 R_{h}^{(n)}:=\Ratio\left( \left. \rho_{\theta_{0}+h/\sqrt{n}}^{(n)} \,\right|\, \rho_{\theta_{0}}^{(n)}\right)
 \quad\mbox{and}\quad
 R_{h}^{(\infty)}:=e^{\frac{1}{2}\left((Fh)^{i}X_{i}^{(\infty)}-\frac{1}{2}h^{\top}F^{\top}\Sigma Fh\right)}
\]
and fix $h\in\R^{d}$ and $t\in\R^{d}$ arbitrarily. 
Then for all $\varepsilon>0$, there exists an $\tilde{h}\in\Q^{d}$ that satisfies
\[
 \phi_{0}\left(\left(R_{h}^{(\infty)}-R_{\tilde{h}}^{(\infty)}\right)^{2}\right)<\varepsilon. 
\]
On the other hand, 
\begin{align} \label{eq:achieve_R1}
 & \limsup_{n\to\infty}\left|\Tr\left(\rho_{\theta_{0}+h/\sqrt{n}}^{(n)}-\rho_{\theta_{0}+\tilde{h}/\sqrt{n}}^{(n)}\right)
 	M_{\star\,t}^{(n)}\right| \\
 &\qquad =\limsup_{n\to\infty}\left|\Tr\left(R_{h}^{(n)}\rho_{\theta_{0}}^{(n)}R_{h}^{(n)}-R_{\tilde{h}}^{(n)}\rho_{\theta_{0}}^{(n)}R_{\tilde{h}}^{(n)}\right)M_{\star\,t}^{(n)}\right| \nonumber \\
 &\qquad \leq\limsup_{n\to\infty}\left\{ \left|\Tr\left(R_{h}^{(n)}-R_{\tilde{h}}^{(n)}\right)\rho_{\theta_{0}}^{(n)}R_{h}^{(n)}M_{\star\,t}^{(n)}\right| \right. \nonumber \\
 &\qquad\qquad \qquad\qquad 
  \left. +\left|\Tr R_{\tilde{h}}^{(n)}\rho_{\theta_{0}}^{(n)}\left(R_{h}^{(n)}-R_{\tilde{h}}^{(n)}\right)M_{\star\,t}^{(n)}\right|\right\}. \nonumber
\end{align}
Here, the second line follows from \eqref{eq:lecam3_sand2} in Corollary \ref{cor:lecam3_sand}, which tells us that the contribution of the singular parts of $\rho_{\theta_{0}+h/\sqrt{n}}^{(n)}$ are asymptotically negligible. 
By using Corollary \ref{cor:lecam3_sand}, 
the third line of \eqref{eq:achieve_R1} is evaluated as follows: 
\begin{align*}
 & \limsup_{n\to\infty}\left|\Tr\left(R_{h}^{(n)}-R_{\tilde{h}}^{(n)}\right)\rho_{\theta_{0}}^{(n)}R_{h}^{(n)}
 	M_{\star\,t}^{(n)}\right|^{2}\\
 &\qquad \leq\limsup_{n\to\infty}\Tr\rho_{\theta_{0}}^{(n)}\left(R_{h}^{(n)}-R_{\tilde{h}}^{(n)}\right)^{2}
 	\times \Tr\rho_{\theta_{0}}^{(n)}\left|  \left(R_{h}^{(n)}M_{\star\,t}^{(n)}\right)^*  \right|^{2}\\
 &\qquad \leq\limsup_{n\to\infty}\Tr\rho_{\theta_{0}}^{(n)}\left(R_{h}^{(n)}-R_{\tilde{h}}^{(n)}\right)^{2}\\
 &\qquad =\phi_{0}\left(\left(R_{h}^{(\infty)}-R_{\tilde{h}}^{(\infty)}\right)^{2}\right)<\varepsilon.
\end{align*}
Since the fourth line of \eqref{eq:achieve_R1} is evaluated similarly, 
we can conclude that
\begin{equation}\label{eq:achieve_R3}
\limsup_{n\to\infty}\left|\Tr\left(\rho_{\theta_{0}+h/\sqrt{n}}^{(n)}-\rho_{\theta_{0}+\tilde{h}/\sqrt{n}}^{(n)}\right)
	M_{\star\,t}^{(n)}\right|<2\sqrt{\varepsilon}
\end{equation}
In a quite similar way, we can prove that
\begin{equation}\label{eq:achieve_R4}
\left|\phi_{h}\left(M_{t}^{(\infty)}\right)-\phi_{\tilde{h}}\left(M_{\star\,t}^{(\infty)}\right)\right|<2\sqrt{\varepsilon}
\end{equation}
Now that we have established \eqref{eq:achieve_Q2} \eqref{eq:achieve_R3} \eqref{eq:achieve_R4}, 
Lemma \ref{lem:extend_conv} leads us to
\begin{align*}
\lim_{n\to\infty}\Tr\rho_{\theta_{0}+h/\sqrt{n}}^{(n)}M_{\star\,t}^{(n)} 
 =\phi_{h}\left(M_{\star\,t}^{(\infty)}\right)
 =\int_{\R^{d}}\chi_{t}(x)p_{h}(x)dx
\end{align*}
for all $t\in\R^{d}$ and $h\in\R^{d}$. 
To put it differently, letting $M_\star^{(n)}$ be the POVM that corresponds to the resolution of identity $M_{\star\,t}^{(n)}$, we have
\[
\left(M_\star^{(n)},\rho_{\theta_{0}+h/\sqrt{n}}^{(n)}\right)\conv hN(h,V_\star)\qquad(\forall h\in\R^{d}). 
\]
This completes the proof. 
\end{proof}

\subsection{Proof of Theorem \ref{thm:q_gauss_minimax}}

\begin{proof}
For the quantum Gaussian shift model $\phi_h\sim N(({\rm Re}\,\tau) h,\Sigma)$, take a family of unitary operators $\{ U(k) \} _{k\in\R^{d}}$ on $\H$ that satisfy
\[
\phi_{h}(U(k)^{*}AU(k))=\phi_{h+k}(A) \qquad (\forall A\in {\rm CCR}({\rm Im}\Sigma) ).
\]
Given a POVM $M$, let $M_t:=M(-\infty, t]$ be the corresponding resolution of identity, and let us define, for each $t\in\R^d$ and $L\in\N$, a bounded operator
\[
 \hat{N}_t^{(L)}
 :=\frac{1}{(2L)^{d}}\int_{[-L,L]^{d}} U(k)^{*}M_{t+k}U(k) dk,
\]
where the integration is taken in the weak operator topology (WOT). 
It is not difficult to show that $\{\hat{N}_t^{(L)}\}_{t\in\R^d}$ is a resolution of identity for all $L\in\N$. 
From this resolution of identity, we shall construct a POVM $N$ that is equivalent in law and surpasses the original $M$.
Here we follow the method used in Step 2 of the proof of Theorem \ref{thm:qRep}.

Take a cyclic vector $\psi$ on the Hilbert space $\H$, and consider the sandwiched coherent state representation 
\[
 \varphi_t^{(L)}(\xi ; \eta):=\left\langle e^{\i \xi^i X_i} \psi, \hat{N}_t^{(L)} e^{\i \eta^i X_i} \psi \right\rangle.
\]
Since $\left| \varphi_t^{(L)}(\xi ; \eta) \right|\le 1$ for all $L\in\N$, $t\in\R^d$, and $\xi,\eta\in \R^r$, the diagonal sequence trick shows that there is a subsequence $\{ L_{m} \}\subset \{L\}$ through which $\varphi_\alpha^{(L_m)}(\xi ; \eta)$ is convergent for all $\alpha\in\Q^d$ and $\xi, \eta\in \Q^r$, yielding a limiting function $\varphi_\alpha(\xi; \eta)$.
Due to Lemma \ref{thm:bochoner_dshift}, this limiting function uniquely determines an operator $\hat{N}_\alpha$ that satisfies
\[ 
 \varphi_\alpha(\xi; \eta)=\left\langle e^{\i \xi^i X_i} \psi, \hat{N}_\alpha e^{\i \eta^i X_i} \psi \right\rangle.
\]
In this way, we obtain the WOT-limit
\begin{equation} \label{eq:gaus_minimax_proof_wot}
 \hat N_\alpha:=\lim_{m\to\infty} \hat{N}_\alpha^{(L_m)}
\end{equation}
for all $\alpha\in\Q^d$.
Further, for each $t\in\R^{d}$, let
\begin{equation}\label{eq:gaus_minimax_proof1}
 \bar{N}_{t}:=\inf_{\alpha>t, \alpha\in\Q^d}\hat{N}_{\alpha}.
\end{equation}
Then $\{\bar N_t\}_{t\in\R^d}$ determines a POVM $\bar N$ over $\bar{\R}^{d}$, 
and by transferring the measure at infinity $\bar{N}(\bar{\R}^{d}\setminus\R^{d})$ to the origin, we have a POVM $N$ over $\R^{d}$ defined by
\[
  N(B):=\bar{N}(B)+\delta_{0}(B)\bar{N}(\bar{\R}^{d}\setminus\R^{d})
  \qquad (B\in\B(\R^d)).
\]

Let us prove that $\phi_h(\bar{N}(B))=\phi_h(N(B))$ for all $B\in\B(\R^d)$.
For each $m\in\N$, let $\hat{N}^{(L_{m})}$ be the POVM that corresponds to the resolution of identity $\hat{N}_t^{(L_{m})}$, and let
\begin{align*}
 \mu_{h}^{(m)}(B)
 &:=\phi_h\left(\hat{N}^{(L_{m})}(B)\right) \\
 &=\frac{1}{(2L_m)^d}  \int_{[-L_{m},L_{m}]^{d}} \phi_h\left(U(k)^* M(B+k) U(k)\right) dk \\
 &=\frac{1}{(2L_m)^d}  \int_{[-L_{m},L_{m}]^{d}} \phi_{h+k}\left(M(B+k)\right) dk.
\end{align*}
Then, letting $y:=x+k$, 
\begin{align}\label{eq:gaus_minimax_proof_bound}
 & \int_{\R^{d}}G_{ij}(x-h)^{i}(x-h)^{j}\mu_{h}^{(m)}(dx) \\
 &\qquad =\frac{1}{(2L_m)^d} \int_{[-L_{m},L_{m}]^{d}}dk
 	\int_{\R^{d}+k} G_{ij}(y-h-k)^{i}(y-h-k)^{j}\phi_{h+k}\left(M(dy)\right)\nonumber \\
 &\qquad \le \sup_{\ell \in\R^{d}}\int_{\R^{d}}G_{ij}(y-\ell)^{i}(y-\ell)^{j} \phi_{\ell} \left(M(dy)\right). \nonumber
\end{align}
This shows that the second moments of $\{ \mu_{h}^{(m)}\} _{m\in\N}$ are uniformly bounded. 
As a consequence, $\{ \mu_{h}^{(m)}\} _{m\in\N}$ is uniformly tight, and by Prohorov's lemma, there exists a subsequence $\{m_{s}\}\subset\{m\}$ and a probability measure $\check{\mu}_{h}$ that satisfy $\mu_{h}^{(m_{s})}\conv{}\check{\mu}_{h}$.
We show that
\begin{equation} \label{eq:gaus_proof_bar_tilde}
 \check{\mu}_{h}(B)=\phi_{h}(\bar{N}(B))=\phi_{h}(N(B))
\end{equation}
for all $B\in\B(\R^d)$. 
Actually, since $\check{\mu}_{h}$ is a probability measure on $\R^d$, having no positive mass at infinity,
it suffices to prove that $\check{\mu}_{h}(-\infty,t]=\phi_{h}(\bar{N}_{t})$ for all continuity point $t\in\R^d$ of $t\mapsto\check{\mu}_{h}(-\infty,t]$. 
For any $\alpha\in\Q^{d}$ satisfying $\alpha>t$,
\begin{align*}
\check{\mu}_{h}(-\infty,t] & =\lim_{s\to\infty}\phi_{h}\left(\hat{N}_{t}^{(m_{s})}\right)
  \leq\lim_{s\to\infty}\phi_{h}\left(\hat{N}_{\alpha}^{(m_{s})}\right)
 \leq\check{\mu}_{h}(-\infty,\alpha].
\end{align*}
In the last inequality, the portmanteau lemma is used. 
Taking the limit $\alpha \downarrow t$, and recalling the definition \eqref{eq:gaus_minimax_proof1} as well as
\[ \lim_{s\to\infty}\phi_{h}\left(\hat{N}_{\alpha}^{(m_{s})}\right)=\phi_{h}\left(\hat{N}_{\alpha}\right), \]
which follows from \eqref{eq:gaus_minimax_proof_wot}, 
we have $\check{\mu}_{h}(-\infty,t]=\phi_{h}(\bar{N}_{t})$. 

Now we proceed to the proof Theorem \ref{thm:q_gauss_minimax}. 
To this end, it suffices to show the following  (i) and (ii):
\begin{itemize}
\item[(i)] $N$ is equivalent in law.
\item[(ii)] $N$ satisfies the following inequality: 
\begin{align*} %\label{eq:gaus_minimax_proof2}
 &\sup_{h\in\R^{d}}\int_{\R^{d}}G_{ij}(x-h)^{i}(x-h)^{j}\phi_{h}\left(N(dx)\right) \\
 &\qquad \leq \sup_{h\in\R^{d}}\int_{\R^{d}}G_{ij}(x-h)^{i}(x-h)^{j}\phi_{h}(M(dx)). 
\end{align*}
\end{itemize}
In fact, suppose that (i) is true. 
Then Lemma \ref{lem:equi_law} tells us that the first line of the inequality in (ii) is further bounded from below by 
the Holevo bound $c^{(H)}_G$.
This is nothing but the desired minimax theorem.

Let us prove (ii) first.
From \eqref{eq:gaus_minimax_proof_bound}, we have
\begin{align*}
 &\sup_{h\in\R^{d}}\int_{\R^{d}}G_{ij}(x-h)^{i}(x-h)^{j}\phi_{h}(M(dx)) \\
 &\qquad \geq\liminf_{s\to\infty}\int_{\R^{d}}G_{ij}(x-h)^{i}(x-h)^{j} \mu_h^{(m_s)}(dx) \\
 &\qquad \geq\int_{\R^{d}}G_{ij}(x-h)^{i}(x-h)^{j}\check{\mu}_{h}(dx)\\
 &\qquad =\int_{\R^{d}}G_{ij}(x-h)^{i}(x-h)^{j}\phi_{h}\left(N(dx)\right).
\end{align*}
Here, the second inequality follows from the portmanteau lemma,
and the last equality from \eqref{eq:gaus_proof_bar_tilde}. 
Since the first line is independent of $h$, we have (ii). 

We next prove (i), that is, 
\begin{equation} \label{eq:gaus_minimax_proof4}
 \phi_{h}\left(N_{t+h}\right)=\phi_{0}\left(N_{t}\right)
\end{equation}
for all $t\in\R^{d}$ and $h\in\R^{d}$.
Since
\[
 \phi_{h}\left(N_{t+h}\right)
 =\phi_{h}\left(\bar{N}_{t+h}\right)
 =\inf_{\alpha>t,\alpha\in\Q^{d}}\phi_{h}\left(\hat{N}_{\alpha+h}\right),
\]
it suffice to prove  
\begin{equation}\label{eq:gaus_minimax_proof5}
 \phi_{h}\left(\hat{N}_{\alpha+h}\right)=\phi_{0}\left(\hat{N}_{\alpha}\right)
\end{equation}
for all $\alpha\in\Q^{d}$ and $h\in\R^{d}$. 
The left-hand side is rewritten as
\begin{align*}
\phi_{h}\left(\hat{N}_{\alpha+h}\right) 
 & =\lim_{m\to\infty}\phi_{h}\left(\hat{N}_{\alpha+h}^{(L_{m})}\right)\\
 & =\lim_{m\to\infty}\frac{1}{(2L_{m})^{d}}\int_{[-L_{m},L_{m}]^{d}} 
 	\phi_{h}\left(U(k)^*M_{\alpha+h+k}U(k)\right) dk\\
 & =\lim_{m\to\infty}\frac{1}{(2L_{m})^{d}}\int_{[-L_{m},L_{m}]^{d}} 
 	\phi_{h+k}\left(M_{\alpha+h+k}\right) dk \\
 & =\lim_{m\to\infty}\frac{1}{(2L_{m})^{d}}\int_{[-L_{m},L_{m}]^{d}+h} 
 	\phi_{\ell}\left(M_{\alpha+\ell}\right) d\ell. 
\end{align*}
Since $|h^i|<2L_m$ ($i=1,\dots,d$) for sufficiently large $m$, 
\begin{align*}
\left|\phi_{h}\left(\hat{N}_{\alpha+h}\right)-\phi_{0}\left(\hat{N}_{\alpha}\right)\right| & \leq\lim_{m\to\infty}\frac{1}{(2L_{m})^{d}}\int_{\R^{d}}\ind_{\left([-L_{m},L_{m}]^{d}+h\right)\triangle\left([-L_{m},L_{m}]^{d}\right)}(k)dk\,\\
 & =\lim_{m\to\infty}\frac{1}{(2L_{m})^{d}} \times 
 	2\left\{ \left(2L_{m}\right)^{d}-\prod_{i=1}^{d}(2L_{m}-| h^{i} | )\right\} \\
 & =\lim_{m\to\infty}2\left\{ 1-\prod_{i=1}^{d}(1-\frac{|h^{i}|}{2L_{m}})\right\} =0. 
\end{align*}
Here, $\triangle$ denotes the symmetric difference. 
This proves \eqref{eq:gaus_minimax_proof5}. 
\end{proof}

\subsection{Proof of Theorem \ref{thm:minimax_local}}

\begin{proof}
For notational simplicity, we denote by $\mu_{h}^{(n)}$ the probability measure of outcomes of POVM $M^{(n)}$ applied to 
$\rho_{\theta_{0}+h/\sqrt{n}}^{(n)}$. 
The first inequality immediately follows from the fact that 
for any $\delta>0$ and finite subset $H$ of $\R^{d}$, there exist $N\in\N$ so that $n\ge N$ implies 
\[
\sup_{\left\Vert h\right\Vert \leq\delta\sqrt{n}}\int_{\R^{d}}G_{ij}(x-h)^{i}(x-h)^{j}\mu_{h}^{(n)}(dx)
\geq\sup_{h\in H}\int_{\R^{d}}G_{ij}(x-h)^{i}(x-h)^{j}\mu_{h}^{(n)}(dx). 
\]
The second inequality is obvious. 
We prove the last inequality.

Following the proof of \cite[Theorem 8.11]{vaart}, place the elements of $\Q^{d}$ in an arbitrary order, 
and let $H_k$ consist of the first $k$ elements in this sequence. 
Let
\[
c_{k}^{n}:=\sup_{h\in H_{k}}\int_{\R^{d}} k\wedge\left\{ G_{ij}(x-h)^{i}(x-h)^{j}\right\}\mu_{h}^{(n)}(dx), 
\]
and let
\[
 c_{k}:=\liminf_{n\to\infty}c_{k}^{n}
 \quad\mbox{and}\quad
 c:=\lim_{k\to\infty}c_{k}.
\]
Since $c$ is not greater than the third line of \eqref{eq:minimax1}, it suffices to show that
$c\geq c_{\theta_{0}}^{(rep)}$.
Since the inequality is trivial when $c=\infty$, we assume that $c<\infty$.

Take a subsequence $\{n_{k}\}\subset\{n\}$ that satisfies
\[
\lim_{k\to\infty}c_{k}^{n_{k}}=c.
\]
In fact, just choose $n_{k}$ so that $n_{k}>n_{k-1}$ and
\[
\left|c_{k}^{n_{k}}-c_{k}\right|<1/k
\]
hold for all $k\in\N$.
Let us prove that $\{ \mu_{h}^{(n_{k})}\} _{k}$ is uniformly tight for all $h\in\Q^{d}$.

Suppose that $\{ \mu_{h}^{(n_{k})}\} _{k}$ is not uniformly tight for some $h\in\Q^d$. 
For this $h$, let
\[
K_{L}:=\left\{ x\in\R^{d}:G_{ij}(x-h)^{i}(x-h)^{j}\leq L\right\},\qquad (L>0). 
\]
Then there exists an $\varepsilon>0$ such that 
\[
\limsup_{k\to\infty}\mu^{(n_{k})}(K_{L}^{c})\geq\varepsilon
\]
for all $L>0$. 
Since $h\in H_k$ for sufficiently large $k$, it holds that
\begin{align*}
c & =\lim_{k\to\infty}c_{k}^{n_{k}}\\
 & \geq\limsup_{k\to\infty}\int_{\R^{d}} k \wedge \left\{ G_{ij}(x-h)^{i}(x-h)^{j}\right\} \mu_{h}^{(n_{k})}(dx)\\
 & \geq L\cdot \limsup_{k\to\infty} \mu_{h}^{(n_{k})}(K_{L}^{c}) \\
 & \geq L\cdot \varepsilon.
\end{align*}
Since $L>0$ is arbitrary, this contradicts the assumption that $c<\infty$. 

Now that $\{ \mu_{h}^{(n_{k})}\} _{k}$ is proved uniformly tight for all $h\in\Q^d$, 
by the Prohorov lemma and the diagonal sequence trick, we can take a further subsequence $\{k_{s}\}\subset\{k\}$ that satisfies 
\[
\mu_{h}^{(n_{k_s})}\conv{} \exists \mu_{h}
\]
for all $h\in\Q^{d}$. 
It then follows from the asymptotic representation theorem for $h\in\Q^{d}$ that there is a POVM $M^{(\infty)}$ on $N(\left({\rm Re}\,\tau\right)h,\Sigma)\sim\left(X,\phi_{h}\right)$
such that 
\[
\phi_{h}\left(M^{(\infty)}(B)\right)=\mu_{h}(B), \qquad (\forall B\in\B(\R^d),\;\forall h\in\Q^d). 
\]

Now, for any $h\in\Q^{d}$ and $L>0$, 
\begin{align*}
c & =\lim_{s\to\infty}c_{k_{s}}^{n_{k_s}}\\
 & \geq \liminf_{s\to\infty}
 	\int_{\R^{d}} L\wedge \left\{ G_{ij}(x-h)^{i}(x-h)^{j}\right\}\mu_{h}^{(n_{k_s})}(dx)\\
 & =\int_{\R^{d}} L \wedge \left\{ G_{ij}(x-h)^{i}(x-h)^{j}\right\} \mu_{h}(dx).
\end{align*}
In the last equality, we used the portmanteau lemma.
Thus
\begin{align*}
c & \geq\sup_{L>0}\sup_{h\in\Q^{d}}\int_{\R^{d}} L \wedge \left\{ G_{ij}(x-h)^{i}(x-h)^{j}\right\} \phi_{h}\left(M^{(\infty)}(dx)\right) \\ %\label{eq:minimax_proof3}\\
 & =\sup_{L>0}\sup_{h\in\R^{d}}\int_{\R^{d}} L \wedge \left\{ G_{ij}(x-h)^{i}(x-h)^{j}\right\}\phi_{h}\left(M^{(\infty)}(dx)\right) \\ %\label{eq:minimax_proof5}\\
 & =\sup_{h\in\R^{d}}\int_{\R^{d}}G_{ij}(x-h)^{i}(x-h)^{j}\phi_{h}\left(M^{(\infty)}(dx)\right) \\
 	%\label{eq:minimax_proof6}\\
 & \geq c_{G}^{(rep)}.  %\label{eq:minimax_proof7}
\end{align*}
Here, the second line follows from the fact that $h\mapsto\phi_{h}(A)$ is continuous for all $A\in {\rm CCR}({\rm Im}\Sigma)$ satisfying $\left\Vert A\right\Vert \leq 1$,
the third line is due to the monotone convergence theorem, 
and the last line follows from the minimax Theorem \ref{thm:q_gauss_minimax} for a quantum Gaussian shift model, as well as the fact that the asymptotic representation bound $c_{G}^{(rep)}$ is nothing but the Holevo bound for the quantum Gaussian shift model 
$\{ N(({\rm Re}\,\tau)h,\Sigma):h\in\R^{d}\}$. 

Finally, we prove that the last inequality of \eqref{eq:minimax1} is tight. 
Recall that the sequence $M_\star^{(n)}$ of POVMs constructed in the proof of Theorem \ref{thm:achieve} satisfies 
\[
\left(M_\star^{(n)},\,\rho_{\theta_{0}+h/\sqrt{n}}^{(n)}\right)\conv{}N(h,V_\star)\qquad (\forall h\in\R^{d})
\]
and
\[
\Tr GV_\star=c_{G}^{(rep)}. 
\]
We show that this sequence $M_\star^{(n)}$ saturates the last inequality in \eqref{eq:minimax1}. 
Let $p_{h}$ be the probability density of the classical Gaussian shift model $N(h,V_\star)$.  Then
\begin{align*}
 & \sup_{L>0}\sup_{H} \liminf_{n\to\infty}\sup_{h\in H}\int_{\R^{d}} L \wedge \left\{ G_{ij}(x-h)^{i}(x-h)^{j}\right\}
 \Tr\rho_{\theta_{0}+h/\sqrt{n}}^{(n)}M_\star^{(n)}(dx)\\
 &\qquad =\sup_{L>0}\sup_{H} \sup_{h\in H}\int_{\R^{d}} L \wedge \left\{ G_{ij}(x-h)^{i}(x-h)^{j}\right\} p_{h}(x)dx\\
 &\qquad =\Tr GV_\star.
\end{align*}
Here, the first equality follows from the portmanteau lemma and the fact that $H$ is a finite set.
The proof is complete. 
\end{proof}


\begin{thebibliography}{99}

\bibitem{Ando}
Ando.~T.  (1978).
\textit{Topics on operator inequalities},
Hokkaido Univ. Lecture Note.

\bibitem{Fujiwara:2006}
Fujiwara,~A. (2006).
Strong consistency and asymptotic efficiency for adaptive quantum estimation problems,
\textit{J. Phys. A: Math. Gen.,} \textbf{39}, 12489--12504;
Fujiwara,~A. (2011) Corrigendum,
\textit{J. Phys. A: Math. Theor.,} \textbf{44}, 079501. 

\bibitem{qcontiguity}
Fujiwara,~A. and Yamagata,~K. (2020).
Noncommutative Lebesgue decomposition and contiguity with application to quantum local asymptotic normality,
\textit{Bernoulli,} 26 (3), 2105-2142.

%\bibitem{qrep_supp}
%{Fujiwara,~A.} and {Yamagata,~K.} (2020).
%Supplementary material to ``Efficiency of estimators for locally asymptotically normal quantum statistical models.''

\bibitem{GillGuta}
Gill,~R.~D. and Gu\c{t}\u{a},~M.  (2013). 
%\textsc{Gill,~R.~D.} and \textsc{Gu\c{t}\u{a},~M.} (2013).
On asymptotic quantum statistical inference. In From Probability to Statistics and Back: 
High-Dimensional Models and Processes, 
volume 9 of IMS Collections, pages 105--127. Inst. Math. Statist., Beachwood, OH.

\bibitem{GillMassar}
{Gill,~R.~D.} and {Massar,~S.} (2000).
State estimation for large ensembles,
\textit{Phys. Rev. A,} \textbf{61}, 042312.

\bibitem{GutaJencova}
{Gu\c{t}\u{a},~M.} and {Jen\v{c}ov\'a,~A.} (2007).
Local asymptotic normality in quantum statistics. 
\textit{Comm. Math. Phys.,} \textbf{276}, 341--379.


%\bibitem[\protect\citeauthoryear{Gu\c{t}\u{a},~M. and Kahn,~J.}{2006}]{guta_qubit}
\bibitem{guta_qubit}
{Gu\c{t}\u{a},~M.} and {Kahn,~J.} (2006).
Local asymptotic normality for qubit states,
\textit{Phys. Rev. A,} \textbf{73}, 052108.

\bibitem{HajekMinimax}
{H\'ajek,~J.} (1972).
Local asymptotic minimax and admissibility in estimation,
{Proceedings of the Sixth Berekeley Symposium on Mathematical Statistics and Probability}, \textbf{1}, 
175--194.

\bibitem{Hiai:2021}
{Hiai,~F.}  (2021).
\textit{Lectures on Selected Topics in von Neumann Algebras}, 
EMS Press, Berlin. 

\bibitem{holevo}
{Holevo,~A.~S.} (1982).
\textit{Probabilistic and Statistical Aspects of Quantum Theory}, 
North-Holland, Amsterdam.

\bibitem{Husimi}
{Husimi,~K.} (1940).
Some formal properties of the density matrix, 
\textit{Proc. Phys. Math. Soc. Japan}, \textbf{22}, 264-314.

\bibitem{qLevyCramer}
{Jak\v{s}i\'c,~V.}, {Pautrat,~Y.}, and {Pillet,~C.-A.} (2010).
A non-commutative L\'evy-Cram\'er continuity theorem, 
\textit{Markov Process and Related Fields}, \textbf{16},  59--78.

\bibitem{qLevyCramer2}
{Jak\v{s}i\'c,~V.}, {Pautrat,~Y.}, and {Pillet,~C.-A.}  (2010).
A quantum central limit theorem for sums of independent identically distributed random variables, 
\textit{J. Math. Phys}., \textbf{51}, 015208.

\bibitem{JudgeBock}
{Judge,~G~.G.} and {Bock,~M.~E.} (1978).
\textit{The statistical implications of pre-test and Stein-rule estimators in econometrics},
North Holland.

\bibitem{KadisonRingrose_v2}
{Kadison,~R.~V.} and {Ringrose,~J.~R.}  (1997).
\textit{Fundamentals of the Theory of Operator Algebras, Volume II: Advanced Theory},
Amer. Math. Soc.

\bibitem{guta_qudit}
{Kahn, J.} and {Gu\c{t}\u{a}, M.} (2009). 
Local asymptotic normality for finite dimensional quantum systems. 
\textit{Comm. Math. Phys.}
289 597--652. MR2506764 https://doi.org/10.1007/s00220-009-0787-3.

\bibitem{kubo}
{Kubo,~F.} and {Ando,~T.}  (1980).
Means of positive linear operators, 
\textit{Math. Ann.,} \textbf{246}, 205--224.

\bibitem{low_rank}
{Lahiry, S.} and  {Nussbaum, M.} (2021).
Local asymptotic normality and optimal estimation of low-rank quantum systems.
arXiv:2111.03279 

\bibitem{ReedSimon}
{Reed,~M.} and {Simon,~B.} (1980).
 \textit{Methods of Mathematical Physics I: Functional Analysis}, Academic Press.

\bibitem{vaart}
{Van der Vaart, A.W.}
(1998). Asymptotic Statistics. 
Cambridge Series in Statistical and Probabilistic Mathematics 3.
Cambridge: Cambridge Univ. Press. MR1652247 https://doi.org/10.1017/CBO9780511802256.

\bibitem{YamagataTomo}
{Yamagata, K.} (2011).
Efficiency of quantum state tomography for qubits.
\textit{Int. J. Quant. Inform.,} \textbf{9}, 1167.

\bibitem{qlan_first}
{Yamagata,~K.}, {Fujiwara,~A.}, and {Gill,~R.~D.}  (2013).
Quantum local asymptotic normality based on a new quantum likelihood ratio,
\textit{Ann. Statist.}, \textbf{41}, 2197--2217. 



\end{thebibliography}
\end{document}